%% file: main.tex
\documentclass[authorversion,nonacm,10pt,screen]{acmart}
\usepackage{graphicx} 
\usepackage{extarrows}
\usepackage{pict2e}
\usepackage{tikz}
\usepackage{framed}
\usepackage{url}
\usepackage[normalem]{ulem}
\usepackage{hyperref}
\AtBeginDocument{
  }

\setcopyright{rightsretained}
\begin{document}
\include{defs}
\title{Braids, Twists, Trace and Duality in Combinatory Algebras}
\author{Masahito Hasegawa}
\orcid{0000-0003-3460-8615}
\affiliation{%
  \institution{Research Institute for Mathematical Sciences, Kyoto University}
  \city{Kyoto}
  \country{Japan}
}

\author{Serge Lechenne}
\orcid{0009-0001-9626-4742}
\affiliation{%
  \institution{\'{E}cole Normale Sup\'{e}riere Paris-Saclay}
  \city{Gif-sur-Yvette}
  \country{France}
}
\affiliation{%
  \institution{National Institute for Informatics}
  \city{Tokyo}
  \country{Japan}
}

\renewcommand{\shortauthors}{Hasegawa and Lechenne}
\authorsaddresses{}

\begin{abstract}
We investigate a class of combinatory algebras, called {\em ribbon combinatory algebras}, in which we can interpret both the braided untyped linear lambda calculus
and framed oriented tangles.  Any reflexive object in a ribbon category gives rise to
a ribbon combinatory algebra. Conversely, From a ribbon combinatory algebra, we
can construct a ribbon category with a reflexive object, from which the combinatory
algebra can be recovered. To show this, and also to give the equational characterisation of ribbon combinatory algebras, we make use of the internal PRO
construction developed in Hasegawa's recent work. 
Interestingly, we can characterise ribbon combinatory algebras in two different ways: as balanced combinatory algebras
with a trace combinator, and as balanced combinatory algebras with duality.
\end{abstract}

\maketitle

\renewcommand{\thefootnote}{\fnsymbol{footnote}}
\footnote[0]{
\copyright 2024 Copyright held by the authors. 
This is the authors' version of the work to appear in Proc. 39th Annual ACM/IEEE Symposium on Logic in Computer Science (LICS '24).
The definitive version will be available from \url{https://doi.org/10.1145/3661814.3662098}.}
\renewcommand{\thefootnote}{\arabic{footnote}}

\section{Introduction}

Since the 1980s, there has been a huge amount of research focusing on the
{\em graphical} formulations of the lambda calculus and related systems,
both from theoretical and practical sides. Major themes include
term graphs and graph reductions \cite{Plu99}, proof nets \cite{Gir87}, and Geometry of Interaction \cite{HS11}.  
Most of them are {\em combinatorial}, in the sense that they concern the graph-theoretic
formalism and techniques. On the other hand, the {\em geometric} perspective 
has been less studied. For instance, one could ask: how can we implement such graphical structures in the three-dimensional space? Do they form meaningful geometric entities?

As an example, let us consider the combinator $\CC=\lambda fxy.f\,y\,x$,
that can be represented, under the rooted trivalent graph interpretation developed in \cite{ZG15,Zei16,Zei18} as
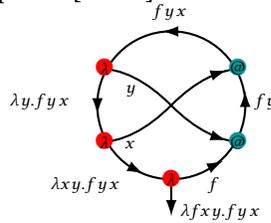
\begin{center}
\unitlength=.7pt
\begin{picture}(100,108)
\thicklines
\put(50,60){\circle{80}}
\symm{10}{40}{20}{20}
\put(50,20){\vector(0,-1){20}}
\put(45,100){\vector(-1,0){0}}
\put(10,55){\vector(0,-1){0}}
\put(90,68){\vector(0,1){0}}
\put(35,23){\vector(2,-1){0}}
\put(78,30){\vector(2,1){0}}
\put(80,41){\vector(3,-1){0}}
\put(80,79){\vector(3,1){0}}
\put(14,80){\lamnode}
\put(14,40){\lamnode}
\put(86,80){\appnode}
\put(86,40){\appnode}
\put(50,20){\lamnode}
\put(55,0){\tiny$\lambda fxy.fy\,x$}
\put(-15,15){\tiny$\lambda xy.fy\,x$}
\put(-37,60){\tiny$\lambda y.fy\,x$}
\put(40,108){\tiny$fy\,x$}
\put(95,60){\tiny$fy$}
\put(70,15){\tiny$f$}
\put(25,37){\tiny$x$}
\put(25,67){\tiny$y$}
\end{picture}    
\end{center}
where the lambda abstraction and application are expressed by the nodes \makebox(7,5)[c]{$\lamnode$} and \makebox(7,5)[c]{$\appnode$} respectively
(see the pictures below left), 
which obey the $\beta$-rule (below right).
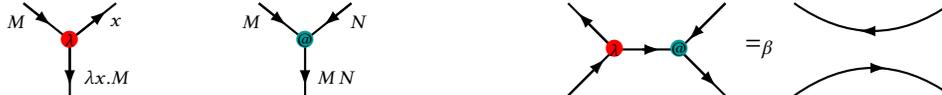
\begin{center}
\unitlength=.7pt
\begin{picture}(70,50)(-10,0)
\thicklines
\lineseg{0}{50}{25}{30}
\lineseg{50}{50}{25}{30}
\lineseg{25}{30}{25}{0}
\put(25,30){\vector(0,-1){25}}
\put(25,30){\vector(5,4){20}}
\put(0,50){\vector(5,-4){15}}
\put(25,30){\lamnode}
\put(-4,40){\makebox(0,0){\scriptsize$M$}}
\put(50,40){\makebox(0,0){\scriptsize$x$}}
\put(45,10){\makebox(0,0){\scriptsize$\lambda x.M$}}
\end{picture} 
\begin{picture}(50,50)
\end{picture}
\begin{picture}(70,50)(-10,0)
\thicklines
\lineseg{0}{50}{25}{30}
\lineseg{50}{50}{25}{30}
\lineseg{25}{30}{25}{0}
\put(25,30){\vector(0,-1){25}}
\put(50,50){\vector(-5,-4){15}}
\put(0,50){\vector(5,-4){15}}
\put(25,30){\appnode}
\put(-4,40){\makebox(0,0){\scriptsize$M$}}
\put(54,40){\makebox(0,0){\scriptsize$N$}}
\put(42,10){\makebox(0,0){\scriptsize$M\,N$}}
\end{picture} 
\begin{picture}(100,50)
\end{picture}
\begin{picture}(85,50)
\thicklines
\lineseg{0}{0}{25}{25}
\lineseg{0}{50}{25}{25}
\lineseg{25}{25}{60}{25}
\lineseg{60}{25}{85}{50}
\lineseg{60}{25}{85}{0}
\put(25,25){\vector(-1,1){20}}
\put(0,0){\vector(1,1){20}}
\put(40,25){\vector(1,0){10}}
\put(60,25){\vector(1,-1){20}}
\put(85,50){\vector(-1,-1){20}}
\put(25,25){\lamnode}
\put(60,25){\appnode}
\end{picture}
\begin{picture}(30,50)
\put(15,25){\makebox(0,0){$=_\beta$}}
\end{picture}
\begin{picture}(85,50)
\thicklines
\qbezier(0,50)(42.5,20)(85,50)
\qbezier(0,0)(42.5,30)(85,0)
\put(35,35){\vector(-1,0){0}}
\put(50,15){\vector(1,0){0}}
\end{picture}
\end{center}
In the three-dimensional space, the graph of $\CC$ can have several implementations which are topologically different. For example:
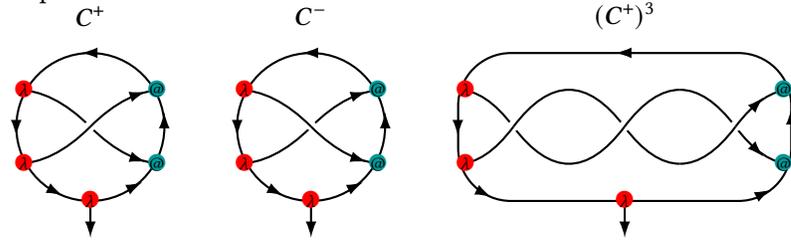
\begin{center}
\unitlength=.7pt
\begin{picture}(100,120)
\thicklines
\put(50,60){\circle{80}}
\braid{10}{40}{20}{20}
\put(50,20){\vector(0,-1){20}}
\put(45,100){\vector(-1,0){0}}
\put(10,55){\vector(0,-1){0}}
\put(90,68){\vector(0,1){0}}
\put(35,23){\vector(2,-1){0}}
\put(78,30){\vector(2,1){0}}
\put(80,41){\vector(3,-1){0}}
\put(80,79){\vector(3,1){0}}
\put(14,80){\lamnode}
\put(14,40){\lamnode}
\put(86,80){\appnode}
\put(86,40){\appnode}
\put(50,20){\lamnode}  
\put(50,120){\makebox(0,0){$\CC^+$}}
\end{picture} 
\hspace{3mm}
\begin{picture}(100,100)
\thicklines
\put(50,60){\circle{80}}
\braidInv{10}{40}{20}{20}
\put(50,20){\vector(0,-1){20}}
\put(45,100){\vector(-1,0){0}}
\put(10,55){\vector(0,-1){0}}
\put(90,68){\vector(0,1){0}}
\put(35,23){\vector(2,-1){0}}
\put(78,30){\vector(2,1){0}}
\put(80,41){\vector(3,-1){0}}
\put(80,79){\vector(3,1){0}}
\put(14,80){\lamnode}
\put(14,40){\lamnode}
\put(86,80){\appnode}
\put(86,40){\appnode}
\put(50,20){\lamnode}  
\put(50,120){\makebox(0,0){$\CC^-$}}
\end{picture}    
\hspace{3mm}
\begin{picture}(200,100)
\thicklines
\put(100,60){\oval(180,80)}
\braid{10}{40}{15}{20}
\braid{70}{40}{15}{20}
\braid{130}{40}{15}{20}
\put(100,20){\vector(0,-1){20}}
\put(95,100){\vector(-1,0){0}}
\put(10,55){\vector(0,-1){0}}
\put(190,70){\vector(0,1){0}}
\put(28,22){\vector(2,-1){0}}
\put(179,26){\vector(2,1){0}}
\put(175,44){\vector(1,-1){0}}
\put(175,76){\vector(1,1){0}}
\put(14,80){\lamnode}
\put(14,40){\lamnode}
\put(186,80){\appnode}
\put(186,40){\appnode}
\put(100,20){\lamnode} 
\put(100,120){\makebox(0,0){$(\CC^+)^3$}}
\end{picture} 
\end{center}
This three dimensional nature of $\CC$ was formalised by Hasegawa \cite{Has20}
as braided terms of a new lambda calculus, called a {\em braided lambda calculus}. In that work,
variables are realised as wires, and permutations/exchanges of variables are
realised by appropriate {\em braids} \cite{Art25,KT08}. Two terms are identified modulo
the $\beta\eta$-equality which includes continuous deformations of such wires.
Thus terms are regarded as entities of low-dimensional topology. 

In the subsequent work \cite{Has22}, Hasegawa developed a general approach to combinatory algebras,
where it is shown that any planar, linear, braided and classical combinatory algebra validating the $\beta\eta$-theory of the corresponding 
lambda calculus gives rise to an {\em internally defined} monoidal category equipped with a reflexive object, and the combinatory algebra
induced by the reflexive object is isomorphic to the original combinatory algebra. The key idea is the notion of {\em arity}, which determines the
structure of the monoidal category in a canonical way. Now combinators get a geometric reading via the string diagram of the internally 
defined monoidal categories, which subsumes the case of the braided lambda calculus.
Furthermore, as noted in \cite{Has22}, 
we could have a non-standard combinator $\Tr$ 
with the following graph,
which does not correspond to the usual lambda term:
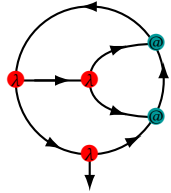
\begin{center}
\unitlength=.7pt
\begin{picture}(100,100)
\thicklines
\put(50,60){\circle{80}}
\lineseg{10}{60}{50}{60}
\arcLR{90}{40}{-40}{20}
\put(50,20){\vector(0,-1){20}}
\put(45,100){\vector(-1,0){0}}
\put(20,60){\vector(1,0){20}}
\put(90,70){\vector(0,1){0}}
\put(35,23){\vector(2,-1){0}}
\put(78,30){\vector(2,1){0}}
\put(70,41){\vector(3,-1){0}}
\put(70,79){\vector(3,1){0}}
\put(10,60){\lamnode}
\put(50,60){\lamnode}
\put(86,80){\appnode}
\put(86,40){\appnode}
\put(50,20){\lamnode}
\end{picture}    
\end{center}
$\Tr$, if combined with braided terms, can be used for creating {\em knots}:
$\beta$-reductions on $\Tr\,(\Tr\,(\CC^+)^3)$ produce a trefoil knot.
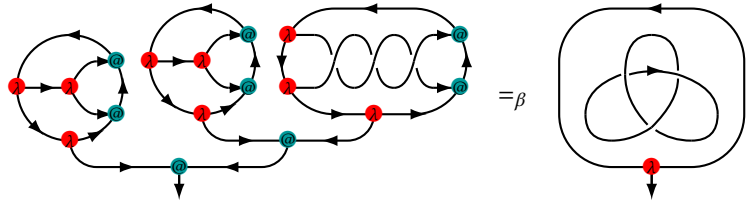
\newsavebox{\boxTr}
\savebox{\boxTr}{%
\unitlength=.5pt
\begin{picture}(80,80)(10,20)\thicklines
\put(50,60){\circle{80}}
\qbezier(10,60)(10,60)(50,60)
\qbezier(50,60)(60,40)(70,40)\qbezier(70,40)(85,40)(85,40)
\qbezier(50,60)(60,80)(70,80)\qbezier(70,80)(85,80)(85,80)
\put(40,60){\vector(1,0){0}}
\put(83,80){\vector(1,0){0}}
\put(83,40){\vector(1,0){0}}
\put(45,100){\vector(-1,0){0}}
\put(90,68){\vector(0,1){0}}
\put(30,24){\vector(1,-1){0}}
\put(76,28){\vector(2,1){0}}
\put(50,20){\lamnode}
\put(50,60){\lamnode}
\put(10,60){\lamnode}
\put(85,80){\appnode}
\put(85,40){\appnode}
\end{picture}
}%
\newsavebox{\boxbraid}
\savebox{\boxbraid}{%
\unitlength=.5pt
\begin{picture}(30,40)\thicklines
\qbezier(0,0)(10,0)(15,20)\qbezier(15,20)(20,40)(30,40)
\qbezier(0,40)(10,40)(13,25)\qbezier(17,15)(20,0)(30,0)
\end{picture}
}%
\begin{center}
\unitlength=.5pt
\begin{picture}(340,140)
\thicklines
\put(270,100){\oval(140,80)}
\qbezier(205,120)(205,120)(225,120)
\qbezier(205,80)(205,80)(225,80)
\put(225,80){\usebox{\boxbraid}}
\put(255,80){\usebox{\boxbraid}}
\put(285,80){\usebox{\boxbraid}}
\qbezier(315,120)(335,120)(335,120)
\qbezier(315,80)(335,80)(335,80)
\put(205,60){\oval(130,40)[b]}
\put(122.5,40){\oval(165,40)[b]}
\put(122.5,20){\line(0,-1){20}}
\put(0,40){\usebox{\boxTr}}
\put(100,60){\usebox{\boxTr}}
\put(90,20){\vector(1,0){0}}
\put(155,20){\vector(-1,0){0}}
\put(180,40){\vector(1,0){0}}
\put(230,40){\vector(-1,0){0}}
\put(250,60){\vector(1,0){0}}
\put(310,60){\vector(1,0){0}}
\put(330,120){\vector(1,0){0}}
\put(330,80){\vector(1,0){0}}
\put(260,140){\vector(-1,0){0}}
\put(340,108){\vector(0,1){0}}
\put(200,92){\vector(0,-1){0}}
\put(122.5,20){\vector(0,-1){25}}
\put(270,60){\lamnode}
\put(205,120){\lamnode}
\put(205,80){\lamnode}
\put(335,120){\appnode}
\put(335,80){\appnode}
\put(205,40){\appnode}
\put(122.5,20){\appnode}
\end{picture}
\begin{picture}(60,140)
\put(30,70){\makebox(0,0){$=_\beta$}}
\end{picture}
\begin{picture}(140,140)
\thicklines
\put(70,80){\oval(140,120)}
\put(70,20){\line(0,-1){20}}
\qbezier(70,120)(50,120)(50,90)\qbezier(50,90)(50,70)(67,52)
\qbezier(73,48)(85,40)(100,40)\qbezier(100,40)(120,40)(120,60)
\qbezier(120,60)(120,80)(90,90)\qbezier(90,90)(70,95)(53,90)
\qbezier(47,90)(20,80)(20,60)\qbezier(20,60)(20,40)(40,40)
\qbezier(40,40)(55,40)(70,50)\qbezier(70,50)(90,70)(90,87)
\qbezier(90,93)(90,120)(70,120)
\put(70,20){\vector(0,-1){25}}
\put(65,140){\vector(-1,0){0}}
\put(79,93){\vector(1,0){0}}
\put(70,20){\lamnode}
\end{picture}
\end{center}
Here, the application of the non-standard combinator $\Tr$ implements a (partial) braid closure, or more generally an abstract {\em trace} \cite{JSV96}.

Following these results, in this paper, we make a further step towards the geometric reading of terms / combinators by investigating 
the structure of {\em ribbon categories} in combinatory algebras. Ribbon categories \cite{Tur94}, also known as tortile categories \cite{Shu94}, 
are monoidal categories with braids, twists and duality which can model the framed tangles (ribbons) \cite{Shu94,Yet01}, just in the same way that cartesian closed
categories can model the simply typed lambda calculus \cite{LS86}. In \cite{Has22} Hasegawa considered braids and symmetry. In this paper, we deal with
twist, trace and duality. 

As a result, we get a situation in which tangles and lambda terms/combinators can be freely mixed, which has not been studied before.
Our results are theoretical and we do not claim that they have some practical application (at least immediately). Nevertheless,
we think that this approach exhibits the less-known geometric perspectives of combinatory logic which deserve further study, and they
might lead to
some practical applications (e.g., low-level physical implementation models such as topological quantum computation \cite{Kit03,Wan10}). 
Also, we shall point out that some of the key ideas were already present in Curry's pioneering work \cite{Cur30a,Cur30b}; it seems that Curry had a
clear geometric reading of combinators in mind, which we believe is not very far from ours. 

In addition, in the presence of braids, twists and trace/duality, we find that combinatory algebras offer a much simpler formalism 
than lambda calculi in which we cannot avoid quite complicated substitution mechanisms as observed in \cite{Has20}. 

\paragraph{Organisation}
The rest of this paper is organised as follows.
In Section \ref{sec:prelim}, we recall monoidal categories and related structures to be used throughout this work.
In Section \ref{sec:intPRO}, we study how monoidal categories can be derived from combinatory algebras, following Hasegawa's approach. Section \ref{sec:ribbon} is devoted to the study of reflexive objects in ribbon categories.
Using the results of Section \ref{sec:intPRO} and \ref{sec:ribbon}, we introduce ribbon combinatory algebras in Section \ref{sec:ribbonCA}.
In Section \ref{sec:comparison} we compare a ribbon category with a reflexive object and the internal PROB of the ribbon combinatory algebra 
derived from the reflexive object. In Section \ref{sec:related} we discuss related work, including Curry's classical work.
We conclude this paper with a summary and some future work in Section \ref{sec:concl}.

\section{Monoidal categories}\label{sec:prelim}

\paragraph{Notations.}
In most cases, we write $f;g$ for the composition of arrows $f\colon A\rightarrow B$
and $g\colon B\rightarrow C$ in a category. On the other hand, we write 
$a\circ b$ for the composition $\BB\,a\,b$ in combinatory algebras and their internal PROs (Section \ref{sec:intPRO}).
Note that, in an internal PRO, $a\colon l\rightarrow m$ and $b\colon m\rightarrow n$ imply
$a\circ b\colon l\rightarrow n$
(Lemma \ref{lem:compositionality}), which is $a;b$ in the other cases and is the opposite of 
the conventional use of the symbol $\circ$.

\subsection{Monoidal categories}\label{subsec:monoidal}

A {\em monoidal category} (also called {\em tensor category}) \cite{Mac71,JS93}
is a category $\mathbb{C}$ equipped with a functor
$\otimes\colon {\mathbb{C}}\times{\mathbb{C}}\rightarrow{\mathbb{C}}$ (tensor product),
an object $I\in{\mathbb{C}}$ (tensor unit) and natural 
isomorphisms
$a_{A,B,C}\colon 
(A\otimes B)\otimes C\stackrel{\sim}{\rightarrow}A\otimes(B\otimes C)$,
$l_A\colon I\otimes A\stackrel{\sim}{\rightarrow}A$ and
$r_A\colon A\otimes I\stackrel{\sim}{\rightarrow}A$
subject to the standard coherence axioms. 
It is said to be strict if $a,l,r$ are  
identity morphisms. Whenever convenient, we will pretend as if our monoidal categories are strict;
the coherence theorem \cite{Mac71,JS93} ensures that nothing goes wrong in doing so.

In the sequel, we will make use of the {\em string diagrams}: the graphical presentation of 
morphisms in monoidal categories \cite{JS91,Sel11}. A morphism
$  f\colon A_1\otimes A_2\otimes\dots\otimes A_m\rightarrow
B_1\otimes B_2\otimes\dots\otimes B_n$ in a monoidal category
will be drawn as (to be read from left to right):
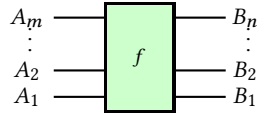
\begin{center}
\unitlength=0.5pt{\footnotesize \thicklines
\begin{picture}(130,80)(0,0)
\rgbfbox{40}{0}{50}{80}{.8}{1}{.8}{\scriptsize$f$}%
\idline{0}{70}{40} \put(-20,70){\makebox(0,0){\small$A_m$}}
\idline{0}{30}{40} \put(-20,30){\makebox(0,0){\small$A_2$}} \put(-20,45){$\vdots$}
\idline{0}{10}{40} \put(-20,10){\makebox(0,0){\small$A_1$}}
\idline{90}{70}{40} \put(145,70){\makebox(0,0){\small$B_n$}}
\idline{90}{30}{40} \put(145,30){\makebox(0,0){\small$B_2$}} \put(145,45){$\vdots$}
\idline{90}{10}{40} \put(145,10){\makebox(0,0){\small$B_1$}}
\end{picture}}
\end{center}
The composition and tensor of morphisms are depicted as follows.
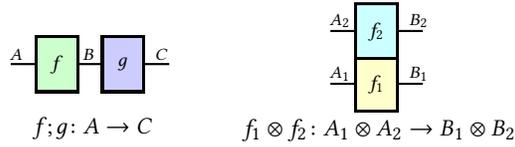
\begin{center}
\unitlength=.5pt\footnotesize
\begin{picture}(120,80)(200,-30)
\thicklines
\idline{200}{20}{20} \put(200,23){\tiny$A$}
\rgbfbox{220}{0}{30}{40}{.8}{1}{.8}{\scriptsize$f$}%
\idline{250}{20}{20} \put(255,23){\tiny$B$}
\rgbfbox{270}{0}{30}{40}{.8}{.8}{1}{\scriptsize$g$}%
\idline{300}{20}{20} \put(310,23){\tiny$C$}
\put(260,-28){\makebox(0,0){ \small$f;g\colon A\rightarrow C$}}
\end{picture}
\hspace{20mm}
\begin{picture}(70,100)(150,-10)
\thicklines
\idline{150}{25}{20} \put(150,30){\tiny$A_1$}
\rgbfbox{170}{5}{30}{40}{1}{1}{.8}{\scriptsize$f_1$}%
\idline{200}{25}{20} \put(210,30){\tiny$B_1$}
\idline{150}{65}{20} \put(150,70){\tiny$A_2$}
\rgbfbox{170}{45}{30}{40}{.8}{1}{1}{\scriptsize$f_2$}%
\idline{200}{65}{20} \put(210,70){\tiny$B_2$}
\put(185,-10){\makebox(0,0){ \small$f_1\otimes f_2\colon A_1\otimes A_2\rightarrow B_1\otimes B_2$}}
\end{picture}
\end{center}

\paragraph{PROs}
A {\em PRO}, meaning a "PROduct category", is a strict monoidal category generated by 
a single object. We identify objects of a PRO with natural numbers $\mathbb{N}$; the tensor product of objects $m$ and $n$ 
is the addition $m+n$, while the tensor unit $I$ is $0$. Thus a PRO $\mathbb{C}$ can be defined
as an $\mathbb{N}^2$-indexed family $\mathbb{C}(m,n)$ with the following constructs
\begin{itemize}
\item the identity $\id_n\in\mathbb{C}(n,n)$
\item the composition $f;g\in\mathbb{C}(l,n)$ for $f\in\mathbb{C}(l,m)$, $g\in\mathbb{C}(m,n)$
\item the tensor $f_1+f_2\in\mathbb{C}(m_1+m_2,n_1+n_2)$ for $f_i\in\mathbb{C}(m_i,n_i)$ 
\end{itemize}
which are subject to the following axioms.
$$
\begin{array}{lr}
\multicolumn{2}{l}{\id_m;f=f=f;\id_n\hspace{2mm}(f\in\mathbb{C}(m,n))}\\
\multicolumn{2}{l}{(f;g);h = f;(g;h)\hspace{2mm}(f\in\mathbb{C}(l,m),~g\in\mathbb{C}(m,n),~h\in\mathbb{C}(n,p))}\\
f+\id_0 = f = \id_0+f \\ 
(f_1+f_2)+f_3=f_1+(f_2+f_3) \\ 
\id_m+\id_n = \id_{m+n}\\
(f_1+f_2);(g_1+g_2) = (f_1;g_1)+(f_2;g_2) & (f_i\in\mathbb{C}(l_i,m_i),~g_i\in\mathbb{C}(m_i,n_i))
\end{array}
$$
From these axioms, we have the {\em interchange law}
$$
(f+\id_{m'});(\id_n+g)=f+g=(\id_m+g);(f+\id_{n'})
$$
for $f\in\mathbb{C}(m,n)$ and $g\in\mathbf{C}(m',n')$.
Indeed, instead of $f+g$ for arrows $f$ and $g$, we may use $f+m$ $(=f+\id_m)$ and $m+f$ $(=\id_m+f)$ as primitive constructs, 
with slightly simplified axioms $f+0=f=0+f$, $(m+f)+m'=m+(f+m')$, $(m+m')+f=m+(m'+f)$, $(f+m)+m'=f+(m+m')$
and the interchange law 
$(f+m');(n+g)=(m+g);(f+n')$
for defining PROs.
$f+g$ is then defined either as $(f+m');(n+g)$ or $(m+g);(f+n')$.
This alternative characterisation will be useful in developing PROs in combinatory algebras.

\subsection{Braids and twists}\label{subsec:braids}

A {\em braid} \cite{JS93} in a monoidal category is a natural isomorphism 
$ \sigma_{A,B}\colon A\otimes B\stackrel{\cong}{\rightarrow} B\otimes A$
such that both $\sigma$ and $\sigma^{-1}$ 
satisfy the following ``bilinearity'' or ``Hexagon Axiom"
(the case for $\sigma^{-1}$ is omitted):
\begin{center}
\unitlength=1.2pt
\begin{picture}(200,70)(0,5)
\put(50,60){\vector(1,0){25}} \put(62.5,65){\makebox(0,0){\footnotesize$a_{A,B,C}$}}
\put(125,60){\vector(1,0){25}} \put(137.5,65){\makebox(0,0){\footnotesize$\sigma_{A,B\otimes C}$}}
\put(50,20){\vector(1,0){25}} \put(62.5,15){\makebox(0,0){\footnotesize$a_{B,A,C}$}}
\put(125,20){\vector(1,0){25}} \put(137.5,15){\makebox(0,0){\footnotesize$B\otimes \sigma_{A,C}$}}
\put(25,52.5){\vector(0,-1){25}} \put(9,40){\makebox(0,0){\footnotesize$\sigma_{A,B}\otimes C$}}
\put(175,52.5){\vector(0,-1){25}} \put(187,40){\makebox(0,0){\footnotesize$a_{B,C,A}$}}
\put(25,60){\makebox(0,0){$(A\otimes B)\otimes C$}}
\put(100,60){\makebox(0,0){$A\otimes (B\otimes C)$}}
\put(175,60){\makebox(0,0){$(B\otimes C)\otimes A$}}
\put(25,20){\makebox(0,0){$(B\otimes A)\otimes C$}}
\put(100,20){\makebox(0,0){$B\otimes (A\otimes C)$}}
\put(175,20){\makebox(0,0){$B\otimes (C\otimes A)$}}
\end{picture}
\end{center}
A {\em braided monoidal category} is a monoidal category equipped with a braid.

A {\em twist} or {\em balance} \cite{JS93,Shu94} for a braided monoidal category is a natural isomorphism
$\theta_A\colon A\stackrel{\cong}{\rightarrow}A$ satisfying 
$\theta_{A\otimes B}=\sigma_{A,B};(\theta_B\otimes\theta_A);\sigma_{B,A}$.
A {\em balanced monoidal category} is a braided monoidal category with a twist.

The braids and twists will be depicted in the string diagrams (oriented from left to right)
as follows (cf. \cite{Has09}).
\unitlength=.6pt
$$
\begin{array}{cccccc}
\sigma_{X,Y}
&
\sigma_{X,Y}^{-1}
&
\theta_X
&
\theta_X^{-1}
\\
\\
\begin{picture}(60,20)(-10,0)
\thicklines
\braid{10}{0}{5}{10}
\put(5,0){\makebox(0,0){\scriptsize$X$}}
\put(5,20){\makebox(0,0){\scriptsize$Y$}}
\put(35,0){\makebox(0,0){\scriptsize$Y$}}
\put(35,20){\makebox(0,0){\scriptsize$X$}}
\end{picture}
&
\begin{picture}(60,20)(-10,0)
\thicklines
\braidInv{10}{0}{5}{10}
\put(5,0){\makebox(0,0){\scriptsize$Y$}}
\put(5,20){\makebox(0,0){\scriptsize$X$}}
\put(35,0){\makebox(0,0){\scriptsize$X$}}
\put(35,20){\makebox(0,0){\scriptsize$Y$}}
\end{picture}
&
\begin{picture}(60,20)(-10,0)
\thicklines
\twist{10}{10}
\put(5,10){\makebox(0,0){\scriptsize$X$}}
\put(35,10){\makebox(0,0){\scriptsize$X$}}
\end{picture}
&
\begin{picture}(60,20)(-10,0)
\thicklines
\twistInv{10}{10}
\put(5,10){\makebox(0,0){\scriptsize$X$}}
\put(35,10){\makebox(0,0){\scriptsize$X$}}
\end{picture}
\\
\end{array}
$$
A braid $\sigma$ is called a {\em symmetry} when it satisfies $\sigma_{X,Y}^{-1}=\sigma_{Y,X}$ for all $X$ and $Y$,
which will be depicted as
\begin{picture}(60,20)(-10,0)
\thicklines
\symm{10}{0}{5}{10}
\put(5,0){\makebox(0,0){\scriptsize$X$}}
\put(5,20){\makebox(0,0){\scriptsize$Y$}}
\put(35,0){\makebox(0,0){\scriptsize$Y$}}
\put(35,20){\makebox(0,0){\scriptsize$X$}}
\end{picture}.
A {\em symmetric monoidal category} is a monoidal category equipped with a symmetry.
Note that a balanced monoidal category is symmetric exactly when its twist is the identity.

\paragraph{PROPs and PROBs}
A {\em PROP}, meaning a "PROduct and Permutation category", is a PRO equipped with a symmetry,
thus a strict symmetric monoidal category generated by 
a single object. In a PROP, its symmetry $\sigma$ is determined by its $(1,1)$-component $\sigma_{1,1}\colon 2\rightarrow 2$
as the Hexagon axiom implies $\sigma_{l,m+n}=(\sigma_{l,m}+n);(m+\sigma_{l,n})$ and 
$\sigma_{m+n,1}=(m+\sigma_{n,1});(\sigma_{m,1}+n)$. We may define a PROP as a PRO equipped with $\sigma\colon 2\rightarrow 2$
such that the {\em Coxter relations} $\sigma;\sigma=\id_2$ and $(1+\sigma);(\sigma+1);(1+\sigma)=(\sigma+1);(1+\sigma);(\sigma+1)$
\begin{center}
\begin{picture}(40,40)
\thicklines
\symm{0}{10}{5}{10}
\symm{20}{10}{5}{10}
\end{picture}
\begin{picture}(20,40)
\put(10,20){\makebox(0,0){$=$}}
\end{picture}
\begin{picture}(30,40)
\thicklines
\lineseg{0}{30}{30}{30}
\lineseg{0}{10}{30}{10}
\end{picture}
\hspace{15mm}
\begin{picture}(60,40)
\thicklines
\symm{0}{20}{5}{10}
\symm{20}{0}{5}{10}
\symm{40}{20}{5}{10}
\lineseg{0}{0}{20}{0}
\lineseg{20}{40}{40}{40}
\lineseg{40}{0}{60}{0}
\end{picture}
\begin{picture}(20,40)
\put(10,20){\makebox(0,0){$=$}}
\end{picture}
\begin{picture}(60,40)
\thicklines
\symm{0}{0}{5}{10}
\symm{20}{20}{5}{10}
\symm{40}{0}{5}{10}
\lineseg{0}{40}{20}{40}
\lineseg{20}{0}{40}{0}
\lineseg{40}{40}{60}{40}
\end{picture}
\end{center}
hold and $\sigma_{m,n}\colon m+n\rightarrow n+m$ induced by $\sigma=\sigma_{1,1}$ satisfies the naturality $(f+g);\sigma_{n,n'}=\sigma_{m,m'};(g+f)$. 

Similarly, a {\em PROB} is a PRO equipped with a braid, i.e., a strict braided monoidal category with a single generating object.
It can be defined as a PRO equipped with $\sigma^+,\sigma^-\colon 2\rightarrow 2$
subject to the Coxter relations (or {\em Reidemeister move II and III}) $\sigma^\pm;\sigma^\mp=\id_2$ and 
$(1+\sigma^\pm);(\sigma^\pm+1);(1+\sigma^\pm)=(\sigma^\pm+1);(1+\sigma^\pm);(\sigma^\pm+1)$
\begin{center}
\begin{picture}(40,40)
\thicklines
\braid{0}{10}{5}{10}
\braidInv{20}{10}{5}{10}
\end{picture}
\begin{picture}(20,40)
\put(10,20){\makebox(0,0){$=$}}
\end{picture}
\begin{picture}(30,40)
\thicklines
\lineseg{0}{30}{30}{30}
\lineseg{0}{10}{30}{10}
\end{picture}
\begin{picture}(20,40)
\put(10,20){\makebox(0,0){$=$}}
\end{picture}
\begin{picture}(40,40)
\thicklines
\braidInv{0}{10}{5}{10}
\braid{20}{10}{5}{10}
\end{picture}
\begin{picture}(40,40)
\put(10,20){}
\end{picture}
\begin{picture}(60,45)
\thicklines
\braid{0}{20}{5}{10}
\braid{20}{0}{5}{10}
\braid{40}{20}{5}{10}
\lineseg{0}{0}{20}{0}
\lineseg{20}{40}{40}{40}
\lineseg{40}{0}{60}{0}
\end{picture}
\begin{picture}(20,40)
\put(10,20){\makebox(0,0){$=$}}
\end{picture}
\begin{picture}(60,40)
\thicklines
\braid{0}{0}{5}{10}
\braid{20}{20}{5}{10}
\braid{40}{0}{5}{10}
\lineseg{0}{40}{20}{40}
\lineseg{20}{0}{40}{0}
\lineseg{40}{40}{60}{40}
\end{picture}
\hspace{5mm}
\begin{picture}(60,40)
\thicklines
\braidInv{0}{20}{5}{10}
\braidInv{20}{0}{5}{10}
\braidInv{40}{20}{5}{10}
\lineseg{0}{0}{20}{0}
\lineseg{20}{40}{40}{40}
\lineseg{40}{0}{60}{0}
\end{picture}
\begin{picture}(20,40)
\put(10,20){\makebox(0,0){$=$}}
\end{picture}
\begin{picture}(60,40)
\thicklines
\braidInv{0}{0}{5}{10}
\braidInv{20}{20}{5}{10}
\braidInv{40}{0}{5}{10}
\lineseg{0}{40}{20}{40}
\lineseg{20}{0}{40}{0}
\lineseg{40}{40}{60}{40}
\end{picture}
\end{center}
and the naturality of $\sigma_{m,n}\colon m+n\rightarrow n+m$ determined from $\sigma_{1,1}=\sigma^+$.

In addition, we can talk about {\em balanced PROB}: it is a PROB equipped with $\theta^+,\theta^-\colon 1\rightarrow 1$
satisfying $\theta^\pm;\theta^\mp=\id_1$ and the naturality of $\theta_m\colon m\rightarrow m$ determined by 
$\theta_0=\id_0$,
$\theta_1=\theta^+$ and
$\theta_{m+n}=\sigma_{m,n};(\theta_n+m);\sigma_{n,m};(\theta_m+n)$.

\subsection{Duality and trace}

\paragraph{Duality.}
In a monoidal category, a  {\em dual pairing} between two objects
$A$ and $B$ is a pair of morphisms $\eta\colon I\rightarrow A\otimes B$, called unit, 
and $\varepsilon\colon B\otimes A\rightarrow I$, called counit, 
depicted in the string diagrams
as follows 
\unitlength=.6pt
$$
\renewcommand{\arraystretch}{.5}
\begin{array}{cccccc}
\eta
&
\varepsilon
\\
\\
\begin{picture}(40,20)(5,0)
\thicklines
\arcLR{20}{0}{-10}{10}
\lineseg{20}{20}{30}{20}
\lineseg{20}{0}{30}{0}
\put(35,0){\makebox(0,0){\scriptsize$A$}}
\put(35,20){\makebox(0,0){\scriptsize$B$}}
\end{picture}
&
\begin{picture}(40,20)(-5,0)
\thicklines
\arcLR{20}{0}{10}{10}
\lineseg{20}{20}{10}{20}
\lineseg{20}{0}{10}{0}
\put(0,20){\makebox(0,0){\scriptsize$A$}}
\put(0,0){\makebox(0,0){\scriptsize$B$}}
\end{picture}
\\
\\
\end{array}
$$
which are subject to the snake equations:
\begin{center}
\unitlength=.6pt
\begin{picture}(60,50)
\thicklines
\arcLR{20}{0}{-10}{10}
\arcLR{40}{20}{10}{10}
\lineseg{0}{40}{40}{40}
\lineseg{20}{20}{40}{20}
\lineseg{20}{0}{60}{0}
\put(5,46){\makebox(0,0){\scriptsize$A$}}
\put(30,26){\makebox(0,0){\scriptsize$B$}}
\put(55,6){\makebox(0,0){\scriptsize$A$}}
\end{picture}   
\begin{picture}(40,50)
\put(20,20){\makebox(0,0){$=$}}
\end{picture}
\begin{picture}(40,50)
\thicklines
\lineseg{0}{20}{40}{20}
\put(5,26){\makebox(0,0){\scriptsize$A$}}
\put(35,26){\makebox(0,0){\scriptsize$A$}}
\end{picture}
\begin{picture}(80,50)
\end{picture}
\begin{picture}(60,50)
\thicklines
\arcLR{20}{20}{-10}{10}
\arcLR{40}{0}{10}{10}
\lineseg{0}{0}{40}{0}
\lineseg{20}{20}{40}{20}
\lineseg{20}{40}{60}{40}
\put(5,6){\makebox(0,0){\scriptsize$B$}}
\put(30,26){\makebox(0,0){\scriptsize$A$}}
\put(55,46){\makebox(0,0){\scriptsize$B$}}
\end{picture}   
\begin{picture}(40,50)
\put(20,20){\makebox(0,0){$=$}}
\end{picture}
\begin{picture}(40,50)
\thicklines
\lineseg{0}{20}{40}{20}
\put(5,26){\makebox(0,0){\scriptsize$B$}}
\put(35,26){\makebox(0,0){\scriptsize$B$}}
\end{picture}
\end{center}
In such a dual pairing, $B$ is called the {\em left dual} of $A$ while $A$ is the {\em right dual} of $B$.
For any object, its left (or right) dual, if exists, is unique up to isomorphism.
A {\em left autonomous category} is a monoidal category $\mathbb{C}$ in which
every object $X$ has a (chosen) left dual $X^*$ with unit $\eta_X\colon I\rightarrow X\otimes X^*$ and counit $\varepsilon_X\colon X^*\otimes X\rightarrow I$.
The map $X\mapsto X^*$ extends to a contravariant functor $(\_)^*\colon \mathbb{C}^\mathrm{op}\rightarrow\mathbb{C}$. 
A symmetric left autonomous category is called a {\em compact closed category} \cite{KL80}.

\paragraph{Ribbon categories.}

A {\em ribbon category} \cite{Tur94}, also called a {\em tortile monoidal category} \cite{Shu94}, is a balanced monoidal category which is left autonomous and satisfies $(\theta_X)^*=\theta_{X^*}$.
Ribbon categories are of great importance in low-dimensional topology and representation theory of quantum groups,
since the free ribbon category is equivalent to the category of (oriented framed) tangles \cite{Shu94,Yet01}, and many useful ribbon categories
are obtained as categories of representations of quantum groups \cite{Has12,Tur94,Yet01}.
Compact closed categories are ribbon categories whose braid is symmetric and whose twist is the identity.

\paragraph{Trace.}

A (right) {\em trace} on a balanced monoidal category $\mathbb{C}$ is a family of maps 
$\mathit{Tr}^X_{A,B}\colon \mathbb{C}(A\otimes X,B\otimes X)\rightarrow\mathbb{C}(A,B)$
subject to a few coherence axioms \cite{JSV96}. A {\em traced monoidal category} is a balanced
monoidal category equipped with a trace.

If $\mathbb{C}$ is a ribbon category,
it has a unique trace \cite{Has09,JSV96} given by
$$
\mathit{Tr}^X_{A,B}(f)=
(A\otimes\eta_A);
(f\otimes X^*);
(B\otimes(\sigma_{X,X'};(X^*\otimes\theta_X);\varepsilon_X)).
$$
\begin{center}
\unitlength=.7pt
\begin{picture}(70,55)(-5,0)
\thicklines
\lineseg{0}{5}{20}{5}
\lineseg{0}{25}{20}{25}
\lineseg{40}{5}{60}{5}
\lineseg{40}{25}{60}{25}
\rgbfbox{20}{0}{20}{30}{.5}{1}{.5}{\scriptsize$f$}%
\put(-5,25){\makebox(0,0){\scriptsize$X$}}
\put(-5,5){\makebox(0,0){\scriptsize$A$}}
\put(65,25){\makebox(0,0){\scriptsize$X$}}
\put(65,5){\makebox(0,0){\scriptsize$B$}}
\end{picture}
\begin{picture}(60,55)
\put(30,20){\makebox(0,0){$\mapsto$}}
\end{picture}
\begin{picture}(60,55)
\thicklines
\arcLR{10}{25}{-10}{10}
\lineseg{10}{45}{50}{45}
\arcLR{50}{25}{10}{10}
\lineseg{0}{5}{20}{5}
\lineseg{10}{25}{20}{25}
\lineseg{40}{5}{60}{5}
\lineseg{40}{25}{50}{25}
\rgbfbox{20}{0}{20}{30}{.5}{1}{.5}{\scriptsize$f$}%
\put(13,30){\makebox(0,0){\scriptsize$X$}}
\put(47,30){\makebox(0,0){\scriptsize$X$}}
\put(-5,5){\makebox(0,0){\scriptsize$A$}}
\put(65,5){\makebox(0,0){\scriptsize$B$}}
\end{picture}
\begin{picture}(40,55)
\put(20,20){\makebox(0,0){$=$}}
\end{picture}
\begin{picture}(100,55)
\thicklines
\arcLR{10}{25}{-10}{10}
\lineseg{10}{45}{50}{45}
\braid{50}{25}{5}{10}
\twist{70}{45}
\arcLR{90}{25}{10}{10}
\lineseg{0}{5}{20}{5}
\lineseg{10}{25}{20}{25}
\lineseg{70}{25}{90}{25}
\lineseg{40}{25}{50}{25}
\lineseg{40}{5}{100}{5}
\rgbfbox{20}{0}{20}{30}{.5}{1}{.5}{\scriptsize$f$}%
\put(10,47){\scriptsize$X^*$}
\put(10,27){\scriptsize$X$}
\put(42,47){\scriptsize$X^*$}
\put(42,27){\scriptsize$X$}
\put(65,47){\scriptsize$X$}
\put(75,27){\scriptsize$X^*$}
\put(87,47){\scriptsize$X$}
\put(-5,5){\makebox(0,0){\scriptsize$A$}}
\put(105,5){\makebox(0,0){\scriptsize$B$}}
\end{picture}
\end{center}
Hence a ribbon category is also a traced monoidal category.
While not all traced monoidal categories are ribbon,
any traced monoidal category can be embedded in 
a ribbon category \cite{JSV96}. Specifically, given a traced monoidal
category $\mathbb{C}$, one can construct a ribbon category $\mathbf{Int}\mathbb{C}$
with a fully faithful structure-preserving functor from $\mathbb{C}$ to $\mathbf{Int}\mathbb{C}$.
So traced monoidal categories can be characterised as the monoidal full subcategories of ribbon categories.

\subsection{Closed structure and reflexive objects}

We say a monoidal category $\mathbb{C}$ is {\em closed} when the functor $(\_)\otimes A\colon \mathbb{C}\rightarrow\mathbb{C}$ has a right adjoint
$A\multimap(\_)\colon \mathbb{C}\rightarrow\mathbb{C}$, hence $\mathbb{C}(X\otimes A,B)\cong\mathbb{C}(X,A\multimap B)$. 
In particular, left autonomous categories and ribbon categories are closed with $A\multimap B= B\otimes A^*$. 

In a 
monoidal category $\mathbb{C}$
(which does not have to be closed), we say an object $A$ is {\em (strictly) reflexive} when the
functor $\mathbb{C}(\_\otimes A,A)$ is represented by $A$ itself, thus there exists a natural isomorphism
$\mathsf{cur}\colon \mathbb{C}(\_\otimes A,A)\stackrel{\cong}{\rightarrow}\mathbb{C}(\_,A)$. When $\mathbb{C}$ is closed,
$A$ is a reflexive object if and only if $A$ is isomorphic to $A\multimap A$.
Using a reflexive object, we can model the abstraction and application in the untyped lambda calculus as follows.
For $f\colon X\otimes A\rightarrow A$, its currying $\mathsf{cur}(f)\colon X\rightarrow A$ amounts to the lambda abstraction. For the application,
we have $\mathsf{app}=\mathsf{cur}^{-1}(\mathit{id}_A)\colon A\otimes A\rightarrow A$. Then the $\beta$-rule  $(\mathsf{cur}(f)\otimes X);\mathsf{app}=f$
holds for $\mathsf{cur}^{-1}(\mathsf{cur}(f))=f$ and the naturality of $\mathsf{cur}$.
The $\eta$-rule is immediate as $\mathsf{cur}(\mathsf{cur}^{-1}(g))=g$.
Alternatively, a reflexive object can be defined as an object $A$ equipped with a morphism $\mathsf{app}\colon A\otimes A\rightarrow A$
such that, for any $f\colon X\otimes A\rightarrow A$, there exists a unique $\mathsf{cur}(f)\colon X\rightarrow A$ satisfying 
$(\mathsf{cur}(f)\otimes A);\mathsf{app}=f$. Note that the naturality $h;\mathsf{cur}(f)=\mathsf{cur}((h\otimes A);f)$ ($h\colon X\rightarrow Y$, $f\colon Y\otimes A\rightarrow A$)
follows from the uniqueness of the currying.

(In the literature, it is more common to define a reflexive object to be an object $A$ such that  $A\multimap A$ is a retract of $A$,
which can model the $\beta$-theory of the lambda calculus.
In this paper, we focus on the extensional (i.e., $\beta\eta$) cases, and by a reflexive object we mean a strict one as defined above.)

\paragraph{Reflexive PROs}

We are particularly interested in PROs in which $1$ is a reflexive object, which gives a neat setting to discuss the untyped lambda calculus
and combinatory algebras.
We call such a PRO a {\em reflexive PRO}. Thus a reflexive PRO is a PRO equipped with $\mathsf{app}\colon 2\rightarrow 1$ and
a unique $\mathsf{cur}(f)\colon n\rightarrow 1$ satisfying $(\mathsf{cur}(f)+1);\mathsf{app}=f$ for $f\colon n+1\rightarrow 1$. 
In particular, for any $f\colon n\rightarrow 1$, $\mathsf{cur}^n(f)\colon 0\rightarrow 1$ is a unique element such that
$(\mathsf{cur}^n(f)+n);(\mathsf{app}+(n-1));\dots;(\mathsf{app}+1);\mathsf{app}=f$ holds.

In this paper ribbon categories with a reflexive object will play a central role. 
For a concrete nontrivial (non-symmetric) example of such a ribbon category, we refer to \cite{Has20}.
A combinatory algebra arising from such an example will be described in Example \ref{ex:trees}. 

\section{Monoidal structures in combinatory algebras}\label{sec:intPRO}

Hasegawa \cite{Has22} has shown that a combinatory algebra arises from a reflexive object in a monoidal category if 
and only if it has a monoidal category with a reflexive object {internally defined} in the combinatory algebra itself. The basic idea 
follows from the standard practice  of encoding tuples of terms $M_1,\dots, M_n$ by $\lambda f.f\,M_1\,\dots\,M_n$ in the lambda calculus
\cite{Bar84,HS08}:
we use a combinator corresponding to the closed lambda term of the form 
$$\lambda fx_1\dots x_m.f\,M_1\,\dots M_n$$ 
as a
morphism with $m$ inputs and $n$ outputs, where the variable $f$, expressing the evaluation context or the continuation of the
process, is not free in $M_i$'s. Intuitively, such a combinator takes a continuation $f$ and inputs $x_1,\dots x_m$, 
and returns the outputs $M_1,\dots,M_n$ to the continuation.
We say that such a combinator $a$ is {\em of arity} $m\rightarrow n$ and express this as $a\colon m\rightarrow n$.
While not all combinators have arity (e.g., the combinator corresponding to $\lambda xy.y\,x$), 
several standard combinators do have arity, including 
$\BB=\lambda fxy.f\,(x\,y)\colon 2\rightarrow 1$, $\II=\lambda f.f\colon 0\rightarrow 0$, $\CC=\lambda fxy.f\,y\,x\colon 2\rightarrow 2$,
$\SS=\lambda fxy.f\,y\,(x\,y)\colon 2\rightarrow 2$, $\KK=\lambda fx.f\colon 1\rightarrow 0$, and $\WW=\lambda fx.f\,x\,x\colon 1\rightarrow 2$.
Also, for any combinator $a$ (which may not have an arity), $\lambda f.f\,a$ is of arity $0\rightarrow 1$.
We can depict these combinators with arity using string diagrams as follows:
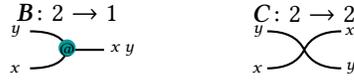
\begin{center}
\unitlength=.7pt
\begin{picture}(60,40)
\put(30,30){\makebox(0,0){$\BB\colon 2\rightarrow 1$}}
\thicklines
\qbezier(10,20)(28,20)(30,10)
\qbezier(10,0)(28,0)(30,10)
\myline{30}{10}{50}{10}
\rvec{17}{20}
\rvec{17}{0}
\rvec{45}{10}
\put(30,10){\appnode}
\put(2,0){\makebox(0,0){\tiny$x$}}
\put(2,20){\makebox(0,0){\tiny$y$}}
\put(60,10){\makebox(0,0){\tiny$x\,y$}}
\end{picture}
\hspace{15mm}
\begin{picture}(60,40)
\put(30,30){\makebox(0,0){$\CC\colon 2\rightarrow 2$}}
\thicklines
\qbezier(10,20)(25,20)(30,10)\qbezier(30,10)(35,0)(50,0)
\qbezier(10,0)(25,0)(30,10)\qbezier(30,10)(35,20)(50,20)
\rvec{50}{20}
\rvec{50}{0}
\put(5,0){\makebox(0,0){\tiny$x$}}
\put(5,20){\makebox(0,0){\tiny$y$}}
\put(55,0){\makebox(0,0){\tiny$y$}}
\put(55,20){\makebox(0,0){\tiny$x$}}
\end{picture}
\end{center}
The central findings of \cite{Has22} are
\begin{itemize}
\item {\bf algebraic characterisation of arities}:  a combinator $a$ with arity $m\rightarrow n$
can be characterised as the one satisfying an equation expressing that the following {\em interchange law} 
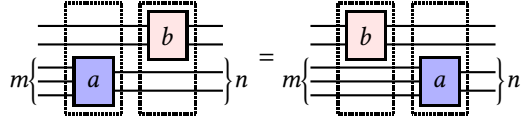
\begin{center}
\unitlength=.7pt
\begin{picture}(100,55)
\thicklines
\put(15,-5){\dashbox(30,60){}}
\put(55,-5){\dashbox(30,60){}}
\myline{0}{42.5}{100}{42.5}
\myline{0}{32.5}{100}{32.5}
\rvec{10}{42.5}
\rvec{10}{32.5}
\rvec{95}{42.5}
\rvec{95}{32.5}
\myline{0}{20}{20}{20}
\myline{0}{12.5}{20}{12.5}
\myline{0}{5}{20}{5}
\rvec{10}{20}
\rvec{10}{12.5}
\rvec{10}{5}
\myline{40}{17.5}{100}{17.5}
\myline{40}{7.5}{100}{7.5}
\rvec{95}{17.5}
\rvec{95}{7.5}
\rgbfbox{20}{0}{20}{25}{.7}{.7}{1}{$a$}%
\rgbfbox{60}{25}{20}{25}{1}{.9}{.9}{$b$}%
\put(-3,12.5){\makebox(0,0){$\Big\{$}}
\put(-10,12.5){\makebox(0,0){$m$}}
\put(103,12.5){\makebox(0,0){$\Big\}$}}
\put(110,12.5){\makebox(0,0){$n$}}
\end{picture}
\begin{picture}(40,50)
\put(20,25){\makebox(0,0){$=$}}
\end{picture}
\begin{picture}(100,50)
\thicklines
\put(15,-5){\dashbox(30,60){}}
\put(55,-5){\dashbox(30,60){}}
\myline{0}{42.5}{100}{42.5}
\myline{0}{32.5}{100}{32.5}
\rvec{10}{42.5}
\rvec{10}{32.5}
\rvec{95}{42.5}
\rvec{95}{32.5}
\myline{0}{20}{60}{20}
\myline{0}{12.5}{60}{12.5}
\myline{0}{5}{60}{5}
\rvec{10}{20}
\rvec{10}{12.5}
\rvec{10}{5}
\myline{80}{17.5}{100}{17.5}
\myline{80}{7.5}{100}{7.5}
\rvec{95}{17.5}
\rvec{95}{7.5}
\rgbfbox{60}{0}{20}{25}{.7}{.7}{1}{$a$}%
\rgbfbox{20}{25}{20}{25}{1}{.9}{.9}{$b$}%
\put(-3,12.5){\makebox(0,0){$\Big\{$}}
\put(-10,12.5){\makebox(0,0){$m$}}
\put(103,12.5){\makebox(0,0){$\Big\}$}}
\put(110,12.5){\makebox(0,0){$n$}}
\end{picture}
\end{center}
holds for any $b$;
\item {\bf internal PRO}: combinators with arity form a strict monoidal category generated by a single object $1$ (the {\em internal PRO});
\item {\bf reflexive object}: $1$ is a reflexive object in the internal PRO, with the application $\BB\colon 2\rightarrow 1$; and
\item {\bf internal representation}: the combinatory algebra derived from the reflexive object is isomorphic to the original combinatory algebra,
\end{itemize}
provided the combinatory algebra validates the $\beta\eta$-theory of the lambda calculus (i.e., they are Curry algebras).
They immediately imply a version of Scott's theorem \cite{Sco80}: any such combinatory algebra arises
 via a reflexive object in a monoidal category.
Also, we can systematically derive axioms for such combinatory algebras from these results: axioms directly follow from those of PROs
and reflexive objects.
Below we first recall the most basic case of the combinatory algebras corresponding to the {\em planar lambda calculus}; 
as shown in \cite{Has22}, cases with symmetry (the linear lambda calculus), braids (the braided lambda calculus), 
copying and discarding (the lambda calculus) are covered by adding appropriate constructs to the planar case.
In the remaining part of this section, We shall explain how we can introduce and axiomatise braids, symmetry and twists in combinatory algebras.
The main contribution of the present paper is then to show that this story can be carried out for the case with duality, i.e., ribbon categories,
which will be done in the following sections.

\subsection{Combinatory algebras}
Let us recall some basics on applicative structures and combinatory algebras
 \cite{Bar84,HS08,Sel02} which will be used in the rest of this paper. 
 A {\em applicative structure} $(\calA,\cdot)$ is a set $\calA$ equipped with
a binary map $(\_)\cdot(\_):\calA\times \calA\rightarrow \calA$ called {\em application}. 
As is customary, applications are assumed to be left associative, and the infix $\cdot$ is often omitted.
Thus $a\,b\,c$ means $(a\cdot b)\cdot c$ while $a\,(b\,c)$ is $a\cdot(b\cdot c)$.
In the literature, a {\em combinatory algebra} often refers to an {\em $\SK$-algebra}, which is an applicative structure 
equipped with elements $\SS$ and $\KK$ satisfying $\SS\,a\,b\,c=(a\,c)\,(b\,c)$ and $\KK\,a\,b=a$.
However, combinatory algebras can also be characterised more conceptually by the {\em combinatory completeness} \cite{Cur30a,Cur30b}.
For an applicative structure $\calA$, let $\calA[x_1,\dots,x_n]$ be the set of expressions (polynomials)
generated by elements of $\calA$
and indeterminates $x_1,\dots,x_n$ and application $\cdot$. 
We say $\calA$ is combinatory complete if, for any $p\in\calA[x_1,\dots,x_n,x]$, there exists
$\lambda^*x.p\in\calA[x_1,\dots,x_n]$ such that $(\lambda^*x.p)\cdot x=p$ holds; and then a combinatory algebra can be defined
as a combinatory complete applicative structure. It is well-known that these two definitions of combinatory algebras coincide:
combinatory completeness implies existence of $\SS$ and $\KK$, while it is standard to define $\lambda^*x.p$ using $\SS$ and $\KK$
\cite{Bar84,HS08}. 

In this paper, we will deal with more general classes of combinatory algebras than those recalled above. 
By focusing on combinatory completeness concerning {\em linear} polynomials (in which each indeterminate is used exactly once),
we have the linear variant of combinatory algebras, which are {\em $\BCI$-algebras} with elements $\BB$, $\CC$ and $\II$
satisfying $\BB\,a\,b\,c=a\,(b\,c)$, $\CC\,a\,b\,c=a\,c\,b$ and $\II\,a=a$. Similarly, by further restricting our attention to
{\em planar} polynomials (in which each indeterminate is used exactly once in the prescribed order), 
we have a planar variant of combinatory algebras, which will be recalled below.
Hasegawa \cite{Has22} argues that the notion of operads (one-object multicategories) gives a uniform foundation handling 
these variations and more (the braided case in particular).

\subsection{Planar combinatory algebras and their internal PROs}

The {\em planar lambda calculus} \cite{Zei16,Tom21,Has22} is an untyped linear lambda calculus in which no exchange of variables is allowed.
Its terms are generated by the following rules for variables, lambda abstractions and applications:
$$
\frac{}{x\vdash x}
\hspace{15mm}
\frac{\Gamma,x\vdash M}{\Gamma\vdash \lambda x.M}
\hspace{15mm}
\frac{\Gamma\vdash M \hspace{3mm}\Gamma'\vdash N}{\Gamma,\Gamma'\vdash M\,N}
$$
Planar terms are closed under $\beta\eta$-conversions. The most important planar terms include
$\BB=\lambda fxy.f\,(x\,y)$, $\II=\lambda f.f$ and $M^\bullet=\lambda f.f\,M$ for closed planar $M$.
On the other hand, $\CC=\lambda fxy.f\,y\,x$ is not planar.

Tomita \cite{Tom21} introduced {\em $\BIdot$-algebras} as the combinatory algebras corresponding to the planar lambda calculus.

\begin{definition}\cite{Tom21}
A {\em $\BIdot$-algebra} is an applicative structure $\calA$ with elements $\BB$, $\II$ and $a^\bullet$ for $a\in\calA$, satisfying
$\BB\,a\,b\,c=a\,(b\,c)$, 
$\II\,a=a$ and 
$a^\bullet\,b=b\,a$.
\end{definition}
Tomita has shown that a $\BIdot$-algebras is {\em planar combinatory complete}:
for a planar term $p$ with its rightmost variable $x$, there is $\lambda^*x.p$ satisfying $(\lambda^*x.p)\,x=p$.
$\lambda^*x.p$ which can be given by
$$
\begin{array}{rcll}
\lambda^*x.x &=& \II\\
\lambda^*x.p\,q &=& \BB\,q^\bullet\,(\lambda^*x.p) & (x~\mathrm{free~in}~p)\\
\lambda^*x.p\,q &=& \BB\,p\,(\lambda^*x.q) & (x~\mathrm{free~in}~q).
\end{array}
$$
$\BIdot$-algebras are central in Tomita's theory of {\em non-symmetric realizability}: the assemblies of a
$\BIdot$-algebra form a {\em closed multicategory} \cite{Tom21} which models the implicational fragment of the
non-commutative multiplicative intuitionistic linear logic.

\paragraph{Extensionality axioms and the internal PRO} 

As is well-known, combinatory completeness alone is not sufficient for interpreting the $\beta\eta$-theory of the
lambda calculus in combinatory algebras, and the same is the case for the planar lambda calculus and $\BIdot$-algebras.
Hasegawa \cite{Has22} identified additional {\em extensionality} axioms for $\BIdot$-algebras making them sound and complete for 
the planar lambda calculus with $\beta\eta$-theory. 
(This notion of extensionality differs from the standard terminology \cite{Bar84,HS08} where an applicative structure 
$A$ is said to be extensional when $a\,x=b\,x$ implies $a=b$ for all {\em elements $x$ of $\calA$}.
Our extensionality agrees with the notion of Curry-algebras; see \cite{Has22,Sel02} for relevant discussions.)

\begin{definition} A $\BIdot$-algebra is {\em extensional} when it satisfies the following axioms,
where $a\circ b$ is an abbreviation for $\BB\,a\,b$. 
\setlength{\FrameSep}{4pt}
\samepage
\begin{oframed}
$$
\begin{array}{rcll}
\BB\,\II &=& \II & (BI)\\
(a\,b)^\bullet &=& b^\bullet\circ(a^\bullet\circ\BB) & (\bullet\mathit{app})\\
\BB^\bullet\circ (\BB\circ(\BB\circ\BB)) &=& (\BB\,\BB)\circ\BB & (\mathit{arity}_\BB)\\
\II^\bullet \circ \BB &=& \BB\,\II & (\mathit{arity}_\II)\\
a^{\bullet\bullet}\circ\BB &=& (\BB\,a^\bullet)\circ\BB & (\mathit{arity}_{a^\bullet})\\
\end{array}
$$
\end{oframed}
\end{definition}

It is routine to check that reflexive objects in monoidal categories give extensional $\BIdot$-algebras:
\begin{proposition}\label{prop:refl:planar}
Suppose that $A$ is a reflexive object in a monoidal category $\mathbb{C}$.
Then $\mathbb{C}(I,A)$ is an extensional $\BIdot$-algebra with 
$a\cdot b=(a\otimes b);\mathsf{app}$,
$\BB=\mathsf{cur}(\mathsf{cur}(\mathsf{cur}((A\otimes\mathsf{app});\mathsf{app})))$, $\II=\mathsf{cur}(\id_A)$ and
$a^\bullet=\mathsf{cur}((A\otimes a);\mathsf{app})$.
\end{proposition}
We will soon see that all extensional $\BIdot$-algebras arise in this way (Cor. \ref{cor:Scott:planar}).

\begin{lemma} \label{lem:assoc}
In an extensional $\BIdot$-algebra, $(a\circ b)\circ c=a\circ(b\circ c)$ and $a\circ\II=a=\II\circ a$ hold.
Moreover, $\BB\,\II=\II$ and $\BB\,(a\circ b)=(\BB\,a)\circ(\BB\,b)$ hold.
\end{lemma}
\begin{proof}
$\BB\,(a\circ b)=((\BB\,\BB)\circ\BB)\,a\,b=
(\BB^\bullet\circ(\BB\circ(\BB\circ \BB)))\,a\,b=(\BB\,a)\circ(\BB\,b)$.
$(a\circ b)\circ c = \BB\,(a\circ b)\,c=((\BB\,a)\circ(\BB\,b))\,c=a\circ(b\circ c)$.    
\end{proof}
For an element $a$ of a $\BIdot$-algebra, let $a^0=\II$ and $a^{n+1}=a\circ a^n$.

\begin{definition}\label{def:arity}
An element $a$ of an extensional $\BIdot$-algebra is said to be of arity $m\rightarrow n$
(notation: $a\colon m\rightarrow n$)
when it satisfies 
$$
a^\bullet\circ\BB^{m+1}\,=\,(\BB\,a)\circ\BB^n \hspace{1cm} (\mathit{arity}_a).
$$
\end{definition}
In the combinatory algebra of closed terms of the $\lambda_{\beta\eta}$-calculus, a term is of arity $m\rightarrow n$
if and only if 
it equals a term of the form $\lambda fx_1\dots x_m.f\,M_1\dots M_n$ with no free $f$ in $M_i$'s \cite{Has22}.

\begin{lemma} \label{lem:compositionality}
In an extensional $\BIdot$-algebra, the following results hold. 
\begin{enumerate}
\item For $a\colon m\rightarrow n$ and $b$, $(\BB^m\,b)\circ a = a\circ (\BB^n\,b)$ holds.
\item 
$\BB\colon 2\rightarrow 1$, 
$\II\colon 0\rightarrow0$
and
$a^\bullet\colon 0\rightarrow 1$ for any $a\in\calA$.
\item 
If $a\colon m\rightarrow n$, then
$\BB\,a\colon m+1\rightarrow n+1$. 
\item
If $a\colon l\rightarrow m$
and $b\colon m\rightarrow n$, then
$a\circ b\colon l\rightarrow n$. 
\item 
If $a\colon m\rightarrow n$, then
$a\colon m+1\rightarrow n+1$.

\item $a\colon m\rightarrow 1$ iff $a=(a\,\II)^\bullet\circ\BB^m$.
\end{enumerate}
\end{lemma}
\begin{proof}
\begin{enumerate}
\item $(\BB^m\,b)\circ a = (a^\bullet\circ \BB^{m+1})\,b= ((\BB\,a)\circ\BB^n)\,b= a\circ (\BB^n\,b)$.
\item Immediate from the axioms of extensionality.
\item $(\BB\,a)^\bullet\circ\BB^{m+2}=a^\bullet\circ\BB^\bullet\circ\BB^{m+3}=
a^\bullet\circ(\BB\,\BB)\circ\BB^{m+1}=\BB\circ a^\bullet\circ\BB^{m+1}=\BB\circ(\BB\,a)\circ\BB^n=(\BB\,a)\circ\BB^{n+1}$.
\item $(a\circ b)^\bullet\circ\BB^{l+1}=b^\bullet\circ(\BB\,a)^\bullet\circ\BB^{l+2}=
b^\bullet\circ(\BB\,(\BB\,a))\circ\BB^{m+1}
= (\BB\,a)\circ b^\bullet\circ\BB^{m+1}
= (\BB\,a)\circ(\BB\,b)\circ\BB^n
= (\BB\,(a\circ b))\circ\BB^n$.
\item Immediate.
\item If $a\colon m\rightarrow 1$, then 
$(a\,\II)^\bullet\circ \BB^m=\II^\bullet\circ a^\bullet\circ\BB^{m+1}=\II^\bullet\circ(\BB\,a)\circ\BB=a\circ\II^\bullet\circ\BB=a\circ\II=a$.
Conversely, if $a=(a\,\II)^\bullet\circ\BB^m$, $a\colon m\rightarrow 1$ because $(a\,\II)^\bullet\colon m\rightarrow m+1$ and $\BB^m\colon m+1\rightarrow 1$.
\end{enumerate}
\end{proof}
From Lemma \ref{lem:assoc} and \ref{lem:compositionality} and discussions on PRO in Section \ref{subsec:monoidal} we have:
\begin{proposition}
Given an extensional $\BIdot$-algebra $\calA$, the following data determine a PRO $\IA$.
\begin{itemize}
\item $\IA(m,n)=\{a\in\calA~|~a\colon m\rightarrow n\}$.
\item For $a\colon l\rightarrow m$ and $b\colon m\rightarrow n$, $a;b=a\circ b\colon l\rightarrow n$.
\item $\id_m=\II\colon m\rightarrow m$.
\item For $a\colon m\rightarrow n$, $l+a=\BB^l\,a\colon l+m\rightarrow l+n$.
\item For $a\colon m\rightarrow n$, $a+l=a\colon m+l\rightarrow n+l$.
\end{itemize}
\end{proposition}
Note the asymmetry between $n+(\_)$ and $(\_)+n$; the former is given by applying $\BB$ $n$-times, while the latter 
we do not need to do anything because of the extensionality. 
To help with the intuition, it might be useful to look at the case of the closed term model of the $\lambda_{\beta\eta}$-calculus
where we have,
for $a=\lambda fx_1\dots x_m.f\,M_1\dots M_n\colon m\rightarrow n$, 
$$
\begin{array}{rll}
{\red 1\,+\,}a
&=&
\lambda f{\red x_0}x_1\dots x_m.f\,{\red x_0}\,M_1\dots M_n\\
&=_\beta&
(\lambda fxy.f\,(x\,y))\,(\lambda fx_1\dots x_m.f\,M_1\dots M_n)\\
&=&\BB\,a\\
\\
a{\red \,+\,1}
&=&
\lambda fx_1\dots x_m{\red x_{m+1}}.f\,M_1\dots M_n\,{\red x_{m+1}}\\
&=\eta&
\lambda fx_1\dots x_m.f\,M_1\dots M_n\\
&=&
a
\end{array}
$$
This difference between left-tensor and right-tensor in internal PROs cannot be ignored even 
in the braided or symmetric cases. In particular, 
we will later discuss {\em left} traces on internal PROBs which are much easier to describe than right traces because of this asymmetry.

Another possible source of confusion would be that $\BB$ plays multiple important roles in this story: it gives the composition $a\circ b=\BB\,a\,b$,
the left-tensor $1+a=\BB\,a$, as well as an element of arity $2\rightarrow1$ serving as the application. 
This complexity is already present in the calculus of $\BB$ studied by Ikebuchi and Nakano \cite{IN19}.

\begin{proposition}
$\IA$ defined as above is reflexive, with $\mathsf{app}=\BB\colon 2\rightarrow 1$
and $\mathsf{cur}(a)=(a\,\II)^\bullet\circ\BB^m\colon m\rightarrow 1$ for $a\colon m+1\rightarrow 1$.
Moreover, $\IA(0,1)=\{a\in\calA~|~a\colon 0\rightarrow 1\}=\{a^\bullet~|~a\in\calA\}=\calA^\bullet$ has a $\BIdot$-algebra structure
with $a\cdot_{\calA^\bullet} b=b\circ a\circ \BB$, $\BB_{\calA^\bullet}=\BB^\bullet$,
$\II_{\calA^\bullet}=\II^\bullet$, and $a^{\bullet_{\calA^\bullet}}=a^\bullet$.
This $\BIdot$-algebra is isomorphic to $\calA$, with isomorphisms $a\mapsto a^\bullet\colon \calA\stackrel{\cong}{\rightarrow}\calA^\bullet$ and
$a\mapsto a\,\II\colon \calA^\bullet\stackrel{\cong}{\rightarrow}\calA$.
\end{proposition}

\begin{corollary} \label{cor:Scott:planar}
Every extensional $\BIdot$-algebra is isomorphic to one arising from a reflexive object in a monoidal category.
\end{corollary}

A pleasant consequence of having the internal PRO is that we can use the string diagrams of monoidal categories for expressing and
reasoning about combinators with arity. This makes the lengthy and dry equational reasoning on combinators 
much easier and more accessible. Let us draw 
{\unitlength=.6pt
\begin{picture}(40,20)(0,5)
\thicklines
\arcLR{0}{0}{20}{10}
\lineseg{20}{10}{40}{10}
\put(20,10){\makebox(0,0){\appnode}}
\end{picture}
}
for $\BB\colon 2\rightarrow 1$ and
{\unitlength=.6pt
\begin{picture}(40,20)(0,5)
\thicklines
\lineseg{20}{10}{40}{10}
\put(10,10){\circle{20}}
\put(10,10){\makebox(0,0){$a$}}
\end{picture}
}
for $a^\bullet\colon 0\rightarrow 1$.
For example, the arity axioms can be drawn (keeping in mind that adding a lower wire amounts to applying $\BB$ while nothing to do for adding an upper wire)
as follows.
\begin{center}
\begin{picture}(90,85)
\thicklines
\lineseg{0}{69}{60}{69}
\lineseg{0}{50}{45}{50}
\lineseg{0}{30}{30}{30}
\lineseg{75}{52}{90}{52}
\arcLR{60}{35}{15}{17}
\arcLR{45}{20}{15}{15}
\arcLR{30}{10}{15}{10}
\put(20,10){\circle{20}}
\put(20,10){\makebox(0,0){$\BB$}}
\put(45,20){\makebox(0,0){\appnode}}
\put(60,35){\makebox(0,0){\appnode}}
\put(75,52){\makebox(0,0){\appnode}}
\put(45,80){\makebox(0,0){$\BB^\bullet\circ\BB\circ\BB\circ\BB$}}
\end{picture}
\begin{picture}(20,70)
\put(10,35){\makebox(0,0){$=$}}
\end{picture}
\begin{picture}(55,70)
\thicklines
\arcLR{0}{35}{20}{10}
\arcLR{20}{15}{20}{15}
\lineseg{0}{15}{20}{15}
\lineseg{40}{30}{55}{30}
\put(20,45){\makebox(0,0){\appnode}}
\put(40,30){\makebox(0,0){\appnode}}
\put(27.5,80){\makebox(0,0){$(\BB\,\BB)\circ\BB$}}
\end{picture}
\end{center}
\begin{center}
\begin{picture}(60,60)
\thicklines
\lineseg{0}{30}{30}{30}
\lineseg{45}{20}{60}{20}
\arcLR{30}{10}{15}{10}
\put(20,10){\circle{20}}
\put(20,10){\makebox(0,0){$\II$}}
\put(45,20){\makebox(0,0){\appnode}}
\put(30,45){\makebox(0,0){$\II^\bullet\circ\BB$}}
\end{picture}
\begin{picture}(20,40)
\put(10,20){\makebox(0,0){$=$}}
\end{picture}
\begin{picture}(40,40)
\thicklines
\lineseg{0}{20}{40}{20}
\put(20,45){\makebox(0,0){$\BB\,\II$}}
\end{picture}
\begin{picture}(40,40)
\end{picture}
\begin{picture}(60,40)
\thicklines
\lineseg{0}{30}{30}{30}
\lineseg{45}{20}{60}{20}
\arcLR{30}{10}{15}{10}
\put(20,10){\circle{20}}\put(20,10){\circle{14}}
\put(20,10){\makebox(0,0){$a$}}
\put(45,20){\makebox(0,0){\appnode}}
\put(30,45){\makebox(0,0){$a^{\bullet\bullet}\circ\BB$}}
\end{picture}
\begin{picture}(20,40)
\put(10,20){\makebox(0,0){$=$}}
\end{picture}
\begin{picture}(60,40)
\thicklines
\lineseg{0}{10}{30}{10}
\lineseg{45}{20}{60}{20}
\arcLR{30}{10}{15}{10}
\put(20,30){\circle{20}}
\put(20,30){\makebox(0,0){$a$}}
\put(45,20){\makebox(0,0){\appnode}}
\put(30,50){\makebox(0,0){$(\BB\,a^\bullet)\circ\BB$}}
\end{picture}
\end{center}
Thus they have very natural graphical meaning, from which it is immediate to see that $\calA^\bullet$ forms a $\BIdot$-algebra.

It is natural to ask if every reflexive PRO arises as the internal PRO of an extensional $\BIdot$-algebra. 
The answer is no; a simple counterexample is as follows. 
\begin{example} \label{ex:pro}
Let $\mathbb{C}$ be the PRO
such that $\mathbb{C}(0,0)=\{0,1\}$ and $\mathbb{C}(m,n)=\{0\}$ for $m+n>0$, 
with composition and tensor of $f$ and $g$ given by the multiplication
$f\cdot g$, and $\id_0=1$ while $\id_m=0$ for $m>0$. Thus $\mathbb{C}$ is trivial except $\mathbb{C}(0,0)$. It is reflexive
as $\mathbb{C}(m+1,1)=\calC(m,1)=\{0\}$. 
$\calA=\mathbb{C}(0,1)$ is the trivial $\BIdot$-algebra, and its
internal PRO $\IA$ is the trivial PRO ($\IA(m,n)=\{0\}$ for any $m$ and $n$).
(In this example, $\mathbb{C}(0,0)$ can be replaced by any commutative monoid with an absorbing element $0$. Hence there are infinitely many reflexive PROs giving rise to the trivial $\BIdot$-algebra.) 
\end{example}

\subsection{Adding braids and symmetry}\label{subsec:braidedCA}

We have seen that, for an extensional $\BIdot$-algebra $\calA$, $\IA$ is a reflexive PRO. Now suppose that
$\IA$ has a braid, i.e., it is a reflexive PROB. Following the discussion on PROBs
(Section \ref{subsec:braids}), we need 
elements $\CC^+$ and $\CC^-$ which are of arity $2\rightarrow 2$ and satisfy the Coxter relations
$$
\begin{array}{rcll}
\CC^\pm\circ\CC^\mp &=& \BB\,(\BB\,\II)\\
(\BB\,\CC^\pm)\circ\CC^\pm\circ(\BB\,\CC^\pm) &=& \CC^\pm\circ(\BB\,\CC^\pm)\circ\CC^\pm
\end{array}
$$
such that the 
family of elements $\bsigma_{m,n}\colon m+n\rightarrow n+m$ given by
$$
\begin{array}{rclrcl}
\bsigma_{0,1} &=& \II & \hspace{5mm}
\bsigma_{m+1,1} &=& 
(\BB\,\bsigma_{m,1})\circ\CC^+ \\
\bsigma_{m,0} &=& \II &
\bsigma_{m,n+1} &=&  
\bsigma_{m,1}\circ(\BB\,\bsigma_{m,n}).
\end{array}
$$
is natural in $m$ and $n$. It turns out that naturality follows from just the following instances:
$$
\begin{array}{rcll}
(1+\BB)\circ\bsigma_{1,1}^{\pm1} &=& \bsigma_{1,2}^{\pm1}\circ (\BB+1) \\
(1+a^\bullet)\circ\bsigma_{1,1}^{\pm1} &=& a^\bullet+1
\end{array}
$$
or equivalently
$$
\begin{array}{rcll}
(\BB\,\BB)\circ\CC^\pm &=& \CC^\pm\circ(\BB\,\CC^\pm)\circ\BB \\
(\BB\,a^\bullet)\circ\CC^\pm &=& a^\bullet
\end{array}
$$
By drawing
{\unitlength=.6pt
\begin{picture}(40,20)(0,5)
\thicklines
\braid{0}{0}{10}{10}
\end{picture}
}
and
{\unitlength=.6pt
\begin{picture}(40,20)(0,5)
\thicklines
\braidInv{0}{0}{10}{10}
\end{picture}
}
for $\CC~+$ and $\CC^-$ respectively,
these naturality with respect to $\BB$ and $a^\bullet$ can be depicted as follows.
\begin{center}
\unitlength=.7pt
\begin{picture}(60,60)
\thicklines
\arcLR{0}{20}{20}{10}  
\lineseg{0}{0}{20}{0}
\braid{20}{0}{10}{15}
\put(20,30){\makebox(0,0){$\appnode$}}
\put(30,50){\makebox(0,0){$(\BB\,\BB)\circ\CC^+$}}
\end{picture}
\begin{picture}(40,40)
\put(20,20){\makebox(0,0){$=$}}
\end{picture}
\begin{picture}(75,40)
\thicklines
\arcLR{40}{0}{20}{10}  
\lineseg{0}{40}{20}{40}
\braid{0}{0}{5}{10}
\braid{20}{20}{5}{10}
\lineseg{20}{0}{40}{0}
\lineseg{60}{10}{75}{10}
\lineseg{40}{40}{75}{40}
\put(60,10){\makebox(0,0){$\appnode$}}
\put(37.5,50){\makebox(0,0){$\CC^+\circ (\BB\,\CC^+)\circ\BB$}}
\end{picture}
\begin{picture}(40,40)
\end{picture}
\begin{picture}(60,60)
\thicklines
\arcLR{0}{20}{20}{10}  
\lineseg{0}{0}{20}{0}
\braidInv{20}{0}{10}{15}
\put(20,30){\makebox(0,0){$\appnode$}}
\put(30,50){\makebox(0,0){$(\BB\,\BB)\circ\CC^-$}}
\end{picture}
\begin{picture}(40,40)
\put(20,20){\makebox(0,0){$=$}}
\end{picture}
\begin{picture}(75,40)
\thicklines
\arcLR{40}{0}{20}{10}  
\lineseg{0}{40}{20}{40}
\braidInv{0}{0}{5}{10}
\braidInv{20}{20}{5}{10}
\lineseg{20}{0}{40}{0}
\lineseg{60}{10}{75}{10}
\lineseg{40}{40}{75}{40}
\put(60,10){\makebox(0,0){$\appnode$}}
\put(37.5,50){\makebox(0,0){$\CC^-\circ (\BB\,\CC^-)\circ\BB$}}
\end{picture}
\end{center}
\begin{center}
\unitlength=.7pt
\begin{picture}(60,55)
\thicklines
\lineseg{0}{0}{20}{0}
\braid{20}{0}{10}{10}
\put(10,20){\circle{20}}
\put(10,20){\makebox(0,0){$a$}}
\put(30,40){\makebox(0,0){$(\BB\,a^\bullet)\circ\CC^+$}}
\end{picture}
\begin{picture}(40,40)
\put(20,10){\makebox(0,0){$=$}}
\end{picture}
\begin{picture}(40,30)
\thicklines
\lineseg{0}{25}{40}{25}
\lineseg{20}{10}{40}{10}
\put(10,10){\circle{20}}
\put(10,10){\makebox(0,0){$a$}}
\put(20,40){\makebox(0,0){$a^\bullet$}}
\end{picture}
\begin{picture}(40,40)
\put(20,10){\makebox(0,0){$=$}}
\end{picture}
\begin{picture}(60,50)
\thicklines
\lineseg{0}{0}{20}{0}
\braidInv{20}{0}{10}{10}
\put(10,20){\circle{20}}
\put(10,20){\makebox(0,0){$a$}}
\put(30,40){\makebox(0,0){$(\BB\,a^\bullet)\circ\CC^-$}}
\end{picture}
\end{center}

\noindent
Note the naturality with respect to $a^\bullet$ implies $\CC^\pm\,a\,b\,c=a\,c\,b$ as well as $\CC^\pm\,\II\,a=a^\bullet$
and more generally $\CC^\pm\,b\,a=a^\bullet\circ b$.
So it is possible to axiomatise algebras with braid without using $(\_)^\bullet$ as a primitive construct.
The equation $a^\bullet\circ\BB^{m+1}=(\BB\,a)\circ\BB^n$ for characterising the arity $m\rightarrow n$ 
(Definition \ref{def:arity}) can be 
replaced by 
$$(\CC^\pm\,\BB\,a)\circ\BB^m=(\BB\,a)\circ\BB^n$$
in the braided setting. As a result, we obtain {\em extensional $\BCpmI$-algebras} \cite{Has22}:

\begin{definition}
An {\em  extensional $\BCpmI$-algebra} is an applicative structure
with elements $\BB$, $\CC^+$, $\CC^-$ and $\II$ satisfying the following axioms.
\setlength{\FrameSep}{4pt}
\begin{oframed}
$$
\begin{array}{rcll}
\BB\,a\,b\,c &=& a\,(b\,c) & (B)\\
\CC^\star\,a\,b\,c &=& a\,c\,b & (C)\\
\II\,a &=& a & (I)\\ 
\CC^+\,a\,b &=& \CC^-\,a\,b & (C2)\\
\BB\,\II &=& \II  & (\lambda)\\
\CC^\star\,\BB\,\II &=& \BB\,\II & (\rho)\\
(\CC^\star\,\BB\,\BB)\circ(\BB\circ\BB) &=& (\BB\,\BB)\circ\BB  & (\alpha)\\
\CC^\pm\circ\CC^\mp &=& \BB\,(\BB\,\II) & (\cox_1)\\
(\CC^\star\,\BB\,\CC^\pm)\circ(\BB\circ\BB) &=& (\BB\,\CC^\pm)\circ(\BB\circ\BB)  & (\cox_2)\\
~~~
(\BB\,\CC^\pm)\circ(\CC^\pm\circ(\BB\,\CC^\pm)) &=& 
\CC^\pm\circ((\BB\,\CC^\pm)\circ\CC^\pm) & (\cox_3)~~~\\
(\BB\,\BB)\circ\CC^\pm &=& \CC^\pm\circ((\BB\,\CC^\pm)\circ\BB) &(bc)\\
\end{array}
$$ 
\end{oframed}
The double signs $\pm$ and $\mp$ in an equation should be taken 
as appropriately linked, while $\star$ indicates an arbitrary choice of $+$ or $-$.
(As we have $(C2)$, assuming just an instance of $\star$ suffices.)
\end{definition}

\begin{proposition}\label{prop:refl:braided}
Suppose that $A$ is a reflexive object in a braided monoidal category $\mathbb{C}$.
Then $\mathbb{C}(I,A)$ is an extensional $\BCpmI$-algebra with (in addition to the constructs from Proposition \ref{prop:refl:planar})
$$\CC^\pm=\mathsf{cur}(\mathsf{cur}(\mathsf{cur}((A\otimes\sigma_{A,A}^{\pm1});(\mathsf{app}\otimes A);\mathsf{app}))).$$
\end{proposition}

\begin{proposition}
For an extensional $\BCpmI$-algebra $\calA$, its internal PRO $\IA$ is braided (i.e., a PROB) with the braid $\bsigma$.
$\IA$ is reflexive, and $\IA(0,1)=\calA^\bullet$ forms an extensional $\BCpmI$-algebra
which is isomorphic to $\calA$.
\end{proposition}
\begin{proof}
The most non-trivial point is the naturality of $\bsigma$. Naturality with respect to elements of arity $m\rightarrow 1$ easily follows from that with respect to $\BB$ and $a^\bullet$
as they can be expressed as $a^\bullet\circ\BB^m$ for some $a$. Naturality for general $h\colon m\rightarrow n$
is reduced to that w.r.t. $(\BB\,h)\circ\BB^n\colon m+1\rightarrow 1$.
\end{proof}

\begin{corollary}
Every extensional $\BCpmI$-algebra is isomorphic to one arising from a reflexive object in a braided monoidal category.
\end{corollary}

Similarly, we can characterise {\em extensional $\BCI$-algebras} as extensional $\BCpmI$-algebras in which $\CC^+$ agrees with $\CC^-$.
For the resulting axiomatisation, we refer to \cite{Has22}. The internal PRO $\IA$
of an extensional $\BCI$-algebra $\calA$ is symmetric (i.e., a PROP),
and its reflexive object $1$ induces an extensional $\BCI$-algebra $\calA^\bullet$  which is isomorphic to $\calA$.

\subsection{Adding twists}

Now we shall add a {\em twist} to the internal PROB $\IA$ of an extensional $\BCpmI$-algebra $\calA$
following the discussions on balanced PROB in Section \ref{subsec:braids}.

\begin{definition}
An extensional $\BCpmI$-algebra is {\em balanced}
when it has elements $\btheta^+$ and $\btheta^-$ satisfying
the following axioms.
\setlength{\FrameSep}{4pt}
\begin{oframed}
$$
\begin{array}{rcll}
(\BB\,\btheta^\pm)\circ\BB &=& (\CC^\star\,\BB\,\btheta^\pm)\circ\BB & 
(\theta_1)\\
\btheta^\pm\circ\btheta^\mp &=& \BB\,\II & (\theta_2)\\
\CC^\star\,\btheta^\star\,a&=& \CC^\star\,\II\,a & (\theta_3)\\
\BB\circ\btheta^\pm &=&
\btheta^\pm\circ\CC^\pm\circ\btheta^\pm\circ\CC^\pm\circ\BB & (\theta_4)\\
\end{array}
$$ 
\end{oframed}
\end{definition}
$(\theta_1)$ says $\btheta^+$ and $\btheta^-$ are of arity $1\rightarrow 1$.
$(\theta_2)$ says $\btheta^+$ and $\btheta^-$ are inverse to each other.
$(\theta_3)$ implies $a^\bullet\circ\btheta^\pm=a^\bullet$
(naturality of the twist for $a^\bullet$). 
It also implies
$\btheta^\pm\,a\,b=a\,b$. 
$(\theta_4)$ is the naturality of the twist for $\BB$.
By drawing
{\unitlength=.6pt
\begin{picture}(20,20)
\thicklines
\twist{0}{0}
\end{picture}
}
and 
{\unitlength=.6pt
\begin{picture}(20,20)
\thicklines
\twistInv{0}{0}
\end{picture}
}
for $\btheta^+$ and $\btheta^-$, $(\theta_4)$ is depicted as follows.
\begin{center}
\unitlength=.7pt
\begin{picture}(60,40)
\thicklines
\arcLR{0}{0}{20}{10}
\lineseg{20}{10}{30}{10}
\twist{30}{10}
\lineseg{50}{10}{60}{10}
\put(20,10){\makebox(0,0){$\appnode$}}
\put(30,30){\makebox(0,0){$\BB\circ\btheta^+$}}
\end{picture}
\begin{picture}(40,20)
\put(20,10){\makebox(0,0){$=$}}
\end{picture}
\begin{picture}(120,20)
\thicklines
\arcLR{80}{0}{20}{10}
\lineseg{0}{20}{20}{20}
\twist{0}{0}
\braid{20}{0}{5}{10}
\lineseg{40}{20}{60}{20}
\twist{40}{0}
\braid{60}{0}{5}{10}
\lineseg{100}{10}{120}{10}
\put(100,10){\makebox(0,0){$\appnode$}}
\put(60,30){\makebox(0,0){$\btheta^+\circ\CC^+\circ\btheta^+\circ\CC^+\circ\BB$}}
\end{picture}
\end{center}

\begin{proposition}\label{prop:refl:balanced}
Suppose that $A$ is a reflexive object in a balanced monoidal category $\mathbb{C}$.
Then $\mathbb{C}(I,A)$ is a balanced extensional $\BCpmI$-algebra with (in addition to the constructs from Prop. \ref{prop:refl:planar} and \ref{prop:refl:braided})
$\btheta^\pm=\mathsf{cur}(\mathsf{cur}((A\otimes\theta_A^{\pm1});\mathsf{app}))$.
\end{proposition}

Note that an extensional $\BCI$-algebra is an extensional balanced 
$\BCpmI$-algbera 
with $\CC^+=\CC^-$ and $\btheta^+=\btheta^-=\II$.

Similarly to the case of braids, the naturality of $\btheta^\pm$ for $\BB$ and $a^\bullet$ is sufficient to 
ensure the naturality of the induced twist. Thus the internal PROB is balanced:

\begin{proposition}
The internal PROB of a balanced extensional $\BCpmI$-algebra
is balanced with the twist $\btheta_m\colon m\rightarrow m$ determined by
$\btheta_0=\II$, $\btheta_1=\btheta^+$ and 
$\btheta_{m+1}=\btheta^+\circ\bsigma_{1,m}\circ\btheta_m\circ\bsigma_{m,1}$.
\end{proposition}

\begin{corollary}
Every balanced extensional $\BCpmI$-algebra is isomorphic to one arising from a reflexive object in a balanced monoidal category.
\end{corollary}

\section{Ribbon categories with a reflexive object}\label{sec:ribbon}

Starting with combinatory algebras which give rise to reflexive PROs, we have enriched our algebras with braids
and twists.
The next step is to add trace and/or duality. But before doing that, it is useful to look at the situation from the categorical side.
Specifically, we consider ribbon categories with a reflexive object and figure out what the combinatory algebra arising from the 
reflexive object looks like. Most interestingly, we will see that these combinatory algebras have an element whose application
sends elements of arity $1+m\rightarrow 1+n$ to its (left) trace of arity $m\rightarrow n$.

\subsection{Reflexive objects in ribbon categories}
Suppose that $A$ is an object of a ribbon category $\mathbb{C}$
with an isomorphism
$\varphi\colon A\otimes A^*\stackrel{\cong}{\rightarrow} A$, 
i.e., $A$ is a reflexive object.
Define $\app\colon A\otimes A\rightarrow A$ and
$\lam\colon A\rightarrow A\otimes A$ by
$$
\app~=~
A\otimes A\xrightarrow{\varphi^{-1}\otimes A}
A\otimes A^*\otimes A \xrightarrow{A\otimes\varepsilon}
A
$$
$$
\begin{array}{rcl}
\lam
&=&
A\xrightarrow{A\otimes\eta}
A\otimes A\otimes A^* \xrightarrow{A\otimes A\otimes\theta^{-1}}
A\otimes A\otimes A^* 
\xrightarrow{A\otimes\sigma^{-1}}
A\otimes A^*\otimes A \xrightarrow{\varphi\otimes A}
A\otimes A
\\
&=&
A\xrightarrow{A\otimes\eta}
A\otimes A\otimes A^* \xrightarrow{A\otimes A\otimes\theta}
A\otimes A\otimes A^* 
\xrightarrow{A\otimes\sigma}
A\otimes A^*\otimes A \xrightarrow{\varphi\otimes A}
A\otimes A
\end{array}
$$
\begin{center}
\unitlength=.7pt
\begin{picture}(60,50)
\thicklines
\lineseg{0}{30}{20}{30}
\lineseg{0}{10}{20}{10}
\lineseg{40}{20}{60}{20}
\rgbfbox{20}{10}{20}{20}{.8}{.8}{1}{\tiny$\app$}
\end{picture}
\begin{picture}(40,50)
\put(20,20){\makebox(0,0){$=$}}
\end{picture}
\begin{picture}(60,50)
\thicklines
\lineseg{0}{40}{50}{40}
\arcLR{50}{20}{10}{10}
\lineseg{0}{10}{20}{10}
\lineseg{40}{20}{50}{20}
\lineseg{40}{0}{60}{0}
\rgbfbox{20}{0}{20}{20}{1}{1}{.5}{\small$\varphi\!\!^{-1}$}%
\end{picture}
\end{center}
\begin{center}
\unitlength=.7pt
\begin{picture}(60,50)
\thicklines
\lineseg{40}{30}{60}{30}
\lineseg{40}{10}{60}{10}
\lineseg{0}{20}{20}{20}
\rgbfbox{20}{10}{20}{20}{1}{.5}{.5}{\tiny$\lam$}
\end{picture}
\begin{picture}(40,50)
\put(20,20){\makebox(0,0){$=$}}
\end{picture}
\begin{picture}(100,50)
\thicklines
\arcLR{10}{20}{-10}{10}
\twistInv{20}{40}
\braidInv{40}{20}{5}{10}
\lineseg{0}{0}{60}{0}
\lineseg{10}{40}{20}{40}
\lineseg{10}{20}{40}{20}
\lineseg{60}{40}{100}{40}
\lineseg{80}{10}{100}{10}
\put(60,0){\framebox(20,20){\small$\varphi$}}
\rgbfbox{60}{0}{20}{20}{.5}{1}{1}{\small$\varphi$}%
\end{picture}
\begin{picture}(40,50)
\put(20,20){\makebox(0,0){$=$}}
\end{picture}
\begin{picture}(100,50)
\thicklines
\arcLR{10}{20}{-10}{10}
\twist{20}{40}
\braid{40}{20}{5}{10}
\lineseg{0}{0}{60}{0}
\lineseg{10}{40}{20}{40}
\lineseg{10}{20}{40}{20}
\lineseg{60}{40}{100}{40}
\lineseg{80}{10}{100}{10}
\put(60,0){\framebox(20,20){\small$\varphi$}}
\rgbfbox{60}{0}{20}{20}{.5}{1}{1}{\small$\varphi$}%
\end{picture}
\end{center}
Using $\lam$ and trace, we can express the currying
$$
\mathsf{cur}(f)=\mathit{Tr}^A_{X,A}(f;\lam)
=(X\otimes \eta_A);(f\otimes A^*);\varphi
\colon X\rightarrow A $$
for $f\colon X\otimes A\rightarrow A$.
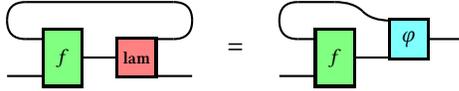
\begin{center}
\unitlength=.7pt
\begin{picture}(100,45)
\thicklines
\arcLR{10}{25}{-10}{10}
\lineseg{10}{45}{90}{45}
\arcLR{90}{25}{10}{10}
\lineseg{0}{5}{20}{5}
\lineseg{10}{25}{20}{25}
\lineseg{40}{15}{60}{15}
\lineseg{80}{25}{90}{25}
\lineseg{80}{5}{100}{5}
\rgbfbox{20}{0}{20}{30}{.5}{1}{.5}{\scriptsize$f$}%
\rgbfbox{60}{5}{20}{20}{1}{.5}{.5}{\tiny$\lam$}%
\end{picture}
\begin{picture}(40,45)
\put(20,20){\makebox(0,0){$=$}}
\end{picture}
\begin{picture}(100,45)
\thicklines
\arcLR{10}{25}{-10}{10}
\lineseg{10}{45}{30}{45}
\qbezier(30,45)(40,45)(45,40)\qbezier(45,40)(50,35)(60,35)
\lineseg{0}{5}{20}{5}
\lineseg{10}{25}{20}{25}
\lineseg{40}{15}{60}{15}
\lineseg{80}{25}{100}{25}
\rgbfbox{20}{0}{20}{30}{.5}{1}{.5}{\scriptsize$f$}%
\rgbfbox{60}{15}{20}{20}{.5}{1}{1}{\small$\varphi$}%
\end{picture}
\end{center}
\begin{lemma}\label{lem:beta-eta}
The following equations hold.
\begin{itemize}
\item The $\beta$-equality:
$(\lam\otimes A);(A\otimes\sigma_{A,A}^{-1});(\app\otimes A)
=A\otimes\theta_A^{-1}\colon A\rightarrow A.$
\begin{center}
\unitlength=.7pt
\begin{picture}(100,50)
\thicklines
\lineseg{0}{40}{40}{40}
\lineseg{0}{10}{20}{10}
\lineseg{40}{0}{60}{0}
\lineseg{60}{40}{100}{40}
\lineseg{80}{10}{100}{10}
\braidInv{40}{20}{5}{10}
\rgbfbox{20}{0}{20}{20}{1}{.5}{.5}{\tiny$\lam$}%
\rgbfbox{60}{0}{20}{20}{.8}{.8}{1}{\tiny$\app$}%
\end{picture}
\begin{picture}(40,50)
\put(20,20){\makebox(0,0){$=$}}
\end{picture}
\begin{picture}(60,50)
\thicklines
\lineseg{0}{40}{20}{40}
\qbezier(20,40)(27,40)(28,41)\qbezier(32,44)(37,50)(30,50)
\qbezier(30,50)(23,50)(30,42)\qbezier(30,42)(33,40)(40,40)
\lineseg{40}{40}{60}{40}
\lineseg{0}{10}{60}{10}
\end{picture}
\end{center}
\item The $\eta$-equality:
$\mathit{Tr}^A_{A,A}(\app;\lam)=\mathit{id}_A\colon A\rightarrow A.$
\begin{center}
\unitlength=.7pt
\begin{picture}(100,40)
\thicklines
\qbezier(10,20)(0,20)(0,30)\qbezier(0,30)(0,40)(10,40)
\qbezier(90,40)(100,40)(100,30)\qbezier(100,30)(100,20)(90,20)
\lineseg{0}{0}{20}{0}
\lineseg{10}{40}{90}{40}
\lineseg{10}{20}{20}{20}
\lineseg{40}{10}{60}{10}
\lineseg{80}{20}{90}{20}
\lineseg{80}{0}{100}{0}
\rgbfbox{20}{0}{20}{20}{.8}{.8}{1}{\tiny$\app$}%
\rgbfbox{60}{0}{20}{20}{1}{.5}{.5}{\tiny$\lam$}%
\end{picture}
\begin{picture}(40,40)
\put(20,20){\makebox(0,0){$=$}}
\end{picture}
\begin{picture}(40,40)
\thicklines
\lineseg{0}{20}{40}{20}
\end{picture}
\end{center}
\end{itemize}
\end{lemma}
The $\beta$-equality above corresponds to the $\beta$-rule of the trivalent graph representation of lambda terms 
mentioned in the introduction.
Note that the usual $\beta\eta$-equalities
$$
(\mathsf{cur}(f)\otimes A);\app=f
\hspace{8pt}(f\colon X\otimes A\rightarrow A)
$$
\begin{center}
\unitlength=.7pt
\begin{picture}(160,65)
\thicklines
\qbezier(10,25)(0,25)(0,35)\qbezier(0,35)(0,45)(10,45)
\lineseg{10}{45}{90}{45}
\qbezier(90,45)(100,45)(100,35)\qbezier(100,35)(100,25)(90,25)
\lineseg{0}{5}{20}{5}
\lineseg{10}{25}{20}{25}
\lineseg{40}{15}{60}{15}
\lineseg{80}{25}{90}{25}
\qbezier(80,5)(90,5)(100,15)\qbezier(100,15)(110,25)(120,25)
\lineseg{0}{65}{80}{65}
\qbezier(80,65)(90,65)(100,55)\qbezier(100,55)(110,45)(120,45)
\lineseg{140}{35}{160}{35}
\rgbfbox{20}{0}{20}{30}{.5}{1}{.5}{\scriptsize$f$}%
\rgbfbox{60}{5}{20}{20}{1}{.5}{.5}{\tiny$\lam$}%
\rgbfbox{120}{25}{20}{20}{.8}{.8}{1}{\tiny$\app$}%
\end{picture}
\begin{picture}(40,65)
\put(20,30){\makebox(0,0){$=$}}
\end{picture}
\begin{picture}(100,65)(0,-20)
\thicklines
\lineseg{0}{5}{20}{5}
\lineseg{0}{25}{20}{25}
\lineseg{40}{15}{60}{15}
\rgbfbox{20}{0}{20}{30}{.5}{1}{.5}{\scriptsize$f$}%
\end{picture}
\end{center}
$$
\mathsf{cur}((g\otimes A);\app)=g
\hspace{8pt}(g\colon X\rightarrow A)
$$
\begin{center}
\unitlength=.7pt
\begin{picture}(140,60)(-40,-15)
\thicklines
\qbezier(-30,20)(-40,20)(-40,30)\qbezier(-40,30)(-40,40)(-30,40)
\qbezier(90,40)(100,40)(100,30)\qbezier(100,30)(100,20)(90,20)
\lineseg{0}{0}{20}{0}
\lineseg{-30}{40}{90}{40}
\lineseg{-30}{20}{20}{20}
\lineseg{40}{10}{60}{10}
\lineseg{80}{20}{90}{20}
\lineseg{80}{0}{100}{0}
\lineseg{-40}{10}{-20}{10}
\lineseg{-40}{-10}{-20}{-10}
\rgbfbox{-20}{-15}{20}{30}{1}{1}{.5}{\scriptsize$g$}%
\rgbfbox{20}{0}{20}{20}{.8}{.8}{1}{\tiny$\app$}%
\rgbfbox{60}{0}{20}{20}{1}{.5}{.5}{\tiny$\lam$}%
\end{picture}
\begin{picture}(40,40)
\put(20,20){\makebox(0,0){$=$}}
\end{picture}
\begin{picture}(100,65)(0,-10)
\thicklines
\lineseg{0}{5}{20}{5}
\lineseg{0}{25}{20}{25}
\lineseg{40}{15}{60}{15}
\rgbfbox{20}{0}{20}{30}{1}{1}{.5}{\scriptsize$g$}%
\end{picture}
\end{center}
easily follow from Lemma \ref{lem:beta-eta}.
A useful reformulation of the $\beta$-equality of Lemma \ref{lem:beta-eta} is
\begin{proposition} \label{prop:trace}
For $f\colon A\otimes B\rightarrow A\otimes C$, 
$$
(\lam\otimes B);(A\otimes f);(\app\otimes C)=A\otimes\mathit{Tr}_L^A(f)\colon A\otimes B\rightarrow A\otimes C
$$
holds,
where the left trace ${\mathit{Tr}_L}^X_{B,C}\colon \mathbb{C}(X\otimes B,X\otimes C)\rightarrow\mathbb{C}(B,C)$
is given by ${\mathit{Tr}_L}^X_{B,C}(f)=\mathit{Tr}^X_{B,C}( \sigma_{X,B}^{-1};f;\sigma_{C,X})$.
\begin{center}
\unitlength=.7pt
\begin{picture}(140,45)
\thicklines
\lineseg{0}{10}{20}{10}
\lineseg{0}{40}{60}{40}
\lineseg{40}{0}{100}{0}
\qbezier(40,20)(45,20)(50,25)\qbezier(50,25)(55,30)(60,30)
\lineseg{80}{40}{140}{40}
\qbezier(80,30)(85,30)(90,25)\qbezier(90,25)(95,20)(100,20)
\lineseg{120}{10}{140}{10}
\rgbfbox{60}{25}{20}{20}{.5}{1}{.5}{\scriptsize$f$}%
\rgbfbox{20}{0}{20}{20}{1}{.5}{.5}{\tiny$\lam$}%
\rgbfbox{100}{0}{20}{20}{.8}{.8}{1}{\tiny$\app$}%
\end{picture}
\begin{picture}(40,45)
\put(20,20){\makebox(0,0){$=$}}
\end{picture}
\begin{picture}(100,45)
\thicklines
\lineseg{0}{0}{100}{0}
\lineseg{0}{40}{40}{40}
\lineseg{60}{40}{100}{40}
\lineseg{30}{30}{70}{30}
\qbezier(30,30)(20,30)(20,20)\qbezier(20,20)(20,10)(30,10)
\qbezier(70,30)(80,30)(80,20)\qbezier(80,20)(80,10)(70,10)
\lineseg{30}{10}{70}{10}
\rgbfbox{40}{25}{20}{20}{.5}{1}{.5}{\scriptsize$f$}%
\end{picture}
\begin{picture}(40,45)
\put(20,20){\makebox(0,0){$=$}}
\end{picture}
\begin{picture}(100,50)
\thicklines
\lineseg{0}{0}{100}{0}
\qbezier(20,50)(10,50)(10,40)\qbezier(10,40)(10,30)(20,30)
\qbezier(80,50)(90,50)(90,40)\qbezier(90,40)(90,30)(80,30)
\lineseg{20}{50}{80}{50}
\braidInv{20}{20}{5}{5}
\lineseg{0}{20}{20}{20}
\braid{60}{20}{5}{5}
\lineseg{80}{20}{100}{20}
\rgbfbox{40}{15}{20}{20}{.5}{1}{.5}{\scriptsize$f$}%
\end{picture}
\end{center}
\end{proposition}
By letting $f$ be $\sigma^{-1}_{A,A}$ in Proposition \ref{lem:beta-eta}, we recover the $\beta$-equality of
Lemma \ref{lem:beta-eta}.
Proposition \ref{prop:trace} will be useful for defining the trace combinator below.
As an iterated version of this, we have: 

\begin{corollary} \label{cor:trace}
For $f\colon A^{\otimes m}\otimes B\rightarrow A^{\otimes m}\otimes C$, 
$$
(\lam_m\otimes B);(A\otimes f);(\app_m\otimes C)=A\otimes\mathit{Tr}_L^{A^{\otimes m}}(f)\colon A\otimes B\rightarrow A\otimes C
$$
where $\lam_m\colon A\rightarrow A^{\otimes (m+1)}$ is defined 
by $\lam_0=\mathit{id}_A$ and $\lam_{m+1}=\lam_m;(\lam\otimes A^{\otimes m})$,
and $\app_n\colon A^{\otimes{n+1}}\rightarrow A$ by
$\app_0=\mathit{id}_A$ and $\app_{n+1}=(\app\otimes A^{\otimes n});\app_n$.
\begin{center}
\unitlength=.7pt
\begin{picture}(140,45)
\thicklines
\lineseg{0}{10}{20}{10}
\lineseg{0}{40}{60}{40}
\lineseg{40}{0}{100}{0}
\lineseg{40}{30}{100}{30}
\lineseg{40}{20}{100}{20}
\lineseg{80}{40}{140}{40}
\lineseg{120}{10}{140}{10}
\rgbfbox{60}{15}{20}{30}{.5}{1}{.5}{\scriptsize$f$}%
\rgbfbox{15}{0}{25}{30}{1}{.5}{.5}{\tiny$\lam_2$}%
\rgbfbox{100}{0}{25}{30}{.8}{.8}{1}{\tiny$\app_2$}%
\end{picture}
\begin{picture}(40,45)
\put(20,20){\makebox(0,0){$=$}}
\end{picture}
\begin{picture}(100,45)
\thicklines
\lineseg{0}{0}{100}{0}
\lineseg{0}{40}{40}{40}
\lineseg{60}{40}{100}{40}
\lineseg{30}{30}{70}{30}
\lineseg{30}{20}{70}{20}
\qbezier(30,30)(20,30)(20,17.5)\qbezier(20,17.5)(20,5)(30,5)
\qbezier(70,30)(80,30)(80,17.5)\qbezier(80,17.5)(80,5)(70,5)
\qbezier(30,20)(25,20)(25,15)\qbezier(25,15)(25,10)(30,10)
\qbezier(70,20)(75,20)(75,15)\qbezier(75,15)(75,10)(70,10)
\lineseg{30}{10}{70}{10}
\lineseg{30}{5}{70}{5}
\rgbfbox{40}{15}{20}{30}{.5}{1}{.5}{\scriptsize$f$}%
\end{picture}
\end{center}
\end{corollary}

\paragraph{Remark on the trivial loop on a reflexive object.}

Let 
$$\bigcirc=\mathit{Tr}^A(\mathit{id}_A)\colon I\rightarrow I$$
which we may call
the {\em trivial loop} on $A$ (or the dimension of $A$). 
In general, $\bigcirc$ may not agree with $\id_I$ --- for example, the PRO in Example \ref{ex:pro} is
actually compact closed, in which $\bigcirc=0$ while $\id_I=1$.
Since we assumed $A\otimes A^*\cong A$, (by noting that $X^{**}\cong X$ and $(X\otimes Y)^*\cong Y^*\otimes X^*$ hold in ribbon categories)
we have $A\cong A^*\cong A\otimes A$. 
Explicitly, the isomorphism between $A$ and $A\otimes A$ can be
given as follows.
$$p=
\mathit{Tr}^A_{A.A\otimes A}((A\otimes\lam);(\app\otimes\lam))
\colon A\rightarrow A\otimes A$$
$$p^{-1}=
\mathit{Tr}^A_{A\otimes A,A}((\lam\otimes\app);(A\otimes\app))
\colon A\otimes A\rightarrow A$$
\begin{center}
\unitlength=.7pt
\begin{picture}(50,70)
\put(25,30){\makebox(0,0){$p=$}}
\end{picture}
\begin{picture}(120,70)
\thicklines
\qbezier(20,30)(0,30)(0,50)\qbezier(0,50)(0,70)(20,70)
\qbezier(80,70)(100,70)(100,60)\qbezier(100,60)(100,50)(80,50)
\lineseg{0}{10}{100}{10}
\lineseg{20}{70}{80}{70}
\lineseg{40}{40}{60}{40}
\lineseg{40}{20}{60}{20}
\lineseg{80}{30}{100}{30}
\rgbfbox{20}{20}{20}{20}{1}{.5}{.5}{\tiny$\lam$}%
\rgbfbox{60}{30}{20}{20}{1}{.5}{.5}{\tiny$\lam$}%
\rgbfbox{60}{0}{20}{20}{.8}{.8}{1}{\tiny$\app$}%
\end{picture}
\begin{picture}(50,70)
\put(25,30){\makebox(0,0){$p^{-1}=$}}
\end{picture}
\begin{picture}(100,70)
\thicklines
\qbezier(20,50)(0,50)(0,60)\qbezier(0,60)(0,70)(20,70)
\qbezier(80,30)(100,30)(100,50)\qbezier(100,50)(100,70)(80,70)
\lineseg{0}{30}{20}{30}
\lineseg{0}{10}{20}{10}
\lineseg{20}{70}{80}{70}
\lineseg{40}{40}{60}{40}
\lineseg{40}{20}{60}{20}
\lineseg{40}{0}{100}{0}
\rgbfbox{20}{30}{20}{20}{.8}{.8}{1}{\tiny$\app$}%
\rgbfbox{20}{0}{20}{20}{1}{.5}{.5}{\tiny$\lam$}%
\rgbfbox{60}{20}{20}{20}{.8}{.8}{1}{\tiny$\app$}%
\end{picture}
\end{center}
From this, we have 
$
\mathit{Tr}^A(\mathit{id}_A)
=\mathit{Tr}^A(p;p^{-1})
=\mathit{Tr}^{A\otimes A}(p^{-1};p)
=\mathit{Tr}^{A\otimes A}(\mathit{id}_{A\otimes A})
=\mathit{Tr}^A(\mathit{id}_A)\otimes \mathit{Tr}^A(\mathit{id}_A)
$.
Thus the trivial loop $\bigcirc$ is idempotent
($\bigcirc=\bigcirc\otimes\bigcirc=\bigcirc;\bigcirc$) and 
it does not make sense to count the number of 
occurrences of $\bigcirc$. In Section \ref{sec:comparison}, $\bigcirc$ will play an important role.

\paragraph{Duality on reflexive objects}
In the sequel, we will write $\mathbb{C}|_A$ for the monoidal full subcategory of $\mathbb{C}$ whose objects are 
tensor products of $A$. Since $A^*\cong A$, $\mathbb{C}|_A$ itself is a ribbon category. 
The dual of $A$ is $A$ itself with unit and counit given by
$$
\alpha=\mathit{Tr}^A_{I,A\otimes A}(\lam;(A\otimes\lam))\colon I\rightarrow A\otimes A
$$
$$
\beta=\mathit{Tr}^A_{A\otimes A,I}((A\otimes\app);\app)\colon A\otimes A\rightarrow I
$$
\begin{center}
\unitlength=.7pt
\begin{picture}(95,60)(45,35)
\thicklines
\qbezier(80,50)(85,50)(90,45)\qbezier(90,45)(95,40)(100,40)
\lineseg{100}{40}{140}{40}
\lineseg{80}{70}{100}{70}
\lineseg{120}{60}{140}{60}
\qbezier(60,60)(50,60)(50,75)\qbezier(50,75)(50,90)(60,90)
\lineseg{60}{90}{120}{90}
\qbezier(120,90)(130,90)(130,85)\qbezier(130,85)(130,80)(120,80)
\rgbfbox{60}{50}{20}{20}{1}{.5}{.5}{\tiny$\lam$}%
\rgbfbox{100}{60}{20}{20}{1}{.5}{.5}{\tiny$\lam$}%
\end{picture}
\unitlength=.7pt
\begin{picture}(80,60)
\end{picture}
\unitlength=.6pt
\begin{picture}(90,60)(85,35)
\thicklines
\qbezier(120,40)(125,40)(130,45)\qbezier(130,45)(135,50)(140,50)
\lineseg{80}{40}{120}{40} \lineseg{120}{70}{140}{70}
\qbezier(100,80)(90,80)(90,85)\qbezier(90,85)(90,90)(100,90)
\lineseg{100}{90}{160}{90}
\lineseg{80}{60}{100}{60}
\qbezier(160,90)(170,90)(170,75)\qbezier(170,75)(170,60)(160,60)
\rgbfbox{100}{60}{20}{20}{.8}{.8}{1}{\tiny$\app$}%
\rgbfbox{140}{50}{20}{20}{.8}{.8}{1}{\tiny$\app$}%
\end{picture}
\end{center}
Then we have $(\alpha\otimes\id_A);(\id_A\otimes\beta)=(\id_A\otimes\alpha);(\beta\otimes \id_A)=\mathit{id}_A$.
The observation below is useful for simplifying calculation on duality:

\begin{lemma}
The following equations hold.
$$
\alpha=\alpha;(\id_A\otimes\theta^{\pm1}_A);\sigma^{\pm1}_{A,A}=\alpha;(\theta^{\pm1}_A\otimes\id_A);\sigma^{\pm1}_{A,A}
$$
$$
\beta=
\sigma^{\pm1}_{A,A};(\id_A\otimes\theta^{\pm1}_A);\beta=
\sigma^{\pm1}_{A,A};(\theta^{\pm1}_A\otimes\id_A);\beta
$$
\begin{center}
\unitlength=.7pt
\begin{picture}(440,50)
\thicklines
\lineseg{20}{30}{40}{30}
\lineseg{20}{10}{40}{10}
\rgbfbox{0}{10}{20}{20}{1}{.8}{.8}{$\alpha$}%
\put(60,20){\makebox(0,0){$=$}}
\twist{100}{30}
\lineseg{100}{10}{120}{10}
\braid{120}{10}{5}{10}
\rgbfbox{80}{10}{20}{20}{1}{.8}{.8}{$\alpha$}%
\put(160,20){\makebox(0,0){$=$}}
\twistInv{400}{10}
\lineseg{400}{30}{420}{30}
\braidInv{420}{10}{5}{10}
\rgbfbox{380}{10}{20}{20}{1}{.8}{.8}{$\alpha$}%
\put(260,20){\makebox(0,0){$=$}}
\twist{300}{10}
\lineseg{300}{30}{320}{30}
\braid{320}{10}{5}{10}
\rgbfbox{280}{10}{20}{20}{1}{.8}{.8}{$\alpha$}%
\put(360,20){\makebox(0,0){$=$}}
\twistInv{200}{30}
\lineseg{200}{10}{220}{10}
\braidInv{220}{10}{5}{10}
\rgbfbox{180}{10}{20}{20}{1}{.8}{.8}{$\alpha$}%
\end{picture}
\end{center}
\begin{center}
\unitlength=.7pt
\begin{picture}(440,50)
\thicklines
\lineseg{0}{30}{20}{30}
\lineseg{0}{10}{20}{10}
\rgbfbox{20}{10}{20}{20}{.6}{.9}{1}{$\beta$}%
\put(60,20){\makebox(0,0){$=$}}
\twist{100}{30}
\lineseg{100}{10}{120}{10}
\braid{80}{10}{5}{10}
\rgbfbox{120}{10}{20}{20}{.6}{.9}{1}{$\beta$}%
\put(160,20){\makebox(0,0){$=$}}
\twistInv{400}{10}
\lineseg{400}{30}{420}{30}
\braidInv{380}{10}{5}{10}
\rgbfbox{420}{10}{20}{20}{.6}{.9}{1}{$\beta$}%
\put(260,20){\makebox(0,0){$=$}}
\twist{300}{10}
\lineseg{300}{30}{320}{30}
\braid{280}{10}{5}{10}
\rgbfbox{320}{10}{20}{20}{.6}{.9}{1}{$\beta$}%
\put(360,20){\makebox(0,0){$=$}}
\twistInv{200}{30}
\lineseg{200}{10}{220}{10}
\braidInv{180}{10}{5}{10}
\rgbfbox{220}{10}{20}{20}{.6}{.9}{1}{$\beta$}%
\end{picture}
\end{center}
In particular, in the symmetric (i.e., compact closed) case, $\alpha=\alpha;\sigma_{A,A}$
and $\beta=\sigma_{A,A};\beta$ hold.
\end{lemma}
\begin{proof}
For the first equation
{
\unitlength=.5pt
\begin{picture}(140,30)(0,10)
\thicklines
\lineseg{20}{30}{40}{30}
\lineseg{20}{10}{40}{10}
\rgbfbox{0}{10}{20}{20}{1}{.8}{.8}{$\alpha$}%
\put(60,20){\makebox(0,0){$=$}}
\twist{100}{30}
\lineseg{100}{10}{120}{10}
\braid{120}{10}{5}{10}
\rgbfbox{80}{10}{20}{20}{1}{.8}{.8}{$\alpha$}%
\end{picture}
}, see below: 
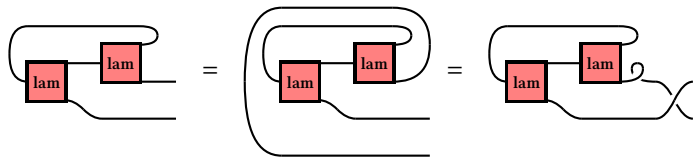
\begin{center}
\unitlength=.7pt
\begin{picture}(60,60)(85,35)
\thicklines
\qbezier(80,50)(85,50)(90,45)\qbezier(90,45)(95,40)(100,40)
\lineseg{100}{40}{140}{40}
\lineseg{80}{70}{100}{70}
\lineseg{120}{60}{140}{60}
\qbezier(60,60)(50,60)(50,75)\qbezier(50,75)(50,90)(60,90)
\lineseg{60}{90}{120}{90}
\qbezier(120,90)(130,90)(130,85)\qbezier(130,85)(130,80)(120,80)
\rgbfbox{60}{50}{20}{20}{1}{.5}{.5}{\tiny$\lam$}%
\rgbfbox{100}{60}{20}{20}{1}{.5}{.5}{\tiny$\lam$}%
\end{picture}
\begin{picture}(20,60)
\put(10,30){\makebox(0,0){$=$}}
\end{picture}
\begin{picture}(105,60)(35,35)
\thicklines
\qbezier(80,50)(85,50)(90,45)\qbezier(90,45)(95,40)(100,40)
\lineseg{100}{40}{140}{40}
\lineseg{80}{70}{100}{70}
\arcLR{120}{60}{20}{20}
\lineseg{60}{100}{120}{100}
\arcLR{60}{20}{-20}{40}
\lineseg{60}{20}{140}{20}
\qbezier(60,60)(50,60)(50,75)\qbezier(50,75)(50,90)(60,90)
\lineseg{60}{90}{120}{90}
\qbezier(120,90)(130,90)(130,85)\qbezier(130,85)(130,80)(120,80)
\rgbfbox{60}{50}{20}{20}{1}{.5}{.5}{\tiny$\lam$}%
\rgbfbox{100}{60}{20}{20}{1}{.5}{.5}{\tiny$\lam$}%
\end{picture}
\begin{picture}(20,60)
\put(10,30){\makebox(0,0){$=$}}
\end{picture}
\begin{picture}(80,60)(45,35)
\thicklines
\qbezier(80,50)(85,50)(90,45)\qbezier(90,45)(95,40)(100,40)
\lineseg{100}{40}{140}{40}
\lineseg{80}{70}{100}{70}
\twist{120}{60}
\braid{140}{40}{5}{10}
\qbezier(60,60)(50,60)(50,75)\qbezier(50,75)(50,90)(60,90)
\lineseg{60}{90}{120}{90}
\qbezier(120,90)(130,90)(130,85)\qbezier(130,85)(130,80)(120,80)
\rgbfbox{60}{50}{20}{20}{1}{.5}{.5}{\tiny$\lam$}%
\rgbfbox{100}{60}{20}{20}{1}{.5}{.5}{\tiny$\lam$}%
\end{picture}
\end{center}
The first step is by sliding and vanishing, and the second by tightening, superposing and yanking of the trace \cite{JSV96}. Other cases are similar.
\end{proof}


\subsection{The balanced extensional \texorpdfstring{$\BCpmI$}-algebra on a reflexive object}

Let $\calA=\mathbb{C}(I,A)$ and define 
the application $(\_)\cdot(\_)\colon \calA^2\rightarrow\calA$
by $f\cdot g = (f\otimes g);\app$. 
Thanks to Prop. \ref{prop:refl:balanced}, we know that $\calA$ is a balanced extensional $\BCpmI$-algebra.
Spelling out its detail, we have:

\begin{theorem}\label{thm:alg_on_refl_obj}
$(\calA,\cdot)$ is a balanced extensional $\BCpmI$-algebra
with $\BB,\CC^+,\CC^-,\II,\btheta^+$ and $\btheta^-$ given in Figure \ref{fig:BCpmI-alg}.
\end{theorem}

In Figure \ref{fig:BCpmI-alg}, $F$ denotes the strict monoidal functor from  $\mathbb{C}|_A$ to the internal PROB $\IA$ 
sending
$f\colon A^{\otimes m}\rightarrow A^{\otimes n}$ 
(emphasised with the greyed area in Figure \ref{fig:BCpmI-alg})
to 
$$
\mathit{Tr}^A_{I,A}
\left(
A\xrightarrow{\!\lam_m\!}
A^{\otimes (m+1)}\xrightarrow{\!\lam\otimes f\!}
A^{\otimes (n+2)}\xrightarrow{\!A\otimes \app_n\!}
A\otimes A
\right)
$$
of $\IA(m,n)$.
For instance, $f\colon A\otimes A\rightarrow A\otimes A$ is sent to
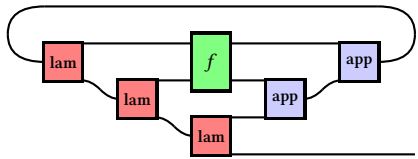
\begin{center}
\unitlength=.7pt
\begin{picture}(220,80)
\thicklines
\qbezier(20,50)(0,50)(0,65)\qbezier(0,65)(0,80)(20,80)
\qbezier(200,50)(220,50)(220,65)\qbezier(220,65)(220,80)(200,80)
\lineseg{20}{80}{200}{80}
\lineseg{40}{60}{100}{60} \lineseg{120}{60}{180}{60}
\qbezier(40,40)(45,40)(50,35)\qbezier(50,35)(55,30)(60,30)
\qbezier(80,20)(85,20)(90,15)\qbezier(90,15)(95,10)(100,10)
\qbezier(160,30)(165,30)(170,35)\qbezier(170,35)(175,40)(180,40)
\lineseg{80}{40}{100}{40} \lineseg{120}{40}{140}{40}
\lineseg{120}{20}{140}{20}
\lineseg{120}{0}{220}{0}
\rgbfbox{20}{40}{20}{20}{1}{.5}{.5}{\tiny$\lam$}%
\rgbfbox{60}{20}{20}{20}{1}{.5}{.5}{\tiny$\lam$}%
\rgbfbox{100}{35}{20}{30}{.5}{1}{.5}{\scriptsize$f$}%
\rgbfbox{100}{0}{20}{20}{1}{.5}{.5}{\tiny$\lam$}%
\rgbfbox{140}{20}{20}{20}{.8}{.8}{1}{\tiny$\app$}%
\rgbfbox{180}{40}{20}{20}{.8}{.8}{1}{\tiny$\app$}%
\end{picture}
\end{center}
of $\IA(2,2)$. For $\sigma_{A,A}\colon A\otimes A\rightarrow A\otimes A$,
we have $F(\sigma_{A,A})=\CC^+\in\IA(2,2)$.
Similarly, $F(\sigma_{A,A}^{-1})=\CC^-\in\IA(2,2)$, $F(\app)=\BB\in\IA(2,1)$,
$F(\mathit{id}_I)=\II=\mathit{Tr}^A(\lam)\in\IA(0,0)$
and 
$F(a)=a^\bullet=\mathit{Tr}^A((\lam\otimes a);\app)\in\IA(0,1)$
for $a\colon I\rightarrow A$.
We shall note that these 
can be seen as a reconstruction of the trivalent graph representations of combinators in \cite{Has22}
in ribbon categories.

On the other hand, $F(\lam)\in\IA(1,2)$ does not correspond to any conventional combinator used in combinatory logic. This is the {\em trace combinator} to be discussed below.

\begin{figure}
\begin{center}
\renewcommand{\arraystretch}{.5}
\begin{tabular}{|c|c|}
\hline
&\\
\unitlength=.5pt
\begin{picture}(40,60)
\put(20,30){\makebox(0,0){$a\cdot b=(a\otimes b);\app$ }}
\end{picture}
&
\unitlength=.5pt
\begin{picture}(90,60)(0,-5)
\thicklines
\qbezier(20,50)(25,50)(30,45)\qbezier(30,45)(35,40)(40,40)
\qbezier(20,10)(25,10)(30,15)\qbezier(30,15)(35,20)(40,20)
\lineseg{60}{30}{90}{30}
\rgbfbox{40}{20}{20}{20}{.8}{.8}{1}{\tiny$\app$}%
\rgbfbox{0}{0}{20}{20}{1}{1}{.5}{\scriptsize$a$}%
\rgbfbox{0}{40}{20}{20}{.5}{1}{1}{\scriptsize$b$}%
\end{picture}
\\
\hline
&\\
\unitlength=.5pt
\begin{picture}(180,80)
\put(90,40){\makebox(0,0){$\BB=F(\app)\colon 2\rightarrow1$}}
\end{picture}
&
\unitlength=.5pt
\begin{picture}(200,80)(0,-5)
\thicklines
\rgbbox{90}{30}{40}{40}{.91}{.91}{.91}{}%
\qbezier(20,50)(0,50)(0,65)\qbezier(0,65)(0,80)(20,80)
\qbezier(180,50)(200,50)(200,65)\qbezier(200,65)(200,80)(180,80)
\lineseg{20}{80}{180}{80}
\lineseg{40}{60}{100}{60} 
\qbezier(40,40)(45,40)(50,35)\qbezier(50,35)(55,30)(60,30)
\qbezier(80,20)(85,20)(90,15)\qbezier(90,15)(95,10)(100,10)
\qbezier(120,20)(125,20)(130,25)\qbezier(130,25)(135,30)(140,30)
\qbezier(160,40)(165,40)(170,45)\qbezier(170,45)(175,50)(180,50)
\lineseg{80}{40}{100}{40} \lineseg{120}{50}{140}{50}
\lineseg{120}{0}{200}{0}
\rgbfbox{20}{40}{20}{20}{1}{.5}{.5}{\tiny$\lam$}%
\rgbfbox{60}{20}{20}{20}{1}{.5}{.5}{\tiny$\lam$}%
\rgbfbox{100}{40}{20}{20}{.8}{.8}{1}{\tiny$\app$}%
\rgbfbox{100}{0}{20}{20}{1}{.5}{.5}{\tiny$\lam$}%
\rgbfbox{140}{30}{20}{20}{.8}{.8}{1}{\tiny$\app$}%
\end{picture}
\\
\hline
&\\
\unitlength=.5pt
\begin{picture}(40,40)
\put(20,20){\makebox(0,0){$\II=F(\mathit{id}_I)\colon 0\rightarrow0$}}
\end{picture}
&
\unitlength=.5pt
\begin{picture}(80,40)(10,-5)
\thicklines
\qbezier(20,20)(10,20)(10,30)\qbezier(10,30)(10,40)(20,40)
\qbezier(80,40)(90,40)(90,30)\qbezier(90,30)(90,20)(80,20)
\lineseg{20}{40}{80}{40}
\lineseg{60}{20}{80}{20}
\lineseg{60}{0}{90}{0}
\qbezier(20,20)(25,20)(30,15)\qbezier(30,15)(35,10)(40,10)
\rgbfbox{40}{0}{20}{20}{1}{.5}{.5}{\tiny$\lam$}%
\end{picture}
\\
\hline
&\\
\unitlength=.5pt
\begin{picture}(40,80)
\put(20,40){\makebox(0,0){$\CC^+=F(\sigma_{A,A})\colon 2\rightarrow2$}}
\end{picture}
&
\unitlength=.5pt
\begin{picture}(240,80)(-10,-5)
\thicklines
\rgbbox{90}{30}{40}{40}{.91}{.91}{.91}{}%
\qbezier(20,50)(0,50)(0,65)\qbezier(0,65)(0,80)(20,80)
\qbezier(200,50)(220,50)(220,65)\qbezier(220,65)(220,80)(200,80)
\lineseg{20}{80}{200}{80}
\lineseg{40}{60}{100}{60} \lineseg{120}{60}{180}{60}
\qbezier(40,40)(45,40)(50,35)\qbezier(50,35)(55,30)(60,30)
\qbezier(80,20)(85,20)(90,15)\qbezier(90,15)(95,10)(100,10)
\qbezier(160,30)(165,30)(170,35)\qbezier(170,35)(175,40)(180,40)
\lineseg{80}{40}{100}{40} \lineseg{120}{40}{140}{40}
\lineseg{120}{20}{140}{20}
\lineseg{120}{0}{220}{0}
\qbezier(100,60)(107,60)(108,54)\qbezier(112,46)(113,40)(120,40)
\qbezier(100,40)(107,40)(110,50)\qbezier(110,50)(113,60)(120,60)
\rgbfbox{20}{40}{20}{20}{1}{.5}{.5}{\tiny$\lam$}%
\rgbfbox{60}{20}{20}{20}{1}{.5}{.5}{\tiny$\lam$}%
\rgbfbox{100}{0}{20}{20}{1}{.5}{.5}{\tiny$\lam$}%
\rgbfbox{140}{20}{20}{20}{.8}{.8}{1}{\tiny$\app$}%
\rgbfbox{180}{40}{20}{20}{.8}{.8}{1}{\tiny$\app$}%
\end{picture}
\\
\hline
&\\
\unitlength=.5pt
\begin{picture}(40,80)
\put(20,40){\makebox(0,0){$\CC^-=F(\sigma^{-1}_{A,A})\colon 2\rightarrow2$}}
\end{picture}
&
\unitlength=.5pt
\begin{picture}(220,80)(0,-5)
\thicklines
\rgbbox{90}{30}{40}{40}{.91}{.91}{.91}{}%
\qbezier(20,50)(0,50)(0,65)\qbezier(0,65)(0,80)(20,80)
\qbezier(200,50)(220,50)(220,65)\qbezier(220,65)(220,80)(200,80)
\lineseg{20}{80}{200}{80}
\lineseg{40}{60}{100}{60} \lineseg{120}{60}{180}{60}
\qbezier(40,40)(45,40)(50,35)\qbezier(50,35)(55,30)(60,30)
\qbezier(80,20)(85,20)(90,15)\qbezier(90,15)(95,10)(100,10)
\qbezier(160,30)(165,30)(170,35)\qbezier(170,35)(175,40)(180,40)
\lineseg{80}{40}{100}{40} \lineseg{120}{40}{140}{40}
\lineseg{120}{20}{140}{20}
\lineseg{120}{0}{220}{0}
\qbezier(100,60)(107,60)(110,50)\qbezier(110,50)(113,40)(120,40)
\qbezier(100,40)(107,40)(108,46)\qbezier(112,54)(113,60)(120,60)
\rgbfbox{20}{40}{20}{20}{1}{.5}{.5}{\tiny$\lam$}%
\rgbfbox{60}{20}{20}{20}{1}{.5}{.5}{\tiny$\lam$}%
\rgbfbox{100}{0}{20}{20}{1}{.5}{.5}{\tiny$\lam$}%
\rgbfbox{140}{20}{20}{20}{.8}{.8}{1}{\tiny$\app$}%
\rgbfbox{180}{40}{20}{20}{.8}{.8}{1}{\tiny$\app$}%
\end{picture}
\\
\hline

&\\
\unitlength=.5pt
\begin{picture}(40,80)
\put(20,40){\makebox(0,0){$a^\bullet=F(a)\colon 0\rightarrow1$ }}
\end{picture}
&
\unitlength=.5pt
\begin{picture}(140,80)(60,-5)
\thicklines
\rgbbox{90}{30}{40}{40}{.91}{.91}{.91}{}%
\qbezier(60,50)(40,50)(40,65)\qbezier(40,65)(40,80)(60,80)
\qbezier(180,50)(200,50)(200,65)\qbezier(200,65)(200,80)(180,80)
\lineseg{60}{80}{180}{80}
\qbezier(60,50)(70,50)(80,30)\qbezier(80,30)(90,10)(100,10)
\qbezier(120,20)(125,20)(130,25)\qbezier(130,25)(135,30)(140,30)
\qbezier(160,40)(165,40)(170,45)\qbezier(170,45)(175,50)(180,50)
\lineseg{120}{50}{140}{50}
\lineseg{120}{0}{200}{0}
\rgbfbox{100}{40}{20}{20}{1}{1}{.5}{\scriptsize$a$}%
\rgbfbox{100}{0}{20}{20}{1}{.5}{.5}{\tiny$\lam$}%
\rgbfbox{140}{30}{20}{20}{.8}{.8}{1}{\tiny$\app$}%
\end{picture}
\\
\hline

&\\
\unitlength=.5pt
\begin{picture}(40,80)
\put(20,40){\makebox(0,0){$\btheta^+=F(\theta)\colon 0\rightarrow1$ }}
\end{picture}
&
\unitlength=.5pt
\begin{picture}(140,80)(50,-5)
\thicklines
\rgbbox{90}{30}{40}{40}{.91}{.91}{.91}{}%
\qbezier(40,50)(20,50)(20,65)\qbezier(20,65)(20,80)(40,80)
\qbezier(180,50)(200,50)(200,65)\qbezier(200,65)(200,80)(180,80)
\lineseg{40}{80}{180}{80}
\qbezier(40,50)(45,50)(50,45)\qbezier(50,45)(55,40)(60,40)
\qbezier(60,50)(70,50)(80,30)\qbezier(80,30)(90,10)(100,10)
\qbezier(120,20)(125,20)(130,25)\qbezier(130,25)(135,30)(140,30)
\qbezier(160,40)(165,40)(170,45)\qbezier(170,45)(175,50)(180,50)
\lineseg{120}{50}{140}{50}
\lineseg{120}{0}{200}{0}
\lineseg{80}{50}{100}{50}
\twistG{100}{50}
\rgbfbox{60}{30}{20}{20}{1}{.5}{.5}{\tiny$\lam$}%
\rgbfbox{100}{0}{20}{20}{1}{.5}{.5}{\tiny$\lam$}%
\rgbfbox{140}{30}{20}{20}{.8}{.8}{1}{\tiny$\app$}%
\end{picture}
\\
\hline

&\\
\unitlength=.5pt
\begin{picture}(40,80)
\put(20,40){\makebox(0,0){$\btheta^-=F(\theta^{-1})\colon 0\rightarrow1$ }}
\end{picture}
&
\unitlength=.5pt
\begin{picture}(140,80)(50,-5)
\thicklines
\rgbbox{90}{30}{40}{40}{.91}{.91}{.91}{}%
\qbezier(40,50)(20,50)(20,65)\qbezier(20,65)(20,80)(40,80)
\qbezier(180,50)(200,50)(200,65)\qbezier(200,65)(200,80)(180,80)
\lineseg{40}{80}{180}{80}
\qbezier(40,50)(45,50)(50,45)\qbezier(50,45)(55,40)(60,40)
\qbezier(60,50)(70,50)(80,30)\qbezier(80,30)(90,10)(100,10)
\qbezier(120,20)(125,20)(130,25)\qbezier(130,25)(135,30)(140,30)
\qbezier(160,40)(165,40)(170,45)\qbezier(170,45)(175,50)(180,50)
\lineseg{120}{50}{140}{50}
\lineseg{120}{0}{200}{0}
\lineseg{80}{50}{100}{50}
\twistInvG{100}{50}
\rgbfbox{60}{30}{20}{20}{1}{.5}{.5}{\tiny$\lam$}%
\rgbfbox{100}{0}{20}{20}{1}{.5}{.5}{\tiny$\lam$}%
\rgbfbox{140}{30}{20}{20}{.8}{.8}{1}{\tiny$\app$}%
\end{picture}
\\
\hline\hline

&\\
\unitlength=.5pt
\begin{picture}(40,80)
\put(20,40){\makebox(0,0){$\Tr=F(\lam)\colon 1\rightarrow2$ }}
\end{picture}
&
\unitlength=.5pt
\begin{picture}(200,80)(20,-5)
\thicklines
\rgbbox{90}{30}{40}{40}{.91}{.91}{.91}{}%
\qbezier(40,50)(20,50)(20,65)\qbezier(20,65)(20,80)(40,80)
\qbezier(200,50)(220,50)(220,65)\qbezier(220,65)(220,80)(200,80)
\lineseg{40}{80}{200}{80}
 \lineseg{120}{60}{180}{60}
\qbezier(40,50)(45,50)(50,40)\qbezier(50,40)(55,30)(60,30)
\qbezier(80,20)(85,20)(90,15)\qbezier(90,15)(95,10)(100,10)
\qbezier(160,30)(165,30)(170,35)\qbezier(170,35)(175,40)(180,40)
\qbezier(80,40)(85,40)(90,45)\qbezier(90,45)(95,50)(100,50)
\lineseg{120}{40}{140}{40}
\lineseg{120}{20}{140}{20}
\lineseg{120}{0}{220}{0}
\rgbfbox{60}{20}{20}{20}{1}{.5}{.5}{\tiny$\lam$}%
\rgbfbox{100}{40}{20}{20}{1}{.5}{.5}{\tiny$\lam$}%
\rgbfbox{100}{0}{20}{20}{1}{.5}{.5}{\tiny$\lam$}%
\rgbfbox{140}{20}{20}{20}{.8}{.8}{1}{\tiny$\app$}%
\rgbfbox{180}{40}{20}{20}{.8}{.8}{1}{\tiny$\app$}%
\end{picture}
\\
\hline
&\\
\unitlength=.5pt
\begin{picture}(40,80)
\put(25,50){\makebox(0,0){$\balpha=F(\alpha)\colon 0\rightarrow2$ }}
\end{picture}
&
\unitlength=.5pt
\begin{picture}(200,100)(20,-5)
\thicklines
\rgbbox{45}{35}{90}{60}{.91}{.91}{.91}{}%
\qbezier(40,10)(20,10)(20,55)\qbezier(20,55)(20,100)(40,100)
\qbezier(200,50)(220,50)(220,75)\qbezier(220,75)(220,100)(200,100)
\lineseg{40}{10}{100}{10}
\lineseg{40}{100}{200}{100}
 \lineseg{120}{60}{180}{60}
\qbezier(160,30)(165,30)(170,35)\qbezier(170,35)(175,40)(180,40)
\qbezier(80,50)(85,50)(90,45)\qbezier(90,45)(95,40)(100,40)
\lineseg{100}{40}{140}{40}
\lineseg{120}{20}{140}{20}
\lineseg{120}{0}{220}{0}
\lineseg{80}{70}{100}{70}
\qbezier(60,60)(50,60)(50,75)\qbezier(50,75)(50,90)(60,90)
\lineseg{60}{90}{120}{90}
\qbezier(120,90)(130,90)(130,85)\qbezier(130,85)(130,80)(120,80)
\rgbfbox{60}{50}{20}{20}{1}{.5}{.5}{\tiny$\lam$}%
\rgbfbox{100}{60}{20}{20}{1}{.5}{.5}{\tiny$\lam$}%
\rgbfbox{100}{0}{20}{20}{1}{.5}{.5}{\tiny$\lam$}%
\rgbfbox{140}{20}{20}{20}{.8}{.8}{1}{\tiny$\app$}%
\rgbfbox{180}{40}{20}{20}{.8}{.8}{1}{\tiny$\app$}%
\end{picture}
\\
\hline
&\\
\unitlength=.5pt
\begin{picture}(40,80)
\put(25,50){\makebox(0,0){${\boldsymbol\beta}=F(\beta)\colon 2\rightarrow 0$}}
\end{picture}
&
\unitlength=.5pt
\begin{picture}(200,100)(0,-5)
\thicklines
\rgbbox{85}{35}{90}{60}{.91}{.91}{.91}{}%
\qbezier(20,50)(0,50)(0,75)\qbezier(0,75)(0,100)(20,100)
\qbezier(180,20)(200,20)(200,60)\qbezier(200,60)(200,100)(180,100)
\lineseg{20}{100}{180}{100}
\lineseg{120}{20}{180}{20}
\lineseg{40}{60}{100}{60} 
\qbezier(40,40)(45,40)(50,35)\qbezier(50,35)(55,30)(60,30)
\qbezier(80,20)(85,20)(90,15)\qbezier(90,15)(95,10)(100,10)
\qbezier(120,40)(125,40)(130,45)\qbezier(130,45)(135,50)(140,50)
\lineseg{80}{40}{120}{40} \lineseg{120}{70}{140}{70}
\lineseg{120}{0}{200}{0}
\qbezier(100,80)(90,80)(90,85)\qbezier(90,85)(90,90)(100,90)
\lineseg{100}{90}{160}{90}
\qbezier(160,90)(170,90)(170,75)\qbezier(170,75)(170,60)(160,60)
\rgbfbox{20}{40}{20}{20}{1}{.5}{.5}{\tiny$\lam$}%
\rgbfbox{60}{20}{20}{20}{1}{.5}{.5}{\tiny$\lam$}%
\rgbfbox{100}{60}{20}{20}{.8}{.8}{1}{\tiny$\app$}%
\rgbfbox{100}{0}{20}{20}{1}{.5}{.5}{\tiny$\lam$}%
\rgbfbox{140}{50}{20}{20}{.8}{.8}{1}{\tiny$\app$}%
\end{picture}
\\
\hline
\end{tabular}
\end{center}
\Description[The ribbon combinatory algebra $\calA=\mathbb{C}(I,A)$]{The ribbon combinatory algebra $\calA=\mathbb{C}(I,A)$}
\caption{The ribbon combinatory algebra $\calA=\mathbb{C}(I,A)$}
\label{fig:BCpmI-alg}
\end{figure}

\subsection{Combinators for duality and trace}
\label{subsec:trace-combinator}

Figure \ref{fig:BCpmI-alg} contains additional ingredients: $\balpha=F(\alpha)$, $\bbeta=F(\beta)$ and $\Tr=F(\lam)$.
Since $\alpha\colon I\rightarrow A\otimes A$ and $\beta\colon A\otimes A\rightarrow I$ form a duality in $\mathbb{C}|_A$ and
$F\colon \mathbb{C}|_A\rightarrow\IA$ strictly preserves the balanced monoidal structure, 
$\balpha=F(\alpha)$ and $\bbeta=F(\beta)$ form a duality in $\IA$. So it is not surprising to see:
\begin{proposition}
The internal PROB $\IA$ is a ribbon category with duality $\balpha\colon 0\rightarrow2$ and $\bbeta\colon 2\rightarrow0$.  
\end{proposition}
Being a ribbon category, $\IA$ also has a unique trace \cite{JSV96,Has09}.
For $f\colon 1+m\rightarrow 1+n$, its left trace ${\mathit{Tr}_L}^1_{m,n}(f)\colon m\rightarrow n$ is $\balpha\circ(\BB\,f)\circ\bbeta$.
Thanks to combinatory completeness, this is equal to $(\lambda^* f.\balpha\circ(\BB\,f)\circ\bbeta)\,f=
(\bbeta^\bullet\circ\BB\circ(\BB\,\balpha)\circ\BB)\,f$.
Indeed,
$$
\begin{array}{rcl}
(\bbeta^\bullet\circ\BB\circ(\BB\,\balpha)\circ\BB)\,f
&=&
\bbeta^\bullet\,(\BB\,(\BB\,\balpha\,(\BB\,f)))\\
&=&
\BB\,(\BB\,\balpha\,(\BB\,f))\,\bbeta\\
&=&
\balpha\circ(\BB\,f)\circ\bbeta.    
\end{array}
$$
Fortunately, this combinator $\bbeta^\bullet\circ\BB\circ(\BB\,\balpha)\circ\BB$ can be simplified as follows.
\begin{lemma}
$\bbeta^\bullet\circ\BB\circ(\BB\,\balpha)\circ\BB= F(\lam)$ holds.
\end{lemma}
\begin{proof}
The left-hand side is, by unfolding the definition of each combinator, equal to 
$
F((F(\beta)\otimes A);\app;(A\otimes\alpha);(\app\otimes A)).
$
Then
$$
\arraycolsep=3pt
\begin{array}{cll}
&
(F(\beta)\otimes A);\app;(A\otimes\alpha);(\app\otimes A)\\
=&
(\mathsf{cur}(\lam;(\lam\otimes A);(A\otimes\beta))\otimes A);\app;(A\otimes\alpha);(\app\otimes A)\\
=&
\lam;(\lam\otimes A);(A\otimes\beta);(A\otimes\alpha);(\app\otimes A)\\
=&
\lam;(A\otimes{\mathit{Tr}_L}^1_{2,2}(\beta;\alpha))\\
=&
\lam
\end{array}
$$
by Lemma \ref{lem:beta-eta} and Proposition \ref{prop:trace}.
\end{proof}
Thus we have:
\begin{proposition} \label{thm:trace_combinator}
Let $\Tr\in\IA(1,2)$ be $F(\lam)$.
Then, for $a\in\IA(1+m,1+n)$, $\Tr\cdot a={\mathit{Tr}_L}^1_{m,n}(a)\in\IA(m,n)$ holds.
\end{proposition}
So the left trace allows a particularly simple description: just the application of $\Tr$.
On the other hand, we also have the right trace on $\IA$, which however is harder to describe because of the
asymmetry of the left-tensor and right-tensor in internal PROBs.

Somewhat surprisingly, we can recover the duality $\balpha$ and $\bbeta$ from the trace combinator $\Tr$ as
$\alpha$ and $\beta$ are defined by $\app$, $\lam$ and {\em trace} on $\mathbb{C}$.
We will discuss this for general balanced extensional $\BCpmI$-algebras in the next section.

\section{Ribbon Combinatory Algebras}\label{sec:ribbonCA}

\begin{definition}
A {\em traced} extensional $\BCpmI$-algebra is a
balanced extensional $\BCpmI$-algebra with an element $\Tr$ satisfying the following axioms.
\setlength{\FrameSep}{4pt}
\begin{oframed}
$$\!\!\!
\begin{array}{rcll}
(\CC^\star\,\BB\,\Tr)\circ\BB &=& (\BB\,\Tr)\circ\BB^2 \!\!\!\!\!& \mathrm{(Left~ Tightening)}\\
(\BB\,\Tr)\circ(\CC^\star\,\BB\,\BB)\circ\BB &=& \BB\circ\Tr & \mathrm{(Right~ Tightening)}\!\!\\
\Tr\,\CC^\pm &=& \btheta^\pm
& \mathrm{(Yanking)}
\\
\!\!\Tr^2\circ(\CC^\star\,\BB\,\CC^-)\circ(\BB\,\CC^+) &=& \Tr^2 & \mathrm{(Exchange)}\\
\end{array}
$$ 
\end{oframed}
\end{definition}
(Left Tightening),  
saying that $\Tr$ is of arity $1\rightarrow 2$, implies
$$
a\circ(\Tr\,f)=((\CC^\star\,\BB\,\Tr)\circ\BB)\,a\,f=((\BB\,\Tr)\circ\BB^2)\,a\,f=\Tr\,((\BB\,a)\circ f).
$$
Similarly, (Right Tightening) says
$$
\Tr\,(f\circ(\BB\,b))=((\BB\,\Tr)\circ(\CC^\star\,\BB\,\BB)\circ\BB)\,f\,b = (\BB\circ\Tr)\,f\,b=(\Tr\,f)\circ b.
$$
Thus they imply 
$\Tr\,((\BB\,a)\circ f\circ (\BB\,b))=a\circ (\Tr\,f)\circ b$.
\begin{center}
\unitlength=.6pt
\begin{picture}(140,65)
\put(15,15){\framebox(110,50){}}
\thicklines
\lineseg{0}{50}{140}{50}
\lineseg{20}{30}{120}{30}
\lineseg{20}{0}{120}{0}
\qbezier(20,30)(0,30)(0,15)\qbezier(0,15)(0,0)(20,0)
\qbezier(120,30)(140,30)(140,15)\qbezier(140,15)(140,0)(120,0)
\rgbfbox{20}{40}{20}{20}{1}{.9}{.9}{\scriptsize$a$}%
\rgbfbox{60}{20}{20}{40}{.5}{1}{.5}{\scriptsize$f$}%
\rgbfbox{100}{40}{20}{20}{.9}{.9}{1}{\scriptsize$b$}%
\end{picture}
\begin{picture}(40,60)
\put(20,30){\makebox(0,0){$=$}}
\end{picture}
\begin{picture}(140,60)
\thicklines
\lineseg{0}{50}{140}{50}
\lineseg{60}{0}{80}{0}
\qbezier(60,30)(40,30)(40,15)\qbezier(40,15)(40,0)(60,0)
\qbezier(80,30)(100,30)(100,15)\qbezier(100,15)(100,0)(80,0)
\rgbfbox{20}{40}{20}{20}{1}{.9}{.9}{\scriptsize$a$}%
\rgbfbox{60}{20}{20}{40}{.5}{1}{.5}{\scriptsize$f$}%
\rgbfbox{100}{40}{20}{20}{.9}{.9}{1}{\scriptsize$b$}%
\end{picture}
\end{center}
(Yanking) would not need any explanation. In the symmetric case, it is simply $\Tr\,\CC=\BB\,\II$.
Finally, (Exchange) says
$$\Tr\,(\Tr\,(\CC^+\circ f\circ \CC^-))=(\Tr^2\circ(\CC^\star\,\BB\,\CC^-)\circ(\BB\,\CC^+))\,f
=\Tr\,(\Tr\,f)$$
\begin{center}
\unitlength=.6pt
\begin{picture}(140,90)
\put(30,15){\framebox(80,70){}}
\put(15,5){\framebox(110,85){}}
\thicklines
\lineseg{0}{70}{140}{70}
\lineseg{50}{50}{90}{50}
\lineseg{50}{30}{90}{30}
\braid{30}{30}{5}{10}
\braidInv{90}{30}{5}{10}
\qbezier(30,30)(20,30)(20,20)\qbezier(20,20)(20,10)(30,10)
\lineseg{30}{10}{110}{10}
\qbezier(110,10)(120,10)(120,20)\qbezier(120,20)(120,30)(110,30)
\lineseg{20}{50}{30}{50}
\lineseg{110}{50}{120}{50}
\qbezier(20,50)(0,50)(0,25)\qbezier(0,25)(0,0)(20,0)
\lineseg{20}{0}{120}{0}
\qbezier(120,0)(140,0)(140,25)\qbezier(140,25)(140,50)(120,50)
\rgbfbox{60}{20}{20}{60}{.5}{1}{.5}{\scriptsize$f$}%
\end{picture}
\begin{picture}(40,90)
\put(20,40){\makebox(0,0){$=$}}
\end{picture}
\begin{picture}(80,90)
\put(15,5){\framebox(50,80){}}
\thicklines
\lineseg{0}{70}{80}{70}
\lineseg{50}{50}{60}{50}
\qbezier(30,30)(20,30)(20,20)\qbezier(20,20)(20,10)(30,10)
\lineseg{30}{10}{50}{10}
\qbezier(50,10)(60,10)(60,20)\qbezier(60,20)(60,30)(50,30)
\lineseg{20}{50}{60}{50}
\qbezier(20,50)(0,50)(0,25)\qbezier(0,25)(0,0)(20,0)
\lineseg{20}{0}{60}{0}
\qbezier(60,0)(80,0)(80,25)\qbezier(80,25)(80,50)(60,50)
\rgbfbox{30}{20}{20}{60}{.5}{1}{.5}{\scriptsize$f$}%
\end{picture}
\end{center}
which is motivated by the axiomatisation of trace in \cite{Has09}. We have 
\begin{proposition}
For a traced balanced extensional $\BCpmI$-algebra $\calA$, its internal PROB $\IA$
is a traced balanced monoidal category.
\end{proposition}
In fact, we will see that $\IA$ is not only traced but also ribbon. 
First, we shall define the notion of duality on combinatory algebras:

\begin{definition}
A {\em duality} in an extensional $\BIdot$-algebra is a pair of elements $\balpha$ and $\bbeta$ satisfying 
the following axioms.
\setlength{\FrameSep}{4pt}
\begin{oframed}   
$$
\begin{array}{rcll}
\balpha^\bullet\circ\BB &=& (\BB\,\balpha)\circ\BB^2 & (\mathrm{Arity}_\alpha)\\
\bbeta^\bullet\circ\BB^3 &=& \BB\,\bbeta & (\mathrm{Arity}_\beta)\\
\balpha\circ(\BB\,\bbeta) &=& \BB\,\II & (\mathrm{Duality}_1)\\
(\BB\,\balpha)\circ\bbeta &=& \BB\,\II & (\mathrm{Duality}_2)\\
\end{array}
$$
\end{oframed}
A {\em ribbon extensional $\BCpmI$-algebra} is a balanced extensional
$\BCpmI$-algebra with a duality satisfying the following axiom.
\setlength{\FrameSep}{4pt}
\begin{oframed}  
$$
\begin{array}{rcll}
\balpha\circ\btheta^\pm &=& \balpha\circ(\BB\,\btheta^\pm) & (\mathrm{Duality}_\theta)
\end{array}
$$
\end{oframed}
\end{definition}
The axiom $(\mathrm{Arity}_\alpha)$ says $\balpha$ is of arity $0\rightarrow2$, while
$(\mathrm{Arity}_\beta)$ says $\bbeta$ is of arity $2\rightarrow0$. 
$(\mathrm{Duality}_1)$ and $(\mathrm{Duality}_2)$ 
state they form a duality.
Finally, $(\mathrm{Duality}_\theta)$ ensures $\theta^*=\theta$ required for ribbon categories.

\begin{proposition}
For a ribbon extensional $\BCpmI$-algebra $\calA$, its internal PROB $\IA$
is a ribbon category.
\end{proposition}

As we already observed for reflexive objects in ribbon categories, we can derive duality from the trace. 
The key insight is that the trace combinator $\Tr$ can be used not only for defining trace but also
as the morphism $\lam\colon A\rightarrow A\otimes A$ in the previous section,
and the lambda-node \makebox(7,5)[c]{$\lamnode$} in the trivalent graph representation in the introduction.
\begin{proposition}
a traced balanced extensional $\BCpmI$-algebra has a duality given by 
$$
\begin{array}{rclcl}
\balpha &=&\Tr\,(\Tr\circ(\BB\,\Tr)\circ(\BB\,\CC^+)\circ\CC^+)&\colon &0\rightarrow2\\
\bbeta &=&\Tr\,(\CC^+\circ(\BB\,\CC^+)\circ(\BB\,\BB)\circ\BB)&\colon &2\rightarrow0
\end{array}
$$
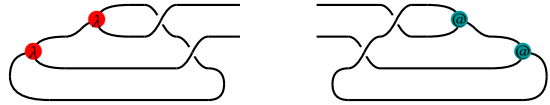
\begin{center}
\unitlength=.6pt
\begin{picture}(145,60)(35,20)
\thicklines
\qbezier(70,60)(75,60)(80,65)\qbezier(80,65)(85,70)(90,70)
\lineseg{70}{40}{140}{40}
\braid{120}{60}{5}{10}
\braid{140}{40}{5}{10}
\lineseg{110}{80}{120}{80}
\lineseg{110}{60}{120}{60}
\lineseg{140}{80}{180}{80}
\lineseg{160}{60}{180}{60}
\qbezier(50,50)(35,50)(35,35)\qbezier(35,35)(35,20)(60,20)
\lineseg{60}{20}{160}{20}
\qbezier(160,20)(170,20)(170,30)\qbezier(170,30)(170,40)(160,40)
\arcLR{70}{40}{-20}{10}
\put(50,50){\lamnode}
\arcLR{110}{60}{-20}{10}
\put(90,70){\lamnode}
\end{picture}
\unitlength=.6pt
\begin{picture}(40,60)
\end{picture}
\unitlength=.6pt
\begin{picture}(145,60)(30,20)
\thicklines
\lineseg{30}{80}{70}{80}
\lineseg{30}{60}{50}{60}
\lineseg{70}{40}{140}{40} 
\braid{50}{40}{5}{10}
\braid{70}{60}{5}{10}
\lineseg{90}{80}{100}{80}
\lineseg{90}{60}{100}{60}
\qbezier(120,70)(125,70)(130,65)\qbezier(130,65)(135,60)(140,60)
\qbezier(50,40)(40,40)(40,30)\qbezier(40,30)(40,20)(50,20)
\lineseg{50}{20}{160}{20}
\qbezier(160,20)(175,20)(175,35)\qbezier(175,35)(175,50)(160,50)
\arcLR{100}{60}{20}{10}
\put(120,70){\appnode}
\arcLR{140}{40}{20}{10}
\put(160,50){\appnode}
\end{picture}
\end{center}
\end{proposition}

On the other hand, a ribbon extensional $\BCpmI$-algebra is traced with $\Tr=(\CC^+\,\BB\,\bbeta)\circ(\BB\,\balpha)\circ\BB$.
Thus we have:
\begin{theorem}
A balanced extensional $\BCpmI$-algebra is traced if and only if it is ribbon.
\end{theorem}

In the sequel, by a {\em ribbon combinatory algebra} we mean a traced balanced extensional $\BCpmI$-algebra, or equivalently a
ribbon extensional $\BCpmI$-algebra. 

\newcommand{\imp}{\!\multimap\!}

\begin{example}\label{ex:trees}
We describe a ribbon combinatory algebra of infinite binary trees labelled by the elements of a group, which was originally 
given as a model of the braided lambda calculus \cite{Has20} in the ribbon category of crossed $G$-sets \cite{Has12}. Fix a group $G=(G,\cdot,e,(\_)^{-1})$ and let
$\mathcal{D}=\{f\colon \{\mathsf{c},\mathsf{d}\}^*\rightarrow G~|~f(w)=f(\mathsf{c}w)\cdot f(\mathsf{d}w)^{-1}\}$.
Define $\multimap\colon \mathcal{D}^2\rightarrow \mathcal{D}$ by 
$(f_1\imp f_2)(\varepsilon)=f_2(\varepsilon)\cdot f_1(\varepsilon)^{-1}$, 
$(f_1\imp f_2)(\mathsf{d}w)=f_1(w)$ and
$(f_1\imp f_2)(\mathsf{c}w)=f_2(w)$. 
There is a $G$-action $\bullet\colon G\times\mathcal{D}\rightarrow\mathcal{D}$ given by
$(g\bullet f)(w)=g\cdot f(w)\cdot g^{-1}$,
and a map $|\_|\colon \mathcal{D}\rightarrow G$ defined by $|f|=f(\varepsilon)$.
We say a subset $a$ of $\mathcal{D}$ is equivariant if, for any $f\in a$ and $g\in G$, $|f|=e$ and $g\bullet f\in a$ hold.
Let $\mathcal{A}$ be the set of equivariant subsets.
$\mathcal{A}$ is a ribbon combinatory algebra with
$$\arraycolsep=3pt
\begin{array}{lcl}
a\,b &=& \{y~|~\exists x~(x\imp y)\in a\,\&\, x\in b\}\\
\BB &=& \{(y\imp z)\imp((x\imp y)\imp(x\imp z))~|~x,y,z\in\mathcal{D}\}\\
\CC^+ &=& \{((|x|\bullet y)\imp (x\imp z))\imp(x\imp(y\imp z))~|~x,y,z\in\mathcal{D}\}\\
\CC^- &=& \{(y\imp ((|y|^{-1}\!\bullet x)\imp z))\imp(x\imp(y\imp z))~|~x,y,z\in\mathcal{D}\}\\
\II &=& \{x\imp x~|~x\in\mathcal{D}\}\\
\btheta^+ &=& \{((|x|\bullet x)\imp y)\imp(x\imp y)~|~x,y\in\mathcal{D}\}\\
\btheta^- &=& \{((|x|^{-1}\!\bullet x)\imp y)\imp(x\imp y)~|~x,y\in\mathcal{D}\}\\
\Tr &=& \{((x\imp y)\imp(x\imp z))\imp(y\imp z)~|~x,y,z\in\mathcal{D}\}\\
\end{array}
$$
When $G$ is abelian, $\mathcal{A}$ is symmetric because $g\bullet f=f$ for any $g\in G$.
\end{example}

We conclude this section by showing that a ribbon combinatory algebra is symmetric (i.e., $\CC^+=\CC^-$ and $\btheta^+=\btheta^-=\II$)
iff $\btheta^+=\II$. Thus having a non-trivial twist is essential.
\begin{proposition}
A ribbon combinatory algebra is symmetric whenever $\btheta^+=\II$ holds.
\end{proposition}
\begin{proof} $\mathbf{p}=\Tr\circ(\BB\,\bbeta)\colon 2\rightarrow 1$ has an inverse $\mathbf{p}^{-1}=(\BB\,\balpha)\circ\BB\colon 1\rightarrow 2$.
The naturality of the twist implies $\mathbf{p}\circ\btheta_2=\btheta_1\circ\mathbf{p}$. On the other hand, $\btheta^+=\II$
implies $\btheta_2=\CC^+\circ\CC^+$ and $\btheta_1=\II$. Hence $\CC^+\circ\CC^+=\II$
by pre-composing $\mathbf{p}^{-1}$, and we have $\CC^+=\CC^-$.
\end{proof}

\section{Ribbon combinatory algebras vs ribbon categories}\label{sec:comparison}

From the results of the previous sections, we have:

\begin{theorem}
Every ribbon combinatory algebra is isomorphic to the one arising from a reflexive object in the ribbon category obtained as the internal PROB. 
\end{theorem}

However, not every reflexive ribbon PROB arises as an internal PROB.
In general, given a ribbon category $\mathbb{C}$ with a reflexive object $A$, the internal PROB $\IA$ of the ribbon combinatory algebra
$\calA=\mathbb{C}(I,A)$ is not isomorphic to the ribbon PROB $\mathbb{C}|_A$, as the simple counterexample of Example \ref{ex:pro} shows.
Nevertheless, the presence of the structure-preserving functor $F\colon \mathbb{C}|_A\rightarrow\IA$ suggests 
they are closely related; we shall look at this situation in more detail.

\subsection{Almost full}

We show that $F$ is {\em almost} full: the map $\mathbb{C}|_A(A^{\otimes m},A^{\otimes n})\rightarrow\IA(m,n)$ on morphisms
is surjective
except in the case of $m=n=0$.

First, the following partial result with codomain $1$ is easily derivable;
it is obtained essentially as a corollary to the equivalence between closed operads and extensional combinatory algebras
(which was stated as an adjunction \cite{Has22} but later turned out to be an equivalence \cite{Has23}).
\begin{lemma}
For any $a\colon m\rightarrow 1$ in $\IA$, there exists  $f\colon A^{\otimes m}\rightarrow A$
in $\bbCA$ such that $a=Ff$ holds.
\end{lemma}
\begin{proof}
By Lemma \ref{lem:compositionality} (6),
any $a\colon m\rightarrow 1$ is equal to 
$b^\bullet\circ\BB^m$ where $b=a\,\II\in\calA$. Let $f=(b\otimes A^{\otimes m});\app_m$. As $F(\app_m)=\BB^m$, we have 
$Ff=Fb\circ F(\app_m)= b^\bullet\circ \BB^m=a$.
\end{proof}

\mbox{}\\
This result extends to the cases with non-zero domain or codomain:
\begin{lemma}
For any $a\colon m\rightarrow 1+n$ in $\IA$, there exists  $f\colon A^{\otimes m}\rightarrow A^{\otimes (1+n)}$
in $\mathbb{C}|_A$ such that $a=Ff$ holds.
\end{lemma}
\begin{proof}
Let $f$ be
$$
A^{\otimes m}\xrightarrow{(a\,\bbeta_n)\otimes A^{\otimes m}\otimes\alpha_n}
A^{\otimes (1+m+2n)}\xrightarrow{\app_{m+n}\otimes A^{\otimes n}}
A^{\otimes (1+ n)}
$$
where $\alpha_n\colon I\rightarrow A^{\otimes(2n)}$ is defined by
$\alpha_1=\alpha$ and $\alpha_{n+1}=\alpha;(A\otimes\alpha_n\otimes A)$.
Similarly define $\beta_n\colon A^{\otimes(2n)}\rightarrow I$, and 
let $\balpha_n=F\alpha_n\colon 0\rightarrow 2n$ and $\bbeta_n=F\beta_n\colon 2n\rightarrow 0$.
Then $F((a\bbeta_n)\otimes A^{\otimes m}\otimes\alpha_n)$ is equal to
$$
F(A^{\otimes m}\otimes\alpha_n)\circ F((a\bbeta_n)\otimes A^{\otimes(m+2n)})
=
(\BB^m\,\balpha_n)\circ(a\,\bbeta_n)^\bullet
$$
while $F(\app_{m+n}\otimes A^{\otimes n})=\BB^{m+n}$,
and we have
$$
\begin{array}{rcll}
Ff
&=&
(\BB^m\,\balpha_n)\circ(a\,\bbeta_n)^\bullet\circ\BB^{m+n}
\\
&=&
(\BB^m\,\balpha_n)\circ\bbeta_n^\bullet\circ a^\bullet\circ\BB^{1+m+n}
\\
&=&
(\BB^m\,\balpha_n)\circ\bbeta_n^\bullet\circ (\BB\,a)\circ\BB^{1+2n}
& (a\colon m\rightarrow n+1)
\\
&=&
(\BB^m\,\balpha_n)\circ a\circ\bbeta_n^\bullet\circ\BB^{1+2n}
& (\bbeta_n^\bullet\colon 0\rightarrow1)
\\
&=&
a\circ(\BB^{n+1}\,\balpha_n)\circ\bbeta_n^\bullet\circ\BB^{1+2n}
& (a\colon m\rightarrow n+1)
\\
&=&
a\circ(\BB^{n+1}\,\balpha_n)\circ (\BB\,\bbeta_n)
& (\bbeta_n\colon 2n\rightarrow 0)
\\
&=&
a\circ\BB((\BB^n\,\balpha_n)\circ\bbeta_n)
& \!\!\!\!\!((\BB u)\circ(\BB v)=\BB\,(u\circ v))
\\
&=&
a\circ(\BB\,\II) 
& ((\BB^n\,\balpha_n)\circ\bbeta_n=\II)
\\
&=&
a.
\end{array}
$$
\end{proof}

\begin{lemma}
For 
$a\colon m+1\rightarrow 0$ in $\IA$, there exists  $f\colon A^{\otimes (m+1)}\rightarrow I$
in $\bbCA$ such that $a=Ff$ holds.
\end{lemma}
\begin{proof}
Let 
$b=((\BB^m\,\balpha)\circ a)\,\II= (\BB^m\,\balpha)\,(a\,\II)$ and
$f=(b\otimes A^{\otimes{m+1}});\app_m;\beta$.
Then
$Ff=b^\bullet\circ\BB^m\circ\bbeta=(\BB^m\,\balpha)\circ a\circ\bbeta=
(\BB^m\balpha)\circ(\BB^{m+1}\bbeta)\circ a =
(\BB^m\,(\balpha\circ\BB\,\bbeta))\circ a = a.
$
\end{proof}
So just the case of $0\rightarrow0$ remains.
This seemingly harmless case turns out to be surprisingly difficult:

\begin{lemma}
For any $a\colon 0\rightarrow 0$ in $\IA$, 
there exists
$f\colon I\rightarrow I$ in $\bbCA$ such that 
$\bigcirc; a=Ff$ holds, where $\bigcirc=\mathit{Tr}^A(\mathit{id}_A)\colon I\rightarrow I$.
\end{lemma}
\begin{proof}
Let $f=\mathit{Tr}^A((a\otimes A);\app)\colon I\rightarrow I$.
Since $a\colon 0\rightarrow 0$, $a^\bullet\circ\BB=\BB\,a$ holds. This implies
$$(A\otimes a\otimes A);(A\otimes\app);\app=(a\otimes\app);\app\colon A\otimes A\rightarrow A$$
\begin{center}
\unitlength=.7pt
\begin{picture}(140,55)
\thicklines
\lineseg{0}{50}{60}{50}
\lineseg{40}{30}{60}{30}
\lineseg{0}{10}{100}{10}
\qbezier(80,40)(90,40)(90,35)\qbezier(90,35)(90,30)(100,30)
\lineseg{120}{20}{140}{20}
\rgbfbox{20}{20}{20}{20}{1}{1}{.5}{\scriptsize$a$}%
\rgbfbox{60}{30}{20}{20}{.8}{.8}{1}{\tiny$\app$}%
\rgbfbox{100}{10}{20}{20}{.8}{.8}{1}{\tiny$\app$}%
\end{picture}
\begin{picture}(40,55)
\put(20,30){\makebox(0,0){$=$}}
\end{picture}
\begin{picture}(100,55)(40,0)
\thicklines
\lineseg{40}{50}{60}{50}
\lineseg{40}{30}{60}{30}
\lineseg{80}{10}{100}{10}
\qbezier(80,40)(90,40)(90,35)\qbezier(90,35)(90,30)(100,30)
\lineseg{120}{20}{140}{20}
\rgbfbox{60}{0}{20}{20}{1}{1}{.5}{\scriptsize$a$}%
\rgbfbox{60}{30}{20}{20}{.8}{.8}{1}{\tiny$\app$}%
\rgbfbox{100}{10}{20}{20}{.8}{.8}{1}{\tiny$\app$}%
\end{picture}
\end{center}
in $\bbCA$. By pre-composing $\lam\colon A\rightarrow A\otimes A$ to both sides of the equation and by applying Proposition \ref{prop:trace}, we have
$$A\otimes(\mathit{Tr}^A((a\otimes A);\app)=((a\otimes A);\app)\otimes(\mathit{Tr}^A(\id_A))\colon A\rightarrow A.$$
Hence $A\otimes f=((a\otimes A);\app)\otimes\bigcirc\colon A\rightarrow A$ holds.
By currying this, we have 
$\mathsf{cur}(A\otimes f)=\mathsf{cur}(((a\otimes A);\app)\otimes\bigcirc)$.
The left-hand side is $Ff$, while the right-hand side is $\bigcirc;a$.
\end{proof}

Therefore the idempotent element $\bigcirc$ is the obstacle in showing fullness,
and we only have the following partial result:

\begin{proposition}\label{prop:fullness}
Suppose $\bigcirc=
\mathit{id}_I$.
Then $F$ is full.
\end{proposition}

\subsection{Almost faithful}

For $a\colon m\rightarrow n$, let $\overline{a}\colon A^{\otimes(m+1)}\rightarrow  A^{\otimes(n+1)}$
be
$$
(\lam_n\otimes\theta_{A^{\otimes m}}^2);
(a\otimes\theta^{-1}_A\otimes A^{\otimes(n+ m)});
(\app\otimes \sigma^{-1}_{A^{\otimes m},A^{\otimes n}});
(\mathit{Tr}^A((A\otimes(\sigma^{-1}_{A,A^{\otimes m}};\sigma^{-1}_{A^{\otimes m},A}));(\app_m\otimes A);\sigma_{A,A})\otimes A^{\otimes n})
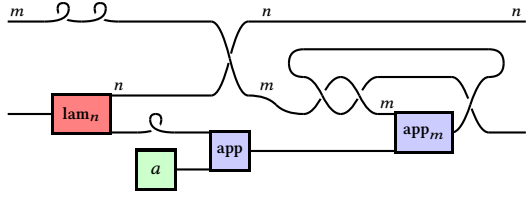
$$
\begin{center}
\unitlength=.7pt
\begin{picture}(280,107)(-40,0)
\thicklines
\put(-35,105){\makebox(0,0){\tiny$m$}}
\lineseg{-40}{100}{-20}{100}
\lineseg{-40}{50}{-15}{50}
\lineseg{20}{100}{70}{100}
\put(20,65){\makebox(0,0){\tiny$n$}}
\lineseg{15}{60}{70}{60}
\lineseg{15}{40}{30}{40}
\put(-10,115){\makebox(0,0){\tiny$m$}}
\twist{-20}{100}
\put(10,115){\makebox(0,0){\tiny$m$}}
\twist{0}{100}
\twistInv{30}{40}
\lineseg{90}{100}{240}{100}
\qbezier(90,60)(100,60)(105,55)\qbezier(105,55)(110,50)(120,50)
\lineseg{50}{40}{70}{40}
\lineseg{50}{20}{70}{20}
\lineseg{220}{40}{240}{40}
\braidInv{70}{60}{5}{20}
\put(100,105){\makebox(0,0){\tiny$n$}}
\put(235,105){\makebox(0,0){\tiny$n$}}
\put(100,65){\makebox(0,0){\tiny$m$}}
\lineseg{90}{30}{170}{30}
\braidInv{120}{50}{5}{10}
\braidInv{140}{50}{5}{10}
\lineseg{160}{70}{200}{70}
\put(165,55){\makebox(0,0){\tiny$m$}}
\lineseg{160}{50}{170}{50}
\braid{200}{40}{5}{15}
\qbezier(120,70)(112,70)(112,77.5)\qbezier(112,77.5)(112,85)(120,85)
\lineseg{120}{85}{220}{85}
\qbezier(220,85)(228,85)(228,77.5)\qbezier(228,77.5)(228,70)(220,70)
\rgbfbox{-15}{40}{30}{20}{1}{.5}{.5}{\tiny$\lam_n$}%
\rgbfbox{30}{10}{20}{20}{.8}{1}{.8}{\scriptsize$a$}%
\rgbfbox{70}{20}{20}{20}{.8}{.8}{1}{\tiny$\app$}%
\rgbfbox{170}{30}{30}{20}{.8}{.8}{1}{\tiny$\app_m$}%
\end{picture}
\end{center}
We show that $\overline{(\_)}$ serves as an "almost" inverse to $F$:
\begin{proposition} \label{prop:overlineF}
For any $f\colon A^{\otimes m}\rightarrow A^{\otimes n}$, $\overline{F(f)}=A\otimes f$ holds.
\end{proposition}
\begin{proof}
From Proposition \ref{prop:trace}, we can derive
$$
(Ff\otimes\theta_A^{-1});\app=
\mathit{Tr}^A((A\otimes\lam_m);(\sigma^{-1}_{A,A});(A\otimes\app_n))
\colon A\rightarrow A
$$
\begin{center}
\unitlength=.7pt
\begin{picture}(100,40)(10,10)
\thicklines
\lineseg{10}{40}{30}{40}
\qbezier(30,40)(37,40)(38,41)\qbezier(42,44)(47,50)(40,50)
\qbezier(40,50)(33,50)(40,42)\qbezier(40,42)(43,40)(50,40)
\lineseg{50}{40}{70}{40}
\lineseg{50}{20}{70}{20}
\lineseg{90}{30}{110}{30}
\rgbfbox{30}{10}{20}{20}{.8}{1}{.8}{\scriptsize$Ff$}%
\rgbfbox{70}{20}{20}{20}{.8}{.8}{1}{\tiny$\app$}%
\end{picture}
\begin{picture}(20,40)
\put(10,20){\makebox(0,0){$=$}}
\end{picture}
\begin{picture}(180,60)(0,0)
\thicklines
\qbezier(30,30)(15,30)(15,45)\qbezier(15,45)(15,60)(30,60)
\lineseg{30}{60}{150}{60}
\qbezier(150,60)(165,60)(165,45)\qbezier(165,45)(165,30)(150,30)
\lineseg{0}{0}{80}{0}
\lineseg{100}{0}{180}{0}
\put(70,45){\makebox(0,0){\tiny$m$}}
\put(110,45){\makebox(0,0){\tiny$n$}}
\lineseg{60}{40}{120}{40}
\lineseg{60}{20}{80}{20}
\lineseg{100}{20}{120}{20}
\braidInv{80}{0}{5}{10}
\rgbfbox{30}{20}{30}{20}{1}{.5}{.5}{\tiny$\lam_m$}%
\rgbfbox{80}{30}{20}{20}{.5}{1}{.5}{\scriptsize$f$}%
\rgbfbox{120}{20}{30}{20}{.8}{.8}{1}{\tiny$\app_n$}%
\end{picture}
\end{center}
Using Corollary \ref{cor:trace} we have
\begin{center}
\unitlength=.7pt
\begin{picture}(260,60)(-180,0)
\thicklines
\lineseg{-180}{50}{-150}{50}
\put(-110,65){\makebox(0,0){\tiny$n$}}
\lineseg{-120}{60}{80}{60}
\lineseg{15}{40}{30}{40}
\qbezier(-120,40)(-110,40)(-95,20)\qbezier(-95,20)(-80,0)(-70,0)
\qbezier(-70,30)(-85,30)(-85,42)\qbezier(-85,42)(-85,54)(-70,54)
\lineseg{-70}{54}{50}{54}
\qbezier(50,54)(65,54)(65,42)\qbezier(65,42)(65,30)(50,30)
\lineseg{-70}{0}{-20}{0}
\lineseg{0}{0}{80}{0}
\put(-30,45){\makebox(0,0){\tiny$m$}}
\put(10,45){\makebox(0,0){\tiny$n$}}
\lineseg{-40}{40}{40}{40}
\lineseg{-40}{20}{-20}{20}
\lineseg{0}{20}{20}{20}
\qbezier(-20,0)(-15,0)(-10,10)\qbezier(-10,10)(-5,20)(0,20)
\put(-10,10){\makebox(0,0){\color[rgb]{1,1,1}$\bullet$}}
\qbezier(-20,20)(-15,20)(-10,10)\qbezier(-10,10)(-5,0)(0,0)
\rgbfbox{-70}{20}{30}{20}{1}{.5}{.5}{\tiny$\lam_m$}%
\rgbfbox{-20}{27}{20}{20}{.5}{1}{.5}{\scriptsize$f$}%
\rgbfbox{20}{20}{30}{20}{.8}{.8}{1}{\tiny$\app_n$}%
\rgbfbox{-150}{40}{30}{20}{1}{.5}{.5}{\tiny$\lam_n$}%
\end{picture}
\begin{picture}(20,55)
\put(10,30){\makebox(0,0){$=$}}
\end{picture}
\begin{picture}(180,55)
\thicklines
\qbezier(20,30)(5,30)(5,45)\qbezier(5,45)(5,60)(20,60)
\lineseg{20}{60}{90}{60}
\qbezier(90,60)(100,60)(100,50)\qbezier(100,50)(100,40)(90,40)
\lineseg{0}{0}{50}{0}
\qbezier(50,0)(55,0)(60,10)\qbezier(60,10)(65,20)(70,20)
\put(60,10){\makebox(0,0){\color[rgb]{1,1,1}$\bullet$}}
\qbezier(50,20)(55,20)(60,10)\qbezier(60,10)(65,0)(70,0)
\lineseg{50}{40}{70}{40}
\qbezier(70,40)(75,40)(80,30)\qbezier(80,30)(85,20)(90,20)
\put(80,30){\makebox(0,0){\color[rgb]{1,1,1}$\bullet$}}
\qbezier(70,20)(75,20)(80,30)\qbezier(80,30)(85,40)(90,40)
\lineseg{70}{0}{180}{0}
\lineseg{90}{20}{110}{20}
\lineseg{130}{20}{150}{20}
\lineseg{170}{20}{180}{20}
\put(160,35){\makebox(0,0){\tiny$n$}}
\qbezier(150,20)(157,20)(158,21)\qbezier(162,24)(167,30)(160,30)
\qbezier(160,30)(153,30)(160,22)\qbezier(160,22)(163,20)(170,20)
\rgbfbox{20}{20}{30}{20}{1}{.5}{.5}{\tiny$\lam_m$}%
\rgbfbox{110}{10}{20}{20}{.5}{1}{.5}{\scriptsize$f$}%
\end{picture}
\end{center}
Hence $\overline{F(f)}$ is equal to
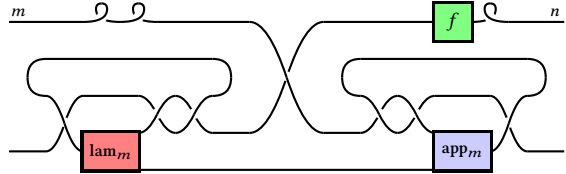
\begin{center}
\unitlength=.7pt
\begin{picture}(300,80)
\thicklines
\put(5,85){\makebox(0,0){\tiny$m$}}
\put(295,85){\makebox(0,0){\tiny$n$}}
\qbezier(20,40)(10,40)(10,50)\qbezier(10,50)(10,60)(20,60)
\lineseg{20}{60}{110}{60}
\qbezier(110,60)(120,60)(120,50)\qbezier(120,50)(120,40)(110,40)
\qbezier(190,40)(180,40)(180,50)\qbezier(180,50)(180,60)(190,60)
\lineseg{190}{60}{280}{60}
\qbezier(280,60)(290,60)(290,50)\qbezier(290,50)(290,40)(280,40)
\twist{40}{80}
\twist{60}{80}
\twistInv{250}{80}
\braidInv{20}{10}{5}{15}
\braid{70}{20}{5}{10}
\braid{90}{20}{5}{10}
\braidInv{130}{20}{10}{30}
\braidInv{190}{20}{5}{10}
\braidInv{210}{20}{5}{10}
\braid{260}{10}{5}{15}
\lineseg{0}{80}{40}{80}
\lineseg{80}{80}{130}{80}
\lineseg{0}{10}{20}{10}
\lineseg{40}{40}{70}{40}
\lineseg{70}{0}{230}{0}
\lineseg{110}{20}{130}{20}
\lineseg{170}{80}{230}{80}
\lineseg{270}{80}{300}{80}
\lineseg{170}{20}{190}{20}
\lineseg{230}{40}{260}{40}
\lineseg{280}{10}{300}{10}
\rgbfbox{40}{0}{30}{20}{1}{.5}{.5}{\tiny$\lam_m$}%
\rgbfbox{230}{0}{30}{20}{.8}{.8}{1}{\tiny$\app_m$}%
\rgbfbox{230}{70}{20}{20}{.5}{1}{.5}{\scriptsize$f$}%
\end{picture}
\end{center}
By applying Corollary \ref{cor:trace} again and by some simplification,
this turns out to be equal to $A\otimes f$.
\end{proof}

\begin{corollary}
The map $\mathbb{C}|_A(A^{\otimes m},A^{\otimes n})\rightarrow\IA(m,n)$ on morphisms is bijective whenever $m\not=0$ or $n\not=0$. 
\end{corollary}
\begin{proof}
Consider $f,g\colon A^{\otimes m}\rightarrow A^{\otimes n}$ with $m\not=0$ and assume $F(f)=F(g)$. Then  
$A\otimes f=\overline{F(f)}=\overline{F(g)}=A\otimes g$, and
$f=(\theta_A^{-1}\otimes A^{m-1});\mathit{Tr}_L((\sigma_{A,A}\otimes A^{m-1});(A\otimes f))=
(\theta_A^{-1}\otimes A^{m-1});\mathit{Tr}_L((\sigma_{A,A}\otimes A^{m-1});(A\otimes g))=g$.
The case of $n\not=0$ is similar.    
\end{proof}
So again the case of $0\rightarrow0$ is the obstacle. Proposition \ref{prop:overlineF} implies
\begin{proposition} \label{prop:faithfulness}
If the functor $A\otimes(\_)\colon \mathbb{C}|_A\rightarrow\mathbb{C}|_A$ is faithful,  $F$ is faithful.
\end{proposition}
When $\bigcirc=\mathit{Tr}^A(\mathit{id}_A)=\mathit{id}_I$,
$A\otimes(-)$ is faithful, as $A\otimes f=A\otimes g$ implies
$
f=\bigcirc\otimes f=\mathit{Tr}_L(A\otimes f)=\mathit{Tr}_L(A\otimes g)=\bigcirc\otimes g=g.
$
Hence we have, by Proposition \ref{prop:fullness}
and \ref{prop:faithfulness},

\begin{theorem} \label{thm:equivalence}
Suppose $\bigcirc=\mathit{id}_I$.
Then $F$ is full and faithful, and gives an equivalence between 
$\mathbb{C}|_A$ and $\IA$.
\end{theorem}
We do not know if $\bigcirc=\mathit{id}_I$ is a necessary condition for the equivalence.

\section{Related Work}\label{sec:related}

\subsection{Curry's Combinatory Logic}
In 1931, Curry \cite{Cur31}
introduced the notion of {\em order} and {\em degree} for {\em regular combinators}, which are combinators corresponding to closed lambda terms of the form
$\lambda fx_1\dots x_m.f\,M_1\,\dots\,M_n$
with no lambda abstraction nor $f$ in $M_i$'s.
In his terminology, such a combinator is {\em of order $m$ and degree $n$};
it is of arity $m\rightarrow n$ in our context.
In \cite{Cur31}, Curry has shown that a regular combinator $U$ is of order $m$ and degree $n$ if and only if it satisfies
$$\CC\,(\BB^{m+1})\,U=(\BB\,U)\circ\BB^n$$
which agrees with the definition of the arity $m\rightarrow n$
(Definition \ref{def:arity}) 
$$U^\bullet\circ\BB^{m+1}=(\BB\,U)\circ\BB^n$$
and also the characterisation 
$$(\CC\,\BB\,a)\circ\BB^m=(\BB\,a)\circ\BB^n$$
in Section \ref{subsec:braidedCA}, 
because $\CC\,a\,b=b^\bullet\circ a$ holds in extensional $\BCI$-algebras.
Curry also has shown items 1 to 5 of Lemma \ref{lem:compositionality} for
regular combinators, thus he was aware of the compositional nature of the 
arity (order and degree). Moreover, he has given an axiomatisation of combinatory logic \cite{Cur30a},
reproduced in Figure \ref{fig:curry} (except the standard axioms $\BB\,a\,b\,c=a\,(b\,c)$, $\CC\,a\,b\,c=a\,c\,b$, $\WW\,a\,b=a\,b\,b$ and $\KK\,a\,b=a$;
$\II$ was $\WW\,\KK$ in \cite{Cur30a}),
which is very close to Hasegawa's axioms 
of extensional $\mathbf{BCIWK}$-algebras in \cite{Has22}.
The commutative axioms state the arity  $\BB\colon 2\rightarrow1$, $\CC\colon 2\rightarrow2$, $\WW\colon 1\rightarrow2$, $\KK\colon 1\rightarrow0$ and
$\II\colon 0\rightarrow0$. Axioms $\mathit(CC)_{1,2}$ are the Coxter relations of $\CC$, and $\mathit(BC)$,$\mathit(CW)$ and $\mathit{(CK)}$
are the naturality of $\CC$ with respect to $\BB$, $\WW$ and $\KK$. 
$\mathit(WC)$, $\mathit(WW)$ and $\mathit(WK)$ say $(1,\WW,\KK)$ forms a commutative co-monoid, and 
$\BB$ is a co-monoid homomorphism by $\mathit(BW)$ and $\mathit(BK)$.
\begin{figure}
$$
\begin{array}{rcll}
\multicolumn{4}{l}{Commutative~Axioms~(Kommutative~Axiome)}\\
\CC\,(\BB^3)\,\BB &=& (\BB\,\BB)\circ\BB & (\mathit{Ax.~B.})\\
\CC\,(\BB^3)\,\CC &=& (\BB\,\CC)\circ\BB^2 & (\mathit{Ax.~C.})\\
\CC\,(\BB^2)\,\WW &=& (\BB\,\WW)\circ\BB^2 & (\mathit{Ax.~W.})\\
\CC\,(\BB^2)\,\KK &=& (\BB\,\KK)\circ\II & (\mathit{Ax.~K.})\\
\CC\,\BB\,\II &=& (\BB\,\II)\circ\II & (\mathit{Ax.~I.})\\
\multicolumn{4}{l}{Transmutative~Axioms~(Transmutative~Axiome)}\\
\BB\circ\CC &=& (\BB\,\CC)\circ\CC\circ(\BB\,\BB) & (\mathit{Ax.~(BC).})\\
\BB\circ\WW &=& 
\multicolumn{2}{l}{(\BB\,\WW)\circ\WW\circ(\BB\,\CC)\circ(\BB^2\,\BB)\circ\BB}
\\
&&& (\mathit{Ax.~(BW).})\\
\BB\circ\KK &=& \KK\circ\KK & (\mathit{Ax.~(BK).})\\
\CC\circ\CC &=& \BB^2\,\II & (\mathit{Ax.~(CC)_1.})\\
(\BB\,\CC)\circ\CC\circ(\BB\,\CC) &=& \CC\circ(\BB\,\CC)\circ\CC & (\mathit{Ax.~(CC)_2.})\\
\CC\circ\WW &=& (\BB\,\WW)\circ\CC\circ(\BB\,\CC) & (\mathit{Ax.~(CW).})\\
\CC\circ\KK &=& \BB\,\KK & (\mathit{Ax.~(CK).})\\
\WW\circ\CC &=& \WW & (\mathit{Ax.~(WC).})\\
\WW\circ\WW &=& \WW\circ(\BB\,\WW) & (\mathit{Ax.~(WW).})\\
\WW\circ\KK &=& \BB\,\II & (\mathit{Ax.~(WK).})\\
\multicolumn{4}{l}{Other~Axiom~(Andere~Axiome)}\\
\BB\,\II &=& \II & (\mathit{Ax.~I_2.})
\end{array}
$$
\Description[Curry's axioms of combinatory logic]{Curry's axioms of combinatory logic}
\caption{Curry's axioms of combinatory logic}
\label{fig:curry}
\end{figure}
For instance, $\mathit(BW)$ can be depicted as follows (with $\WW=$~
{\unitlength=.4pt
\begin{picture}(30,20)(0,5)
\thicklines
\lineseg{0}{10}{10}{10}
\arcLR{30}{0}{-20}{10} 
\put(10,10){\makebox(0,0){$\bullet$}}
\end{picture}}~).
\begin{center}
\unitlength=.7pt
\begin{picture}(60,73)(0,2)
\thicklines
\arcLR{0}{20}{20}{10}    
\arcLR{60}{20}{-20}{10}
\lineseg{20}{30}{40}{30}
\put(20,30){\makebox(0,0){$\appnode$}}
\put(40,30){\makebox(0,0){$\bullet$}}
\put(30,70){\makebox(0,0){\small$\BB\circ\WW$}}
\end{picture}
\begin{picture}(50,60)(0,2)
\put(25,30){\makebox(0,0){$=$}}
\end{picture}
\begin{picture}(140,60)(0,2)
\thicklines
\lineseg{0}{50}{20}{50}
\arcLR{40}{40}{-20}{10}
\lineseg{0}{10}{40}{10}
\lineseg{40}{60}{80}{60}
\lineseg{40}{40}{60}{40}
\symm{60}{20}{5}{10}
\arcLR{60}{0}{-20}{10}
\lineseg{60}{0}{100}{0}
\lineseg{80}{20}{100}{20}
\arcLR{80}{40}{20}{10}
\lineseg{100}{50}{140}{50}
\arcLR{100}{0}{20}{10}
\lineseg{120}{10}{140}{10}
\put(100,50){\makebox(0,0){$\appnode$}}
\put(120,10){\makebox(0,0){$\appnode$}}
\put(20,50){\makebox(0,0){$\bullet$}}
\put(40,10){\makebox(0,0){$\bullet$}}
\put(70,70){\makebox(0,0){\small$(\BB\,\WW)\circ\WW\circ(\BB\,\CC)\circ(\BB^2\,\BB)\circ\BB$}}
\end{picture}
\end{center}
In a sense, Curry identified the PROP of regular combinators ---
long before monoidal categories were invented. 
It seems  
that, throughout his early work on combinatory logic \cite{Cur30a,Cur30b,Cur30c,Cur31}, Curry had certain geometric intuition, 
which we believe was not far from the string diagrams of internal PROs.
However, non-regular combinators such as $a^\bullet=\CC\,\II\,a=\lambda f.f\,a\colon 0\rightarrow1$ were not included in Curry's 
study, and this is the crucial difference from Hasegawa's approach. 

Unfortunately, Curry never used these notions of order and degree in his later work
(except the book \cite{CF58} which contains slightly different definitions of order and degree used for unrelated purposes), 
and we are not aware of any subsequent work developing Curry's early ideas further.

\subsection{Semantics of lambda calculi and operads}

Our work is closely related to Hyland's approach using semi-closed operads 
as semantic models of the $\lambda_\beta$-calculus \cite{Hyl17},
though its emphasis was more on algebraic theory rather than geometric reading.
Hasegawa's approach using closed operads \cite{Has22} is largely the same as ours using reflexive PROs;
the only difference is that \cite{Has22} focuses on the internal operad $\IA(\_,1)$ rather than whole of the internal PRO $\IA$.
We looked at PROs more seriously because they are more suitable for studying trace and duality.
Although we considered only the extensional case in this paper, Hasegawa \cite{Has23} recently generalised his approach to the non-extensional case,
which can be applied to ribbon combinatory algebras.

\subsection{Knot theory and semantics}

Classical results on braids \cite{Art25,KT08} and quantum topology \cite{Tur94,Yet01} are of great importance in our approach.
Also, studies on knotted graphs \cite{Kau89,Yet89} are relevant to our study. For instance, as noted in \cite{Has22},
the Reidemeister Move IV for trivalent graphs
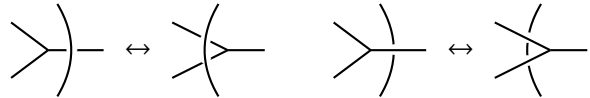
\cite{Kau89} 
\begin{center}
\unitlength=.7pt
\begin{picture}(50,50)\thicklines
\lineseg{0}{40}{20}{25}
\lineseg{0}{10}{20}{25}
\lineseg{20}{25}{50}{25}
\put(33,25){\makebox(0,0){\color[rgb]{1,1,1}\large$\bullet$}}
\qbezier(25,0)(40,25)(25,50)
\end{picture}
\begin{picture}(30,50)
\put(15,25){\makebox(0,0){$\leftrightarrow$}}    
\end{picture}
\begin{picture}(50,50)\thicklines
\lineseg{0}{40}{30}{25}
\lineseg{0}{10}{30}{25}
\lineseg{30}{25}{50}{25}
\put(18,30){\makebox(0,0){\color[rgb]{1,1,1}\large$\bullet$}}
\put(18,20){\makebox(0,0){\color[rgb]{1,1,1}\large$\bullet$}}
\qbezier(25,0)(10,25)(25,50)
\end{picture}
\begin{picture}(30,50)
\end{picture}
\begin{picture}(50,50)\thicklines
\qbezier(25,0)(40,25)(25,50)
\put(33,25){\makebox(0,0){\color[rgb]{1,1,1}\large$\bullet$}}
\lineseg{0}{40}{20}{25}
\lineseg{0}{10}{20}{25}
\lineseg{20}{25}{50}{25}
\end{picture}
\begin{picture}(30,50)
\put(15,25){\makebox(0,0){$\leftrightarrow$}}    
\end{picture}
\begin{picture}(50,50)\thicklines
\qbezier(25,0)(10,25)(25,50)
\put(18,31){\makebox(0,0){\color[rgb]{1,1,1}\large$\bullet$}}
\put(18,19){\makebox(0,0){\color[rgb]{1,1,1}\large$\bullet$}}
\lineseg{0}{40}{30}{25}
\lineseg{0}{10}{30}{25}
\lineseg{30}{25}{50}{25}
\end{picture}
\end{center}
corresponds to the naturality of the braid with respect to $\BB\colon 2\rightarrow1$
(depicted in Sec. \ref{subsec:braidedCA}). 

A braided version of relational semantics is developed in \cite{Has12} using results from quantum topology,
from which the ribbon combinatory algebra in Example \ref{ex:trees} is derived.

\section{Summary and future work}\label{sec:concl}

\paragraph{Summary}
We have seen that a reflexive object in a ribbon category gives rise to
a ribbon combinatory algebra, 
whose internal PROB is a ribbon category which is {\em almost} equivalent to
the appropriate subcategory of the original ribbon category.

This situation can be summarised as follows:
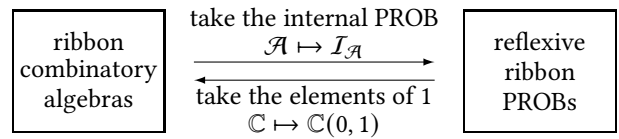
\begin{center}
\unitlength=.40mm
\begin{picture}(200,50)
\put(0,5){\framebox(50,40){}}
\put(150,5){\framebox(50,40){}}
\put(25,35){\makebox(0,0){ribbon}}
\put(25,25){\makebox(0,0){combinatory}}
\put(25,15){\makebox(0,0){algebras}}
\put(175,35){\makebox(0,0){reflexive}}
\put(175,25){\makebox(0,0){ribbon}}
\put(175,15){\makebox(0,0){PROBs}}
\put(60,28){\vector(1,0){80}}
\put(140,22){\vector(-1,0){80}}
\put(100,43){\makebox(0,0){take the internal PROB}}
\put(100,33){\makebox(0,0){$\calA\mapsto\IA$}}
\put(100,17){\makebox(0,0){take the elements of $1$}}
\put(100,7){\makebox(0,0){$\mathbb{C}\mapsto\mathbb{C}(0,1)$}}
\end{picture}
\end{center}
where a reflexive ribbon PROB means a PROB $\mathbb{C}$ 
which is a ribbon category and $1$ is reflexive ($\mathbb{C}((\_)+1,1)\cong\mathbb{C}(-,1)$).

One might expect that this would form an equivalence of categories. However,
while $\IA(0,1)=\calA^\bullet\cong\calA$ holds for
a ribbon combinatory algebra $\calA$,  there is a nasty gap
between $\mathbb{C}$ and $\mathcal{I}_{\mathbb{C}(0,1)}$ for a ribbon PROB $\mathbb{C}$ as the functor
$F\colon \mathbb{C}\rightarrow\mathcal{I}_{\mathbb{C}(0,1)}$ does not have to give an equivalence. Nevertheless,
under a mild condition $\bigcirc=\mathit{id}_I$, $F$ indeed gives an equivalence
(Theorem \ref{thm:equivalence}).

We have given an axiomatisation of ribbon combinatory algebras in terms of the trace combinator $\Tr$, and shown that it is equivalent to
another axiomatisation using the combinators $\balpha$ and $\bbeta$ for the duality. 
Thus, in the untyped (reflexive) extensional setting, to give a trace combinator is 
to give a ribbon structure.

\paragraph{Future work} 
We have studied combinatory algebras whose internal
PROB (PROP) is ribbon ({compact closed} when symmetric). 
A related question is
to characterise combinatory algebras whose internal PROB (PROP) is {\em traced}
(but not necessarily ribbon or compact closed). 
For example, let $\mathbb{D}$ be a cartesian closed category with 
a Conway fixed-point operator and a reflexive object $D$ --- there are plenty of such cases found in domain theory, and a Conway operator
determines a trace \cite{Has99,Has09}.
The global elements of $D$
form an extensional $\BCI$-algebra (in fact an extensional $\mathbf{SK}$-algebra). Its internal PROP is a cartesian category with a Conway operator, hence traced; but it is not compact closed. 

Another important question is to give finer examples of ribbon combinatory algebras. The crossed $G$-set models
(Example \ref{ex:trees}) give ribbon combinatory algebras with non-trivial braids and twists, but they do not provide 
interesting invariants of tangles. We expect to obtain better examples by applying techniques from knot theory.
For instance, the recent work relating knots and links to Thompson's group $F$ \cite{Jon17,Jon19} might be relevant, as $F$ is the group of automorphisms on an object $A$  with an isomorphism
$A\otimes A\cong A$ in a monoidal category \cite{FL10}.

\begin{acks}
The first author is supported by 
\grantsponsor{key}{JSPS KAKENHI}{url}
Grant No. \grantnum{key}{21K11753} and 
\grantnum{key}{24K14822}. 
The second author is supported by
\grantsponsor{hasuo}{JST ERATO HASUO Metamathematics for Systems Design Project}{url}
(No. \grantnum{hasuo}{JPMJER1603}).
\end{acks}

\bibliographystyle{ACM-Reference-Format}
\bibliography{braid}

\end{document}

%% file: defs.tex
%
%
\newdimen\cboxwidth
\cboxwidth=\unitlength
\newdimen\cboxheight
\cboxheight=\unitlength
\newcommand{\rgbfbox}[8]{%
\cboxwidth=\unitlength
\cboxheight=\unitlength
\multiply \cboxwidth by #3
\multiply \cboxheight by #4
{\color[rgb]{#5,#6,#7}\put(#1,#2){\rule{\cboxwidth}{\cboxheight}}}
\put(#1,#2){\framebox(#3,#4){#8}}
\cboxwidth=\unitlength
\cboxheight=\unitlength
}
%
%
%
\newcommand{\colorfbox}[6]{%
\cboxwidth=\unitlength
\cboxheight=\unitlength
\multiply \cboxwidth by #3
\multiply \cboxheight by #4
{\color{#5}\put(#1,#2){\rule{\cboxwidth}{\cboxheight}}}
\put(#1,#2){\framebox(#3,#4){#6}}
\cboxwidth=\unitlength
\cboxheight=\unitlength
}
%
%
%
\newcommand{\rgbbox}[8]{%
\cboxwidth=\unitlength
\cboxheight=\unitlength
\multiply \cboxwidth by #3
\multiply \cboxheight by #4
{\color[rgb]{#5,#6,#7}\put(#1,#2){\rule{\cboxwidth}{\cboxheight}}}
\put(#1,#2){\makebox(#3,#4){#8}}
\cboxwidth=\unitlength
\cboxheight=\unitlength
}
%
%

\newcommand{\app}{\mathbf{app}}
\newcommand{\lam}{\mathbf{lam}}
\newcommand{\BB}{\mathit{\boldsymbol{B}}}
\newcommand{\CC}{\mathit{\boldsymbol{C}}}
\newcommand{\DD}{\mathit{\boldsymbol{D}}}
\newcommand{\II}{\mathit{\boldsymbol{I}}}
\renewcommand{\SS}{\mathit{\boldsymbol{S}}}
\newcommand{\KK}{\mathit{\boldsymbol{K}}}
\newcommand{\WW}{\mathit{\boldsymbol{W}}}
\newcommand{\YY}{\mathit{\boldsymbol{Y}}}
\newcommand{\fix}{\mathbf{fix}}
\newcommand{\BIdot}{\mathbf{BI}(\_)^\bullet}
\newcommand{\BCI}{\mathbf{BCI}}
\newcommand{\BCpmI}{\mathbf{BC^\pm I}}
\newcommand{\SK}{\mathbf{SK}}
\newcommand{\BCIWK}{\mathbf{BCIWK}}
\newcommand{\id}{\mathit{id}}
\newcommand{\mycomment}[1]{}
\newcommand{\CA}{\mathcal{C}_\mathcal{A}}
\newcommand{\bbCA}{\mathbb{C}|_A}
\newcommand{\IA}{\mathcal{I}_\mathcal{A}}
\newcommand{\calA}{\mathcal{A}}
\newcommand{\calC}{\mathcal{C}}
\newcommand{\cox}{\mathit{cox}}
\newcommand{\sem}[1]{[\![#1]\!]}
\newcommand{\Tr}{\mathit{\boldsymbol{T\!r}}}
\newcommand{\lineseg}[4]{\qbezier(#1,#2)(#1,#2)(#3,#4)}
\newcommand{\balpha}{{\boldsymbol\alpha}}
\newcommand{\bbeta}{{\boldsymbol\beta}}
\newcommand{\bsigma}{{\boldsymbol\sigma}}
\newcommand{\btheta}{{\boldsymbol\theta}}

\newcommand{\red}{\color[rgb]{1,0,0}}
\newcounter{braidX}
\newcounter{braidXone}
\newcounter{braidXtwo}
\newcounter{braidXthree}
\newcounter{braidXfour}
\newcounter{braidY}
\newcounter{braidYone}
\newcounter{braidYtwo}
\newcounter{braidH}
\newcounter{braidW}
\newcommand{\braid}[4]{%
\setcounter{braidX}{#1}
\setcounter{braidY}{#2}
\setcounter{braidW}{#3}
\setcounter{braidH}{#4}
\setcounter{braidXone}{\value{braidX}}
\addtocounter{braidXone}{\value{braidW}}
\setcounter{braidXtwo}{\value{braidXone}}
\addtocounter{braidXtwo}{\value{braidW}}
\setcounter{braidXthree}{\value{braidXtwo}}
\addtocounter{braidXthree}{\value{braidW}}
\setcounter{braidXfour}{\value{braidXthree}}
\addtocounter{braidXfour}{\value{braidW}}
\setcounter{braidYone}{\value{braidY}}
\addtocounter{braidYone}{\value{braidH}}
\setcounter{braidYtwo}{\value{braidYone}}
\addtocounter{braidYtwo}{\value{braidH}}
\qbezier(\value{braidX},\value{braidYtwo})(\value{braidXone},\value{braidYtwo})(\value{braidXtwo},\value{braidYone})
\qbezier(\value{braidXtwo},\value{braidYone})(\value{braidXthree},\value{braidY})(\value{braidXfour},\value{braidY})
\put(\value{braidXtwo},\value{braidYone})
{\makebox(0,0){\large\color[rgb]{1,1,1}$\bullet$}}
\qbezier(\value{braidX},\value{braidY})(\value{braidXone},\value{braidY})(\value{braidXtwo},\value{braidYone})
\qbezier(\value{braidXtwo},\value{braidYone})(\value{braidXthree},\value{braidYtwo})(\value{braidXfour},\value{braidYtwo})
}%
\newcommand{\braidInv}[4]{%
\setcounter{braidX}{#1}
\setcounter{braidY}{#2}
\setcounter{braidW}{#3}
\setcounter{braidH}{#4}
\setcounter{braidXone}{\value{braidX}}
\addtocounter{braidXone}{\value{braidW}}
\setcounter{braidXtwo}{\value{braidXone}}
\addtocounter{braidXtwo}{\value{braidW}}
\setcounter{braidXthree}{\value{braidXtwo}}
\addtocounter{braidXthree}{\value{braidW}}
\setcounter{braidXfour}{\value{braidXthree}}
\addtocounter{braidXfour}{\value{braidW}}
\setcounter{braidYone}{\value{braidY}}
\addtocounter{braidYone}{\value{braidH}}
\setcounter{braidYtwo}{\value{braidYone}}
\addtocounter{braidYtwo}{\value{braidH}}
\qbezier(\value{braidX},\value{braidY})(\value{braidXone},\value{braidY})(\value{braidXtwo},\value{braidYone})
\qbezier(\value{braidXtwo},\value{braidYone})(\value{braidXthree},\value{braidYtwo})(\value{braidXfour},\value{braidYtwo})
\put(\value{braidXtwo},\value{braidYone})
{\makebox(0,0){\color[rgb]{1,1,1}$\bullet$}}
\qbezier(\value{braidX},\value{braidYtwo})(\value{braidXone},\value{braidYtwo})(\value{braidXtwo},\value{braidYone})
\qbezier(\value{braidXtwo},\value{braidYone})(\value{braidXthree},\value{braidY})(\value{braidXfour},\value{braidY})
}%
\newcommand{\symm}[4]{%
\setcounter{braidX}{#1}
\setcounter{braidY}{#2}
\setcounter{braidW}{#3}
\setcounter{braidH}{#4}
\setcounter{braidXone}{\value{braidX}}
\addtocounter{braidXone}{\value{braidW}}
\setcounter{braidXtwo}{\value{braidXone}}
\addtocounter{braidXtwo}{\value{braidW}}
\setcounter{braidXthree}{\value{braidXtwo}}
\addtocounter{braidXthree}{\value{braidW}}
\setcounter{braidXfour}{\value{braidXthree}}
\addtocounter{braidXfour}{\value{braidW}}
\setcounter{braidYone}{\value{braidY}}
\addtocounter{braidYone}{\value{braidH}}
\setcounter{braidYtwo}{\value{braidYone}}
\addtocounter{braidYtwo}{\value{braidH}}
\qbezier(\value{braidX},\value{braidY})(\value{braidXone},\value{braidY})(\value{braidXtwo},\value{braidYone})
\qbezier(\value{braidXtwo},\value{braidYone})(\value{braidXthree},\value{braidYtwo})(\value{braidXfour},\value{braidYtwo})
%
%
\qbezier(\value{braidX},\value{braidYtwo})(\value{braidXone},\value{braidYtwo})(\value{braidXtwo},\value{braidYone})
\qbezier(\value{braidXtwo},\value{braidYone})(\value{braidXthree},\value{braidY})(\value{braidXfour},\value{braidY})
}%
\newcounter{twistX}
\newcounter{twistX1}
\newcounter{twistX11}
\newcounter{twistX12}
\newcounter{twistX2}
\newcounter{twistY}
\newcounter{twistY1}
\newcounter{twistY2}
\newcounter{twistY0}
\newcommand{\twist}[2]{%
\setcounter{twistX}{#1}
\setcounter{twistX1}{#1}
\addtocounter{twistX1}{10}
\setcounter{twistX11}{#1}
\addtocounter{twistX11}{7}
\setcounter{twistX12}{#1}
\addtocounter{twistX12}{13}
\setcounter{twistX2}{#1}
\addtocounter{twistX2}{20}
\setcounter{twistY}{#2}
\setcounter{twistY1}{#2}
\addtocounter{twistY1}{7}
\setcounter{twistY2}{#2}
\addtocounter{twistY2}{10}
\setcounter{twistY0}{#2}
\addtocounter{twistY0}{2}
\qbezier(\value{twistX1},\value{twistY2})(\value{twistX11},\value{twistY2})(\value{twistX11},\value{twistY1})
\qbezier(\value{twistX11},\value{twistY1})(\value{twistX11},\value{twistY})(\value{twistX2},\value{twistY})
\put(\value{twistX1},\value{twistY0})
{\makebox(0,0){\color[rgb]{1,1,1}\large$\bullet$}}
\qbezier(\value{twistX},\value{twistY})(\value{twistX12},\value{twistY})(\value{twistX12},\value{twistY1})
\qbezier(\value{twistX12},\value{twistY1})(\value{twistX12},\value{twistY2})(\value{twistX1},\value{twistY2})
}%
\newcommand{\twistInv}[2]{%
\setcounter{twistX}{#1}
\setcounter{twistX1}{#1}
\addtocounter{twistX1}{10}
\setcounter{twistX11}{#1}
\addtocounter{twistX11}{7}
\setcounter{twistX12}{#1}
\addtocounter{twistX12}{13}
\setcounter{twistX2}{#1}
\addtocounter{twistX2}{20}
\setcounter{twistY}{#2}
\setcounter{twistY1}{#2}
\addtocounter{twistY1}{7}
\setcounter{twistY2}{#2}
\addtocounter{twistY2}{10}
\setcounter{twistY0}{#2}
\addtocounter{twistY0}{2}
\qbezier(\value{twistX},\value{twistY})(\value{twistX12},\value{twistY})(\value{twistX12},\value{twistY1})
\qbezier(\value{twistX12},\value{twistY1})(\value{twistX12},\value{twistY2})(\value{twistX1},\value{twistY2})
\put(\value{twistX1},\value{twistY0})
{\makebox(0,0){\color[rgb]{1,1,1}\large$\bullet$}}
\qbezier(\value{twistX1},\value{twistY2})(\value{twistX11},\value{twistY2})(\value{twistX11},\value{twistY1})
\qbezier(\value{twistX11},\value{twistY1})(\value{twistX11},\value{twistY})(\value{twistX2},\value{twistY})
}%
\newcommand{\twistG}[2]{%
\setcounter{twistX}{#1}
\setcounter{twistX1}{#1}
\addtocounter{twistX1}{10}
\setcounter{twistX11}{#1}
\addtocounter{twistX11}{7}
\setcounter{twistX12}{#1}
\addtocounter{twistX12}{13}
\setcounter{twistX2}{#1}
\addtocounter{twistX2}{20}
\setcounter{twistY}{#2}
\setcounter{twistY1}{#2}
\addtocounter{twistY1}{7}
\setcounter{twistY2}{#2}
\addtocounter{twistY2}{10}
\setcounter{twistY0}{#2}
\addtocounter{twistY0}{2}
\qbezier(\value{twistX1},\value{twistY2})(\value{twistX11},\value{twistY2})(\value{twistX11},\value{twistY1})
\qbezier(\value{twistX11},\value{twistY1})(\value{twistX11},\value{twistY})(\value{twistX2},\value{twistY})
\put(\value{twistX1},\value{twistY0})
{\makebox(0,0){\color[rgb]{.91,.91,.91}\large$\bullet$}}
\qbezier(\value{twistX},\value{twistY})(\value{twistX12},\value{twistY})(\value{twistX12},\value{twistY1})
\qbezier(\value{twistX12},\value{twistY1})(\value{twistX12},\value{twistY2})(\value{twistX1},\value{twistY2})
}%
\newcommand{\twistInvG}[2]{%
\setcounter{twistX}{#1}
\setcounter{twistX1}{#1}
\addtocounter{twistX1}{10}
\setcounter{twistX11}{#1}
\addtocounter{twistX11}{7}
\setcounter{twistX12}{#1}
\addtocounter{twistX12}{13}
\setcounter{twistX2}{#1}
\addtocounter{twistX2}{20}
\setcounter{twistY}{#2}
\setcounter{twistY1}{#2}
\addtocounter{twistY1}{7}
\setcounter{twistY2}{#2}
\addtocounter{twistY2}{10}
\setcounter{twistY0}{#2}
\addtocounter{twistY0}{2}
\qbezier(\value{twistX},\value{twistY})(\value{twistX12},\value{twistY})(\value{twistX12},\value{twistY1})
\qbezier(\value{twistX12},\value{twistY1})(\value{twistX12},\value{twistY2})(\value{twistX1},\value{twistY2})
\put(\value{twistX1},\value{twistY0})
{\makebox(0,0){\color[rgb]{.91,.91,.91}\large$\bullet$}}
\qbezier(\value{twistX1},\value{twistY2})(\value{twistX11},\value{twistY2})(\value{twistX11},\value{twistY1})
\qbezier(\value{twistX11},\value{twistY1})(\value{twistX11},\value{twistY})(\value{twistX2},\value{twistY})
}%
\newcounter{arcX}
\newcounter{arcY}
\newcounter{arcX1}
\newcounter{arcY1}
\newcounter{arcY2}
\newcommand{\arcLR}[4]{%
\setcounter{arcX}{#1}
\setcounter{arcY}{#2}
\setcounter{arcX1}{#1}
\setcounter{arcY1}{#2}
\addtocounter{arcX1}{#3}
\addtocounter{arcY1}{#4}
\setcounter{arcY2}{\value{arcY1}}
\addtocounter{arcY2}{#4}
\qbezier(\value{arcX},\value{arcY})(\value{arcX1},\value{arcY})(\value{arcX1},\value{arcY1})
\qbezier(\value{arcX},\value{arcY2})(\value{arcX1},\value{arcY2})(\value{arcX1},\value{arcY1})
}%
\newcommand{\idline}[3]{\put(#1,#2){\line(1,0){#3}}}
\newcommand{\idarrow}[3]{\put(#1,#2){\vector(1,0){#3}}}

\newcommand{\rvec}[2]{}
\newcommand{\myline}[4]{\qbezier(#1,#2)(#1,#2)(#3,#4)}

\newcommand{\lamnode}{\color[rgb]{1,0,0}\makebox(0,0){\huge$\bullet$}\color[rgb]{0,0,0}\makebox(0,0){\tiny$\lambda$}}
\newcommand{\appnode}{\color[rgb]{0,.6,.6}\makebox(0,0){\huge$\bullet$}\color[rgb]{0,0,0}\makebox(0,0){\tiny$@$}}
\newcommand{\tinylamnode}{\color[rgb]{1,0,0}\makebox(0,0){$\bullet$}}
\newcommand{\tinyappnode}{\color[rgb]{0,.6,.6}\makebox(0,0){$\bullet$}}